\documentclass[12pt]{article}

\usepackage[english]{babel} 
\usepackage[protrusion=true,expansion=true]{microtype} 
\usepackage{setspace}
\usepackage{indentfirst}

\usepackage{amsmath,amsfonts,amsthm}
\usepackage{amssymb}
\usepackage{mathtools}
\mathtoolsset{showonlyrefs}
\usepackage{bigints}
\usepackage{yhmath}

\usepackage{graphicx,psfrag,epsf}
\usepackage{tikz}
\usetikzlibrary{shapes,decorations,arrows,calc,arrows.meta,fit,positioning}

\usepackage{array}
\usepackage{booktabs} 
\usepackage{tabularx}
\usepackage{multirow}
\usepackage{diagbox}
\usepackage[export]{adjustbox}
\usepackage{footnote}
\usepackage{placeins}
\usepackage[font = small,labelfont=bf,textfont=it]{caption} 
\usepackage[font = small,labelfont=bf,textfont=it]{subcaption}

\usepackage{enumerate}
\usepackage{enumitem}
\usepackage{algorithm}
\usepackage{algpseudocode}

\usepackage{apacite}
\usepackage{natbib}
\usepackage{hyperref}
\usepackage{url}
\usepackage[toc,page]{appendix}
\usepackage{textcomp}
\usepackage{chngcntr}
\usepackage{comment}
\usepackage{pdfprivacy}
\usepackage{soul}
\usepackage[table]{xcolor}
\usepackage{etoolbox}
\usepackage{authblk}
\usepackage{makecell}

\newlist{steps}{enumerate}{1}
\setlist[steps, 1]{label = Step \arabic*:}
\tikzstyle{process} = [rectangle, rounded corners, minimum width=2cm, minimum height=1cm, text width=2cm, text centered, draw=black, fill=white]
\tikzstyle{arrow} = [thick,->,>=stealth]
\tikzset{
    -Latex,auto,node distance =1 cm and 1 cm,semithick,
    state/.style ={ellipse, draw, minimum width = 0.7 cm},
    point/.style = {circle, draw, inner sep=0.04cm,fill,node contents={}},
    bidirected/.style={Latex-Latex,dashed},
    el/.style = {inner sep=2pt, align=left, sloped}
}
\tikzstyle{decision} = [diamond, 
minimum width=2cm, 
minimum height=1cm, 
text centered, 
draw=black,
fill=white]

\algblock{ParFor}{EndParFor}
\algnewcommand\algorithmicparfor{\textbf{for}}
\algnewcommand\algorithmicpardo{\textbf{do\ parallel}}
\algnewcommand\algorithmicendparfor{\textbf{end\ parallel\ for}}
\algrenewtext{ParFor}[1]{\algorithmicparfor\ #1\ \algorithmicpardo}
\algrenewtext{EndParFor}{\algorithmicendparfor}

\makeatletter
\def\BState{\State\hskip-\ALG@thistlm}

\newcommand{\distas}[1]{\mathbin{\overset{#1}{\kern\z@\sim}}}%

\newsavebox{\mybox}\newsavebox{\mysim}
\newcommand{\distras}[1]{%
  \savebox{\mybox}{\hbox{\kern3pt$\scriptstyle#1$\kern3pt}}%
  \savebox{\mysim}{\hbox{$\sim$}}%
  \mathbin{\overset{#1}{\kern\z@\resizebox{\wd\mybox}{\ht\mysim}{$\sim$}}}%
}
\newtheorem{theorem}{Theorem}

\newtheorem{proposition}{Proposition}
\newtheorem{corollary}{Corollary}
\newtheorem{lemma}{Lemma}

\newcommand{\be}{\begin{equation}}
\newcommand{\ee}{\end{equation}}
    \newcommand{\bi}{\begin{itemize}}
\newcommand{\ei}{\end{itemize}}
\newcommand{\ben}{\begin{enumerate}}
\newcommand{\een}{\end{enumerate}}

\newcommand*{\rom}[1]{\expandafter\@slowromancap\romannumeral #1@}

\newcolumntype{K}[1]{\geq {\centering\arraybackslash}p{#1}}
\allowdisplaybreaks
\DeclareMathOperator*{\argmin}{\arg\!\min}
\makeatother

\let\oldbibliography\thebibliography
\renewcommand{\thebibliography}[1]{\oldbibliography{#1}
\setlength{\itemsep}{0pt}} 

\newcommand{\blind}{0}

\patchcmd{\footnotemark}{\stepcounter{footnote}}{\refstepcounter{footnote}}{}{}

\begin{document}

\def\spacingset#1{\renewcommand{\baselinestretch}%
{#1}\small\normalsize} \spacingset{1}

\if0\blind
{
  \title{\bf Rational Quadratic Kriging
  }
  \small
   \author{Difan Song and V. Roshan Joseph\thanks{Corresponding author: roshan@gatech.edu.} \vspace{3mm} \\
    H. Milton Stewart School of Industrial and Systems Engineering, \\
    Georgia Institute of Technology, Atlanta, GA, 30332}
  \date{}
  \maketitle
} \fi

\if1\blind
{
 \title{\bf }
  \bigskip
  \bigskip
  \bigskip
  \begin{center}
    {\LARGE\bf Rational Quadratic Kriging}
\end{center}

  \medskip
} \fi

\begin{abstract}
This article proposes rational quadratic kriging, a kriging estimator that addresses the limitations of existing approaches: the mean reversion issue of ordinary kriging and the diverging variance issue of rational kriging. Through the rational quadratic kernel, we establish a connection between rational quadratic kriging and kernel regression. The connection allows us to prove desirable theoretical properties of the new estimator and further propose a novel estimation procedure that utilizes factor weights from kernel regression. We also derive expressions to enable efficient leave-one-out estimation and conformal prediction. Extensive simulations across both deterministic functions and noisy datasets show significant improvements in both prediction and uncertainty quantification. A fast implementation is available in the open-source R package \texttt{RQK}.
\end{abstract}

\noindent
{\it Keywords: Kernel regression; Ordinary kriging; Rational kriging; Surrogate modeling; Uncertainty quantification.}
\vfill

\newpage
\spacingset{1.55} 

\section{Introduction} \label{sec:intro}

Kriging is a widely used nonparametric technique for fitting scattered data while quantifying uncertainty. Originating in geostatistics \citep{Matheron_1963}, kriging/Gaussian process has become prominent in fields including spatial statistics \citep{cressie1993statistics}, computer experiments \citep{Santner_Williams_Notz_2018}, and machine learning \citep{Rasmussen_Williams_2005} for modeling data that are either deterministic or stochastic in nature. Some examples include modeling multiple types of forces acting on rockets \citep{Gramacy_Lee_2012}, optimizing 3D-printed aortic tissues \citep{Chen_2021}, and predicting damage of Aircraft-UAV collisions \citep{liu2023statistical}. 

The simplest form of kriging is \textit{ordinary kriging}, which is the best linear unbiased predictor under the assumption that the data-generating process is second-order stationary \citep{Santner_Williams_Notz_2018}. If the stochastic process has mean $\mu$, variance $\sigma^2$, and correlation kernel $R(\cdot, \cdot)$, then, given the data vectors $\{(\mathbf{x}_i, y_i)\}_{i = 1}^n$, $\mathbf{x}_i, \mathbf{x} \in \mathbb{R}^d$, the ordinary kriging predictor and predictive variance are given by: \begin{align}
    \hat{m}_{\mathrm{OK}}(\mathbf{x}) &= \hat \mu + \mathbf{r}(\mathbf{x})^\top \mathbf{R}^{-1} (\mathbf{y} - \hat \mu \mathbf{1}), \label{eqn:OK_mean} \\
    s^2_{\mathrm{OK}}(\mathbf{x}) &= \sigma^2 \left\{1 - \mathbf{r}(\mathbf{x})^\top \mathbf{R}^{-1}\mathbf{r}(\mathbf{x}) + \frac{[\mathbf{r}(\mathbf{x})^\top \mathbf{R}^{-1} \mathbf{1} - 1]^2}{\mathbf{1}^\top \mathbf{R}^{-1} \mathbf{1}} \right\} \label{eqn:OK_sd},
\end{align} where $\mathbf{1}$ is a vector of ones, $\mathbf{R}$ is the correlation matrix of the $n$ observations, $\mathbf{r}(\mathbf{x}) = [R(\mathbf{x}, \mathbf{x}_1), \dots, R(\mathbf{x}, \mathbf{x}_n)]^\top$, and $\hat \mu = \frac{\mathbf{1}^\top \mathbf{R}^{-1} \mathbf{y}}{\mathbf{1}^\top \mathbf{R}^{-1} \mathbf{1}}$ is the generalized least squares estimate of $\mu$.

Despite its popularity, ordinary kriging is known to suffer from mean reversion \citep{Joseph_2006,Zhang_Apley_2014}. Assume the correlation function \begin{align}
    R(\mathbf{x}_i, \mathbf{x}_j) = G\left(\frac{\|\mathbf{x}_i - \mathbf{x}_j\|}{\theta}\right),
\end{align} where $G(\cdot)$ is nonnegative and decreasing, $\|\cdot\|$ is a norm, and $\theta > 0$ is the length scale parameter. In a low-correlation regime, $\theta \to 0$, and the ordinary kriging predictions converge to $\mu$. We demonstrate this phenomenon with a one-dimensional toy example: \begin{align}
    y = 2\sin(0.9x) + 0.15(x - 5)^2 - 1 + 0.5 \mathbb{I}\{2.5 \leq x \leq 3.5\} \sin(20x),
\end{align} where $\mathbb{I}\{\cdot\}$ denotes the indicator function. Consider panel (a) of Figure \ref{fig:motivating_example}, where we observe 10 samples in the high-frequency interval, and five other samples in the smooth region. The ordinary kriging predictions shown in the dashed line are significantly ``pulled'' towards the mean, and the plug-in prediction intervals $[\hat m_{\mathrm{OK}}(\mathbf{x}) -1.96 s_{\mathrm{OK}}(\mathbf{x}), \hat m_{\mathrm{OK}}(\mathbf{x}) + 1.96 s_{\mathrm{OK}}(\mathbf{x})]$ become large outside the high-frequency interval.

\citet{Joseph_2025} proposed an alternative predictor named \textit{rational kriging} predictor: \begin{align}
    \hat{m}_{\mathrm{RK}}(\mathbf{x}) &= \frac{\mathbf{r}(\mathbf{x})^\top \mathbf{R}^{-1} \mathrm{diag}(\mathbf{R} \mathbf{c})\mathbf{y}}{\mathbf{r}(\mathbf{x})^\top \mathbf{c}}, \label{eqn:RK_mean} \\
    s^2_{\mathrm{RK}}(\mathbf{x}) &= \frac{\nu^2}{\left\{\mathbf{r}(\mathbf{x})^\top \mathbf{c}\right\}^2}\left\{1 - \mathbf{r}(\mathbf{x})^\top \mathbf{R}^{-1}\mathbf{r}(\mathbf{x})\right\}, \label{eqn:RK_sd}
\end{align} where $\mathbf{c} \in \mathbb{R}_+^n$. Although \citet{Joseph_2025} proposed rational kriging with the motivation of improving parameter estimation, a specific choice of $\mathbf{c}$ ties the predictor to the existing \textit{limit kriging} predictor \citep{Joseph_2006}, which is a solution to mean reversion. Panel (b) of Figure \ref{fig:motivating_example} shows the predictions given by rational kriging on the toy example, where $\mathbf{c}$ is calculated by the Perron eigenvector of $\mathbf{R}$. We can see that the mean reversion issue is mitigated, as expected. Unfortunately, this is at the price of undesirable uncertainty quantification. For finite $\mathbf{c}$, the prediction intervals become unreasonably large due to the rational form. \citet{Joseph_2025} proposed a modification named \textit{generalized rational kriging} by replacing $\mathbf{r}(\mathbf{x})^\top \mathbf{c}$ with $c_0 + \mathbf{r}(\mathbf{x})^\top \mathbf{c}$. However, this was only a partial remedy, as in the low-correlation regime, the generalized predictor is both theoretically and practically similar to ordinary kriging, as shown in panel (c) of Figure \ref{fig:motivating_example}.

\begin{figure}[tb]
\centering

\begin{subfigure}{0.33\textwidth}
\centering
\includegraphics[width=\textwidth]{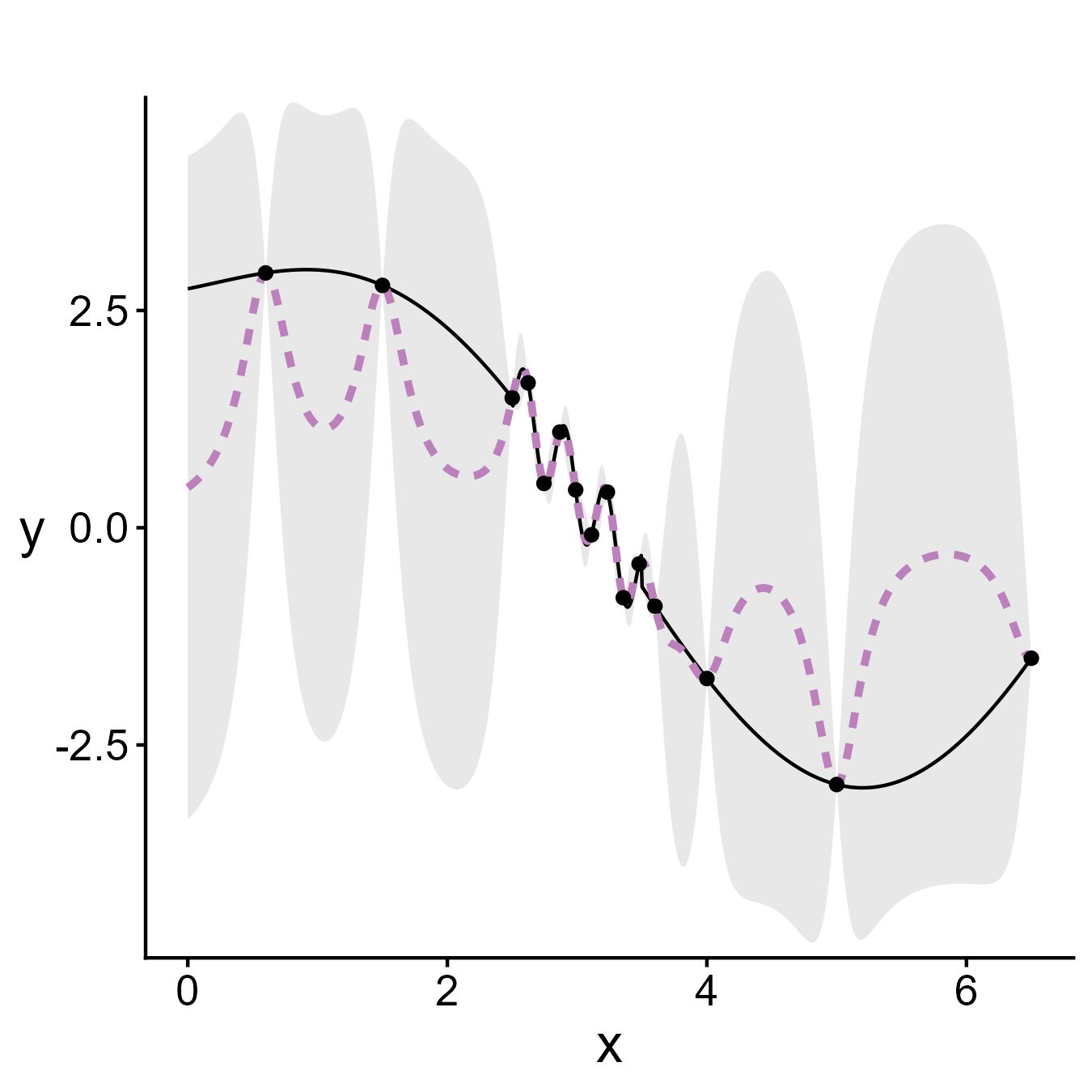}
\caption{Ordinary kriging.}
\end{subfigure}%
\begin{subfigure}{0.33\textwidth}
\centering
\includegraphics[width=\textwidth]{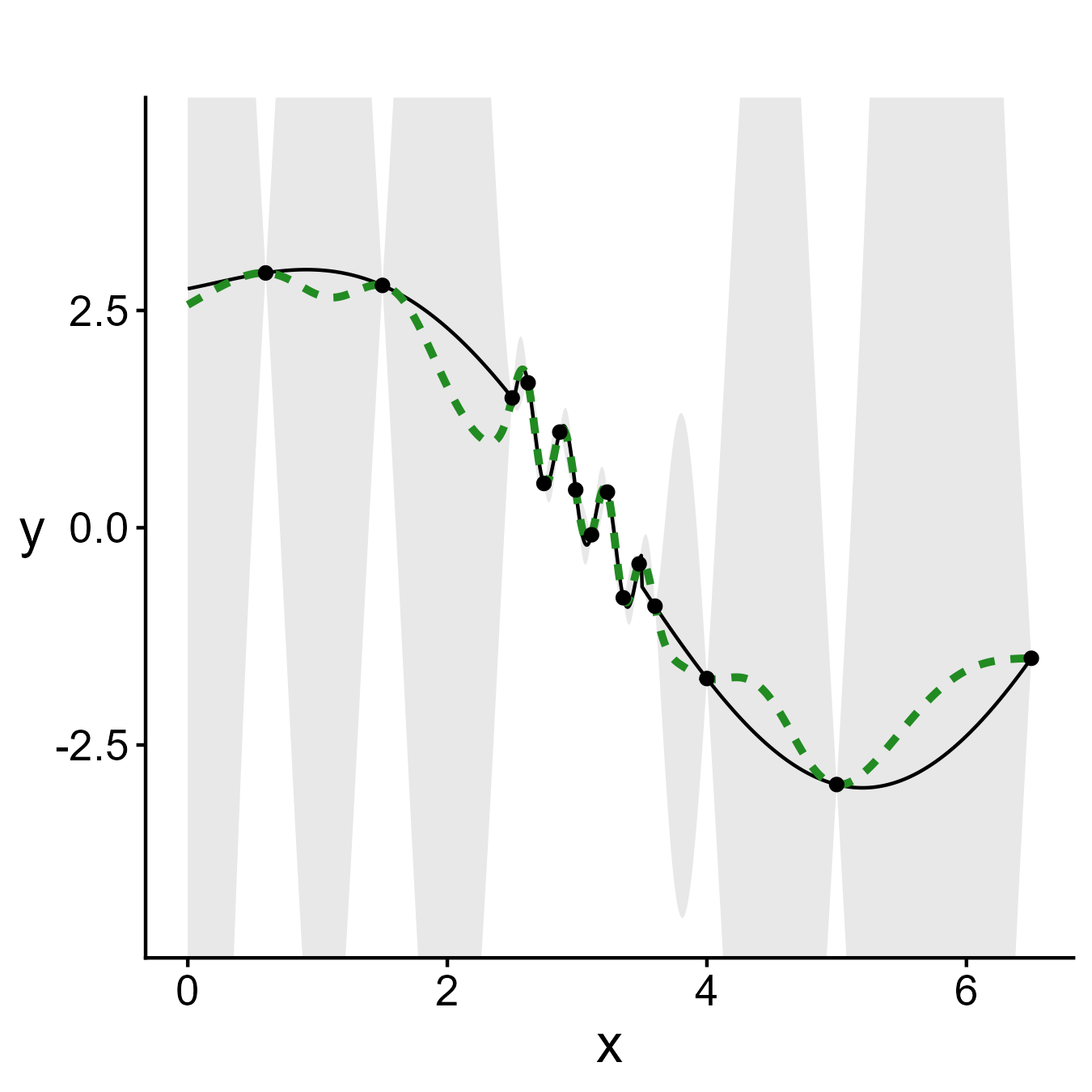}
\caption{Rational kriging.}
\end{subfigure}%
\begin{subfigure}{0.33\textwidth}
\centering
\includegraphics[width=\textwidth]{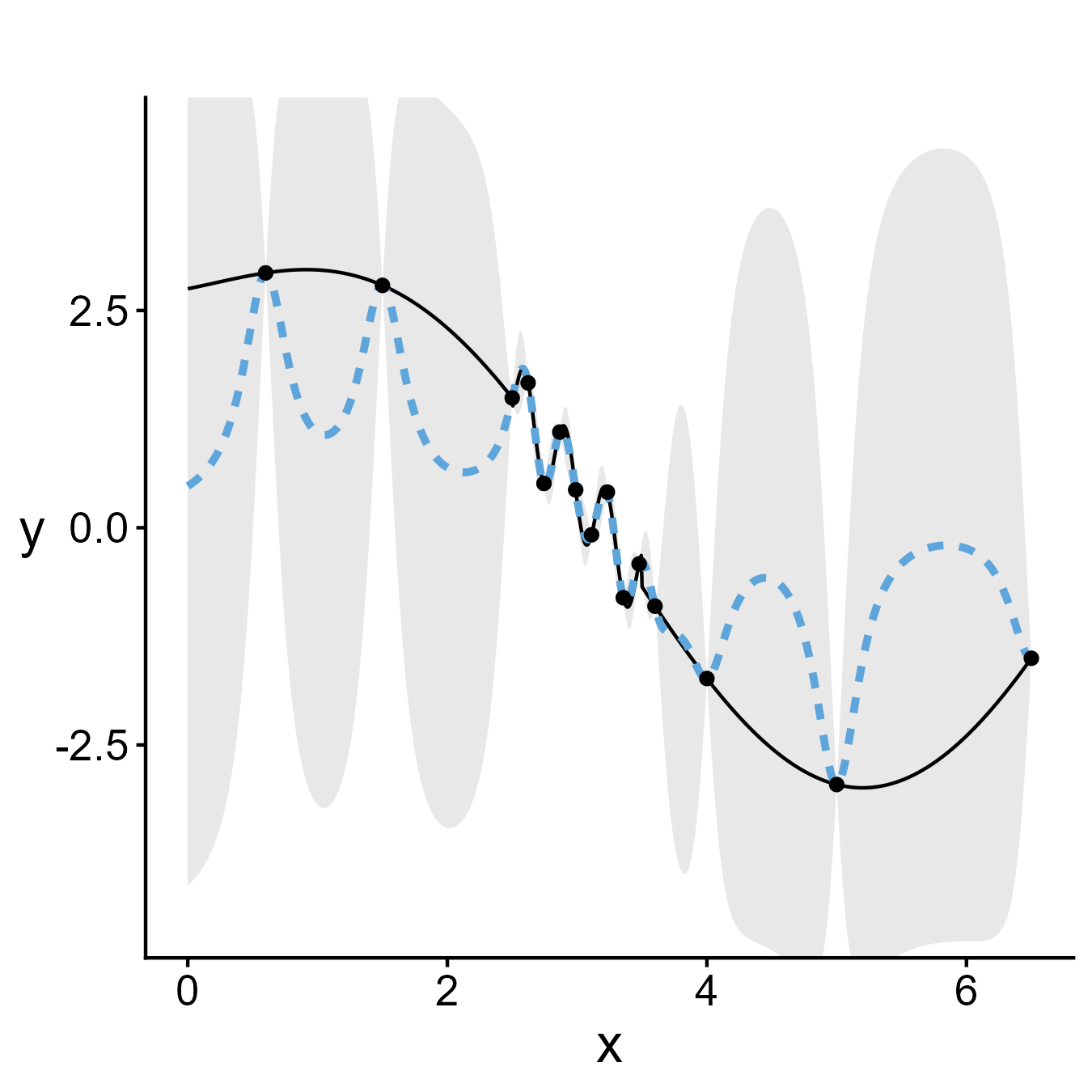}
\caption{Generalized rational kriging.}
\end{subfigure}%

\caption{Predictions and 95\% prediction intervals on the toy example for several existing kriging approaches. The fits are from the \texttt{rkriging} package \citep{rkriging}.}
\label{fig:motivating_example}

\end{figure}

The rational kriging predictor \eqref{eqn:RK_mean} is a special case of rational radial basis functions (RBFs), introduced by \citet{Jakobsson_Anderson_Edelvik_2009}. Later developments \citep{Sarra_Bai_2018, DeMarchi_2019} demonstrated strong performance of such predictors for functions with steep gradients and discontinuities. The theoretical results of \citet{Buhmann_2020} further showed that rational predictors achieve convergence rates similar to those of traditional RBFs, while increasing the stability of the estimates. However, rational RBFs do not provide uncertainty quantification. Therefore, we stick to the kriging framework, aiming to achieve the best of both worlds: accurate point predictions and reliable prediction intervals.

In this article, we show that this is indeed possible by an extension of rational kriging, enabled by the flexibility in choosing $\mathbf{c}$ in Equation \ref{eqn:RK_mean}. Specifically, we use the rational quadratic kernel and a $\mathbf{c}$ vector directly tied to the kernel evaluations. This specification allows us to connect with classic nonparametric estimators, including Shepard-type inverse-distance weighting (IDW) \citep{Shepard_1968} and Nadaraya-Watson-type kernel regression estimators \citep{Nadaraya_1964}. The relationships guide both theoretical analysis and practical implementation. We name our method rational quadratic kriging (RQK), and advocate its use over the existing ordinary and rational kriging implementations.

The remainder of the article is organized as follows. Section \ref{sec:RQK} introduces the RQK estimator. Section \ref{sec:convergence} provides a theoretical analysis of the convergence of the proposed estimator. Section \ref{sec:fitting} outlines practical fitting strategies that increase accuracy and stability, while also improving uncertainty quantification. Section \ref{sec:simulation} presents simulation results across a suite of simulations, including both deterministic functions and noisy datasets. Section \ref{sec:conclusion} concludes with some remarks and future directions.

\section{The rational quadratic kriging estimator} \label{sec:RQK}

We now specify the RQK model as an adaptation of the rational kriging model with particular choices of the kernel $R(\cdot, \cdot)$ and the vector $\mathbf{c}$.

First, we use the rational quadratic kernel \citep{Rasmussen_Williams_2005} as the correlation kernel: \begin{align} \label{eqn:RQ_kernel}
     R(\mathbf{x}_i, \mathbf{x}_j) = \left(1+\frac{\|\mathbf{x}_i - \mathbf{x}_j\|_{\mathbf{w}}^2}{\theta^2}\right)^{-p},
\end{align} where $p > 0$ is the shape parameter, and $\|\mathbf{x}\|_{\mathbf{w}}^2=\sum_{k = 1}^d w_k x_k^2$ with the weights satisfying $w_k > 0, \sum_{k=1}^d w_k = 1$. Note that the kernel is anisotropic due to the weighting, and the usual anisotropic length scales can be recovered by $\theta_k = \theta/\sqrt{w_k}$. For this kernel, we have the following result.

\begin{proposition}
\label{prop:RK_limit}
Assume $\mathbf{x} \notin \{\mathbf{x}_i\}_{i = 1}^n$. If $\mathbf{c} = \mathbf{1}$ or chosen as the adjusted Perron eigenvector of $\mathbf{R}$ given in \cite{Joseph_2025}, then under the rational quadratic kernel in \eqref{eqn:RQ_kernel}: \begin{align}
     \lim_{\theta \to 0} \hat m_{\mathrm{RK}}(\mathbf{x}) &= \frac{\sum_{i = 1}^n \|\mathbf{x} - \mathbf{x}_i\|_{\mathbf{w}}^{-2p} y_i}{\sum_{i = j}^n \|\mathbf{x} - \mathbf{x}_j\|_{\mathbf{w}}^{-2p}}, \\
     \lim_{\theta \to 0} s^2_{\mathrm{RK}}(\mathbf{x}) &= \infty.
\end{align}
\end{proposition}

The first part of the proposition is Theorem 3 of \citet{Joseph_2025}, which states that $\hat m_{\mathrm{RK}}(\mathbf{x})$ converges to the IDW estimator \citep{Shepard_1968} as the length scale $\theta \to 0$. However, the predictive variance diverges, which we observed for the toy example in Figure \ref{fig:motivating_example}. The choice of $\mathbf{c}$ plays a crucial role in the convergence of the predictor and the variance. Here, we investigate if another choice of $\mathbf{c}$ can make the prediction variance finite.

To do this, we first establish a connection between kernel regression and rational kriging. Consider a weighted version of the Nadaraya-Watson kernel regression estimator \citep{Nadaraya_1964}: \begin{align}
    \hat m_{\mathrm{KR}}(\mathbf{x}) = \frac{\sum_{i = 1}^n R(\mathbf{x}, \mathbf{x}_i) c_i y_i}{\sum_{j = 1}^n R(\mathbf{x}, \mathbf{x}_j) c_j}, \label{eqn:KR}
\end{align} 
where $c_i \geq 0$ are some weights. In contrast with kriging, kernel regression only smooths the data and does not interpolate. To achieve interpolation, \citet{Kang_Joseph_2016} described the following iterative procedure: first, fit $\mathbf{y}$ by kernel regression, and acquire estimators $\hat m_{\mathrm{KR}}^{(1)}(\mathbf{x})$ and the residuals $\mathbf{e}^{(1)}$. Then, fit $\mathbf{e}^{(1)}$ by kernel regression using the same kernel to acquire $\hat m_{\mathrm{KR}}^{(2)}(\mathbf{x})$ and the residuals $\mathbf{e}^{(2)}$. Repeating the procedure infinitely many times yields rational kriging, as stated in the following result.

\begin{proposition} If the correlation kernel is symmetric positive definite and $\mathbf{r}(\mathbf{x})^\top \mathbf{c} > 0$, 
\begin{align}
    \lim_{N \to \infty} \sum_{l = 1}^N m_{\mathrm{KR}}^{(l)}(\mathbf{x}) \to \hat m_{\mathrm{RK}}(\mathbf{x}).
\end{align}
\end{proposition} 
The proof follows directly from Theorem 1 of \citet{Kang_Joseph_2016}. Since rational kriging is the limit of the iterated kernel regressions, we can choose $\mathbf{c}$ so that the initial kernel regression in \eqref{eqn:KR} is good. It is well known that the Nadaraya-Watson estimator is the solution to the local weighted least squares problem: \begin{align}
    \hat m(\mathbf{x}) = \argmin_{m} \int_{\mathcal{X}} R(\mathbf{x}, \mathbf{t}) \{y(\mathbf{t}) - m\}^2 d\mathbf{t} = \frac{\int_{\mathcal{X}} R(\mathbf{x}, \mathbf{t}) y(\mathbf{t}) d\mathbf{t}}{\int_{\mathcal{X}} R(\mathbf{x}, \mathbf{t}) d\mathbf{t}}.
\end{align}
Assuming the density $f_{\mathbf{X}}(\cdot)$ exists, the estimator can be written as \begin{align}
    \hat m(\mathbf{x}) = \int_{\mathcal{X}} \frac{R(\mathbf{x}, \mathbf{t}) y(\mathbf{t})}{f_{\mathbf{X}}(\mathbf{t})} f_{\mathbf{X}}(\mathbf{t}) d\mathbf{t} \bigg/ \int_{\mathcal{X}} \frac{R(\mathbf{x}, \mathbf{t})}{f_{\mathbf{X}}(\mathbf{t})} f_{\mathbf{X}}(\mathbf{t}) d\mathbf{t},
\end{align} with sample analog 
\begin{align} \label{eqn:KR_compare}
    \hat m(\mathbf{x}) = \frac{1}{n} \sum_{i = 1}^n \frac{R(\mathbf{x}, \mathbf{x}_i) y_i}{f_{\mathbf{X}}(\mathbf{x}_i)} \bigg/ \frac{1}{n} \sum_{j = 1}^n \frac{R(\mathbf{x}, \mathbf{x}_j)}{f_{\mathbf{X}}(\mathbf{x}_j)}.
\end{align}
Thus, $c_i \propto 1/f_{\mathbf{X}}(\mathbf{x}_i)$, $i = 1, \ldots, n$, seems to be a good choice. Assuming $p > d/2$ and $w_l > 0$ for $l = 1, \dots, d$, the kernel density estimator of $f_{\mathbf{X}}$ using the rational quadratic kernel \eqref{eqn:RQ_kernel} is given by \begin{align}
    \hat{f}_{\mathbf{X}}(\mathbf{x}_i) = \frac{1}{n} \sum_{j = 1}^n  R(\mathbf{x}_i, \mathbf{x}_j)/Z,
\end{align}
where $Z=\Gamma(p - d/2) \pi^{d/2}  \theta^d/\{\Gamma(p) \prod_{k = 1}^d \sqrt{w_l}\}$ is the normalizing constant and $\Gamma(\cdot)$ is the Gamma function. This suggests choosing $c_i = \left\{\sum_{j = 1}^n R(\mathbf{x}_i, \mathbf{x}_j)\right\}^{-1}$. However, a slight modification of this choice using the leave-one-out density estimator is needed to achieve convergence of the variance, as stated in the following result.

\begin{proposition}
\label{prop:RQK_limit}
Assume $\mathbf{x} \notin \{\mathbf{x}_i\}_{i = 1}^n$ and let $a_{i, \mathrm{IDW}}(\mathbf{x}) = \|\mathbf{x} - \mathbf{x}_i\|_{\mathbf{w}}^{-2p}/\sum_{j\neq i}\|\mathbf{x}_i - \mathbf{x}_j\|_{\mathbf{w}}^{-2p}$. If 
\begin{align} \label{eqn:new_c}
     c_i = \left\{\sum_{j \neq i} R(\mathbf{x}_i, \mathbf{x}_j)\right\}^{-1}, \quad i = 1, \dots, n,
\end{align}
then under the rational quadratic kernel 
\begin{align}
     \lim_{\theta \to 0} \hat m_{\mathrm{RQK}}(\mathbf{x}) &= \frac{\sum_{i = 1}^n a_{i, \mathrm{IDW}}(\mathbf{x}) y_i}{\sum_{i = 1}^n a_{i, \mathrm{IDW}}(\mathbf{x})}, \\
     \lim_{\theta \to 0} s^2_{\mathrm{RQK}}(\mathbf{x}) &= \frac{\nu^2}{\left\{\sum_{i = 1}^n a_{i, \mathrm{IDW}}(\mathbf{x})\right\}^2},
\end{align}
where $\hat m_{\mathrm{RQK}}(\mathbf{x})$ and $s^2_{\mathrm{RQK}}(\mathbf{x})$ have the same expressions as Equations \eqref{eqn:RK_mean} and \eqref{eqn:RK_sd}, respectively.
\end{proposition}

\begin{figure}[tb]
\centering

\begin{subfigure}{0.33\textwidth}
\centering
\includegraphics[width=\textwidth]{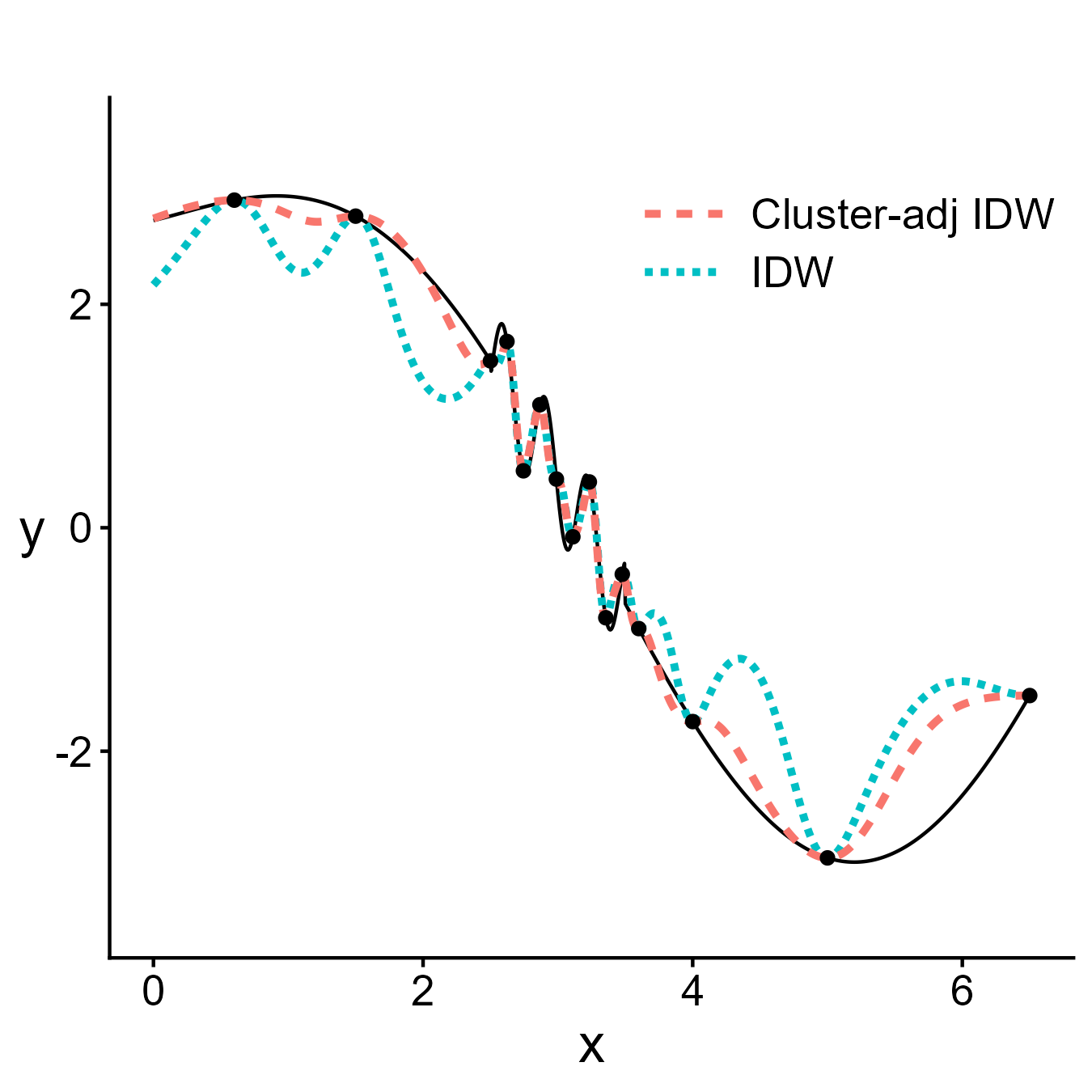}
\caption{IDW estimators.}
\end{subfigure}%
\begin{subfigure}{0.33\textwidth}
\centering
\includegraphics[width=\textwidth]{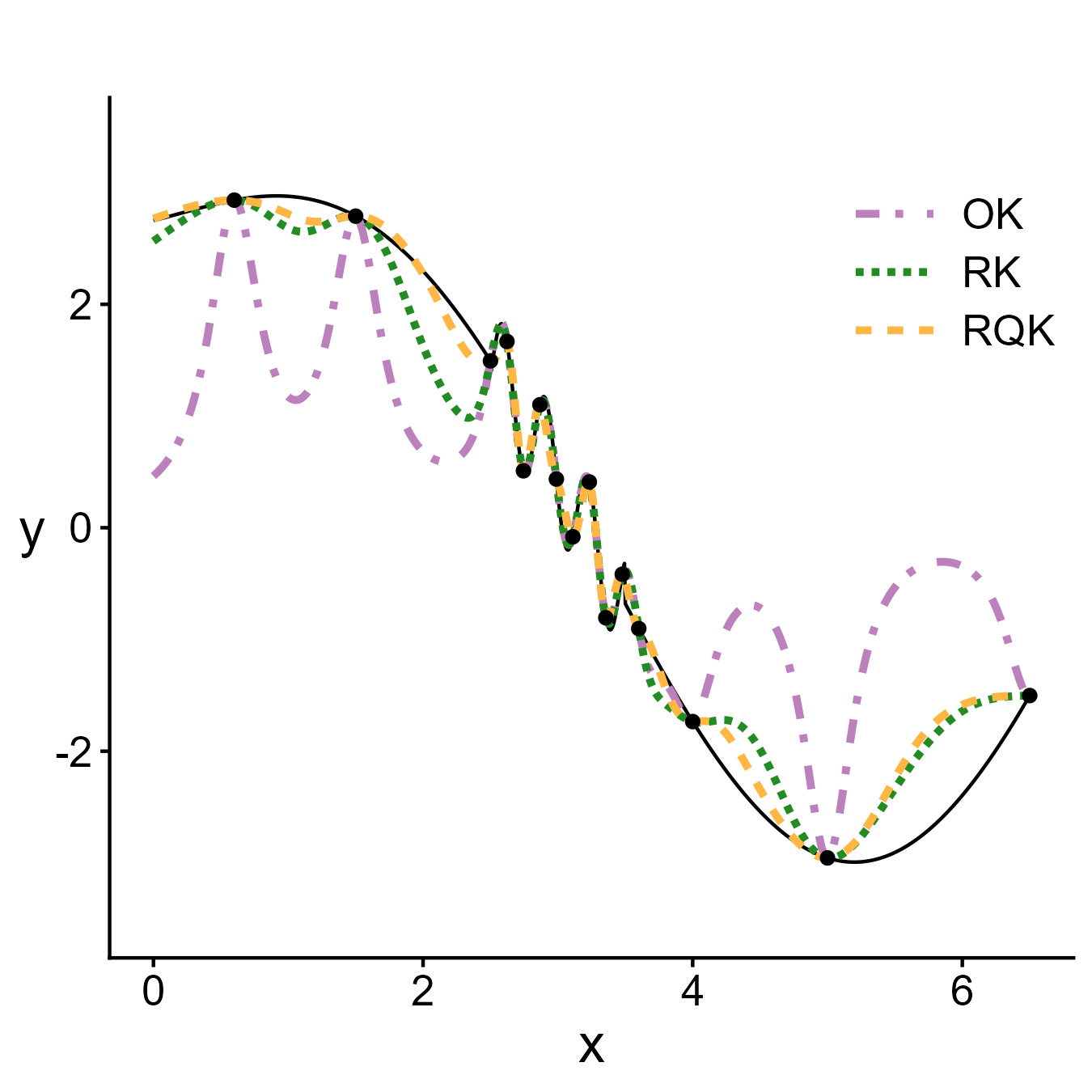}
\caption{Kriging predictors.}
\end{subfigure}%
\begin{subfigure}{0.33\textwidth}
\centering
\includegraphics[width=\textwidth]{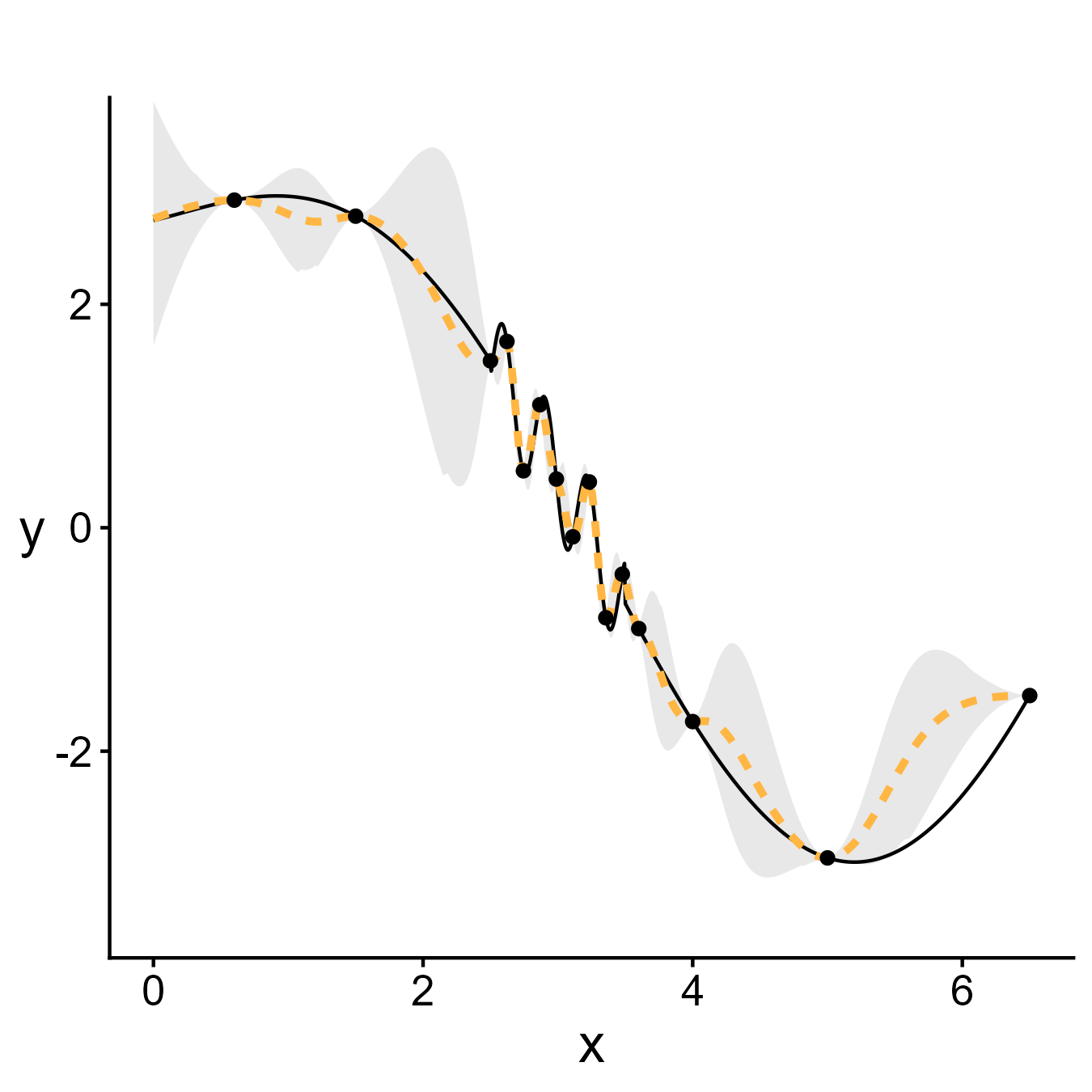}
\caption{RQK.}
\end{subfigure}%

\caption{Predictions for IDW and kriging estimators on the toy example. The 95\% confidence interval for RQK is also provided.}
\label{fig:motivating_example_RQK}

\end{figure}

The first part of Proposition \ref{prop:RQK_limit} shows that RQK also reduces to an IDW-type estimator in the limit. In Section 3, we will investigate the conditions under which the estimator becomes ``local enough'' to overcome the mean reversion issue. Furthermore, instead of the vanilla IDW estimator with weights $\|\mathbf{x} - \mathbf{x}_i\|_{\mathbf{w}}^{-2p}$, the limiting estimator has an extra denominator $\sum_{j \neq i} \|\mathbf{x}_i - \mathbf{x}_j\|_{\mathbf{w}}^{-2p}$. This term downweights observations in dense local clusters, thereby reducing their impact on out-of-sample estimates. Therefore, we expect this \textit{cluster-adjusted} IDW estimator to further mitigate mean reversion caused by non-uniform sampling, which is exactly what panel (a) of Figure \ref{fig:motivating_example_RQK} shows for the toy example introduced in Section \ref{sec:intro}. The pattern also extends to the comparison between RQK and existing kriging predictors. Panel (b) of Figure \ref{fig:motivating_example_RQK}, where ordinary and rational kriging are the same as Figure \ref{fig:motivating_example}, demonstrates that RQK overcomes mean reversion and accurately approximates the smooth portion of the function.

The second part of the proposition indicates that when we fix $\nu^2$, the predictive variance becomes $\mathcal{O}(1)$ in the limiting case, providing a sharp contrast with Proposition \ref{prop:RK_limit}. In RQK, the same kernel that drives $\mathbf{r}(\mathbf{x})$ also scales $\mathbf{c}$, making the denominator of the rational form $\mathbf{r}(\mathbf{x})^\top \mathbf{c}$ stable. This stability is crucial to improving the uncertainty quantification of our model. Panel (c) of Figure \ref{fig:motivating_example_RQK} provides the RQK prediction intervals, which are much shorter than those in Figure \ref{fig:motivating_example} and reasonably reflect the prediction uncertainty. In Section \ref{subsec:conformal}, we will discuss more about calibrating the prediction intervals.

\section{Convergence of the RQK predictor} \label{sec:convergence}

In this section, we investigate the convergence of the RQK predictor as $n \to \infty$. Following Proposition \ref{prop:RQK_limit}, we continue to work in the low-correlation regime where $\theta \to 0$. Throughout the remainder of this article, we assume that $\mathbf{x} \in \mathcal{X} = [0, 1]^d$, where $d$ is the dimension. 


Previously, we established the connection of RQK to the kernel regression estimator \eqref{eqn:KR}. We next show that, in a low-correlation regime, the two estimators converge, without the need for the iterative procedure in \citet{Kang_Joseph_2016}. For notational simplicity, let $\mathbf{a}_{\mathrm{RQK}}(\mathbf{x})$ and $\mathbf{a}_{\mathrm{KR}}(\mathbf{x})$ denote the weight vectors of $\hat m_{\mathrm{RQK}}(\mathbf{x})$ and $\hat {m}_{\mathrm{KR}}(\mathbf{x})$, respectively. In other words, \begin{align}
    \mathbf{a}_{\mathrm{RQK}}(\mathbf{x}) = \frac{\mathrm{diag}(\mathbf{R} \mathbf{c}) \mathbf{R}^{-1} \mathbf{r}(\mathbf{x})}{\mathbf{r}(\mathbf{x})^\top \mathbf{c}}, \quad a_{i, \mathrm{KR}}(\mathbf{x}) = \frac{R(\mathbf{x}, \mathbf{x}_i) c_i}{\sum_{j = 1}^n R(\mathbf{x}, \mathbf{x}_j) c_j}.
\end{align}

\begin{lemma}[Proximity of RQK to kernel regression] \label{lemma:RQK_KR}
Fix $n$ and $\theta > 0$, define \begin{align}
    A_n(\theta) := \frac{\max_{1 \leq i \leq n} c_i}{\min_{1 \leq i \leq n} c_i}, \quad \zeta_n(\theta) := \max_{1 \leq i \leq n} \sum_{j \neq i} R(\mathbf{x}_i, \mathbf{x}_j) = \frac{1}{\min_{1 \leq i \leq n} c_i}
\end{align}

Suppose (i) $A_n(\theta) < \infty$, and (ii) $\zeta_n(\theta) < 1$. Then, there exists a constant $C = C(A_n(\theta), \zeta_n(\theta))$ such that, for every $\mathbf{x}\in \mathcal{X}$, \begin{align}
    \|\mathbf{a}_{\mathrm{RQK}}(\mathbf{x}) - \mathbf{a}_{\mathrm{KR}}(\mathbf{x})\|_1 \leq C \zeta_n(\theta).
\end{align}
Consequently, \begin{align}
    \sup_{x \in \mathcal{X}} |\hat m_{\mathrm{RQK}}(\mathbf{x}) - \hat m_{\mathrm{KR}}(\mathbf{x})| \leq C \zeta_n(\theta) \|\mathbf{y}\|_{\infty}.
\end{align}

Moreover, if for some $\bar \theta > 0$, (i) $\sup_{0 < \theta \leq \bar \theta} A_n(\theta) \leq \bar A < \infty$, and (ii) $\zeta_n(\bar \theta) < 1$, there exists a constant $C = C(\bar A, \zeta_n(\bar \theta))$ such that \begin{align}
    \sup_{0 < \theta \leq \bar \theta} \sup_{x \in \mathcal{X}} |\hat m_{\mathrm{RQK}}(\mathbf{x}) - \hat m_{\mathrm{KR}}(\mathbf{x})| \leq C \zeta_n(\bar \theta) \|\mathbf{y}\|_{\infty}. \label{eqn:lemma_uniform_bound}
\end{align}
\end{lemma} The proof of the lemma and all remaining theoretical results are presented in the Appendix.

The lemma shows that the proximity is characterized by \textit{diagonal dominance}. To make the gap vanish, we will require $\zeta_n(\theta_n) \to 0$. It is straightforward to show that \begin{align}
    \zeta_n(\theta_n) \to 0 \Rightarrow n\theta_n^d \to 0.
\end{align} Classic nonparametric regression results, for instance \citet{Devroye_1981}, are often built on the condition $\theta_n \to 0$ and $n \theta_n^d \to \infty$. In contrast, equation \eqref{eqn:lemma_uniform_bound} covers the case when $\theta$ is arbitrarily close to zero, which controls the worst-case of RQK when points are nearly independent of each other. In Section \ref{subsec:initial}, we will use this result to motivate the use of kernel regression to acquire reasonable hyperparameters for RQK.

The next result shows that RQK is weakly consistent, with minimal assumptions on the data-generating process.

\begin{theorem}[Weak consistency]
\label{theorem:Weak_consistency}
Suppose $\{(\mathbf{X}_i, Y_i)\}_{i \geq 1}$ are i.i.d. random vectors, where $\mathbf{X}$ has a bounded density $f_{\mathbf{X}}$ on $\mathcal{X}$, and $Y$ is also bounded. Let \begin{align}
    m(\mathbf{x}) = \mathbb{E}[Y | \mathbf{X} = \mathbf{x}].
\end{align} Suppose the rational quadratic kernel \eqref{eqn:RQ_kernel} is used with $p = d/2$. Under the assumptions (i) $\theta_n \to 0$ and $\zeta_n(\theta_n) \overset{p}{\to} 0$; (ii) $\frac{\max_{1 \leq i \leq n} c_i}{\min_{1 \leq i \leq n} c_i} = \mathcal{O}_p(1)$,
\begin{enumerate}
    \item at almost all $\mathbf{x} \in \mathcal{X}$ with $f_{\mathbf{X}}(\mathbf{x}) > 0$, $\hat m_{\mathrm{RQK}}(\mathbf{x}) \overset{p}{\to} m(\mathbf{x})$;

    \item $\displaystyle \int_{\mathcal{X}} \left|\hat m_{\mathrm{RQK}}(\mathbf{x}) - m(\mathbf{x})\right| f_{\mathbf{X}}(\mathbf{x}) d\mathbf{x} \overset{p}{\to} 0$.
\end{enumerate}
\end{theorem} The proof uses Lemma \ref{lemma:RQK_KR}, and an additional lemma in the Appendix that establishes the connection between kernel regression and IDW. Then, we can invoke the results of \citet{Devroye_Gyorfi_Krzyzak_1998} on the IDW estimator for convergence. 

The data-generating process of Theorem \ref{theorem:Weak_consistency} can be either deterministic or stochastic, and only the boundedness of $Y$ is assumed. In this case, only $p = d/2$ yields a consistent predictor. According to the discussions of \citet{Devroye_Gyorfi_Krzyzak_1998}: \begin{itemize}
    \item [--] When $p < d/2$, the contributions of far-away data pairs are too large. This will also cause mean reversion, though the mechanism is different from $\theta \to 0$.

    \item [--] When $p > d/2$, the nearest neighbors dominate the estimates, and the variance caused by noise in those observations will not go to zero.
\end{itemize}
This discussion motivates us to further investigate the important case where the data-generating process is deterministic and continuous, which is covered in the next theorem.

\begin{theorem}[Strong consistency for continuous functions]
\label{theorem:noiseless_continuity}
Suppose $\{\mathbf{X}_i\}_{i \geq 1}$ are i.i.d. random vectors, where $\mathbf{X}$ has a bounded density $f_{\mathbf{X}}$ on $\mathcal{X}$. Also, $Y = m(\mathbf{X})$ for some continuous and bounded function $m: \mathcal{X} \to \mathbb{R}$. Suppose the rational quadratic kernel \eqref{eqn:RQ_kernel} is used with $p \geq d/2$. Further assume (i) $\theta_n \to 0$ and $\zeta_n(\theta_n) \overset{p}{\to} 0$; (ii) $\frac{\max_{1 \leq i \leq n} c_i}{\min_{1 \leq i \leq n} c_i} = \mathcal{O}_p(1)$, then
\begin{enumerate}
    \item at almost all $\mathbf{x} \in \mathcal{X}$ with $f_{\mathbf{X}}(\mathbf{x}) > 0$, $\hat m_{\mathrm{RQK}}(\mathbf{x}) \overset{p}{\to} m(\mathbf{x})$;

    \item $\displaystyle \int_{\mathcal{X}} \left|\hat m_{\mathrm{RQK}}(\mathbf{x}) - m(\mathbf{x})\right| f_{\mathbf{X}}(\mathbf{x}) d\mathbf{x}\overset{p}{\to} 0$.
\end{enumerate}

If, in addition (i') $\zeta_n(\theta_n) \to 0$ almost surely; (ii') $\sup_{n \geq 1} \frac{\max_{1 \leq i \leq n} c_i}{\min_{1 \leq i \leq n} c_i} < \infty$ almost surely; (iii) $\sum_{i = 1}^\infty n\theta_n^d < \infty$, then both results can be strengthened to almost sure convergence.
\end{theorem}

Theorem \ref{theorem:noiseless_continuity} shows that a continuous and deterministic function indeed allows $p \geq d/2$. Furthermore, the estimator converges almost surely under stronger conditions on $\mathbf{c}$ and $\theta$. The conditions on $\theta$ cover cases where $\theta$ becomes arbitrarily close to zero, reassuring that RQK will not suffer from mean reversion.

As a next step, we investigate the use of quasi-random designs, which are widely used to explore the domain of black-box functions. For a design $\{\mathbf{x}_i\}_{i = 1}^n$, define \citep{fasshauer2007meshfree}
\begin{align}
    \text{(Fill distance)} \quad h_n &:= \sup_{\mathbf{x} \in \mathcal{X}} \min_{1 \leq i \leq n} \|\mathbf{x} - \mathbf{x}_i\|_{\mathbf{w}}, \\
    \text{(Separation distance)} \quad q_n &:= \frac{1}{2} \min_{i \neq j} \|\mathbf{x}_i - \mathbf{x}_j\|_{\mathbf{w}}, \\
    \text{(Mesh ratio)} \quad \rho_n &:= h_n/q_n.
\end{align} In this work, we focus on quasi-random designs that are deterministic and upper bound the mesh ratio: \begin{align}
    \sup_n \rho_n < \infty. 
\end{align} The results can be extended to stochastic quasi-random designs by assuming the condition holds with probability one. Under the deterministic assumption, the next Theorem gives the deterministic convergence rates for RQK.

\begin{theorem}[Deterministic bound under quasi-random designs]
\label{theorem:quasirandom}
Suppose $\{\mathbf{x}_i\}_{i \geq 1}$ is a quasi-random design on $\mathcal{X}$ satisfying $\rho = \sup_n \rho_n < \infty$. Also, $Y = m(\mathbf{X})$ for some continuous and bounded function $m: \mathcal{X} \to \mathbb{R}$. Let \begin{align}
    \bar \theta_n = c q_n^{1 + \delta} \text{ for some } c > 0, \delta > 0,
\end{align} then there exists a constant $C = C(d, p, \rho, \mathbf{w})$ such that for every $0 < \theta \leq \bar \theta_n$, \begin{align}
    \sup_{\mathbf{x}\in\mathcal{X}} |\hat m_{\mathrm{RQK}}(\mathbf{x}) -m(\mathbf{x})| \leq C \zeta_n(\bar\theta_n) \|m\|_\infty + C \inf_{r\in[h_n, D]} \left\{
    \varphi_m(r) + 2\|m\|_\infty \psi_{p,n}(r)\right\},
\end{align} where \begin{align}
    \varphi_m(r) := \sup_{\|\mathbf{u} - \mathbf{v}\|_{\mathbf{w}} \leq r} |m(\mathbf{u}) - m(\mathbf{v})|
\end{align} is the modulus of continuity, and \begin{align}
    \psi_{p, n}(r) := \begin{cases}
    \left(\dfrac{h_n}{r}\right)^{2p - d}, & p > d/2; \\
    \dfrac{\log(D/r)}{\log(D/h_n)}, & p = d/2,
\end{cases}
\end{align} with $D := \sup_{\mathbf{u}, \mathbf{v} \in \mathcal{X}} \|\mathbf{u} - \mathbf{v}\|_{\mathbf{w}}$ the $w$-weighted diameter.
\end{theorem}

The bound consists of two parts. $C \zeta_n(\bar \theta_n) \|m\|_{\infty}$ is the gap between RQK and kernel regression, which we derived in Lemma \ref{lemma:RQK_KR}. The second part results from dividing the kernel regression estimator into the influences of nearby and far-away points. 

Because the designs and the function are fully deterministic, we can claim the bounds for any $0 < \theta < \bar \theta_n$ as long as $\bar \theta_n$ diminishes quicker than the separation and fill distances. This is in contrast with Theorems \ref{theorem:Weak_consistency} and \ref{theorem:noiseless_continuity}, where we made assumptions on $\mathbf{c}$ and $\zeta_n(\bar \theta_n)$: the properties of quasi-random designs allow us to show these conditions hold (see Lemmas \ref{lemma:c_quasirandom} and \ref{lemma:zeta_quasirandom} in the Appendix). More importantly, a similar uniform convergence result cannot exist for ordinary kriging, because as $\theta$ becomes arbitrarily close to $0$, mean reversion will result in a non-trivial error for non-constant functions on $\mathcal{X}$.

Finally, we look at the special class of $\alpha$-H\"older functions, where the modulus of continuity in \ref{theorem:noiseless_continuity} is made explicit. This class covers a wide range of functions relevant to kriging on a compact domain, for instance, functions residing in the reproducing kernel Hilbert space of the rational quadratic kernel \eqref{eqn:RQ_kernel}, the squared-exponential kernel, and the Mat\'ern kernel with smoothness parameter $\nu > \alpha$. 

\begin{corollary}[Bounds for H\"older functions]
\label{cor:Holder}
Under the conditions of Theorem \ref{theorem:quasirandom}, assume in addition, $m$ is $\alpha$-H\"older on $\mathcal{X}$ with $\alpha \in (0, 1]$, i.e. \begin{align}
    |m(\mathbf{u}) - m(\mathbf{v})| \leq L\|\mathbf{u}-\mathbf{v}\|_{\mathbf{w}}^\alpha \text{ for any } \mathbf{u}, \mathbf{v} \in \mathcal{X}.
\end{align}

If \begin{align}
    \bar\theta_n = c b_n^{1 + \delta} \text{ for some } c > 0, \delta > 0 \text{ and }\delta \geq \frac{1}{2p} \min\{\alpha, 2p - d\},
\end{align} then there exists a constant $C = C(d, p, \rho, \mathbf{w}, \alpha, L)$ such that \begin{align}
    \sup_{0 < \theta \leq \bar \theta_n} \sup_{\mathbf{x} \in \mathcal{X}} |\hat m_{\mathrm{RQK}}(\mathbf{x}) - m(\mathbf{x})| \leq C \phi_n,
\end{align} where \begin{align}
    \phi_n =\begin{cases}
    h_n^\alpha, & p > \dfrac{d + \alpha}{2}, \\
    h_n^{2p - d}, & \dfrac{d}{2} < p \leq \dfrac{d + \alpha}{2}, \\
    \dfrac{1}{\log(D/a_n)},
    & p=\dfrac d2.
    \end{cases}
\end{align}
\end{corollary}

In the best case $p > (d + \alpha/2)$, the error is upper-bounded by $\mathcal{O}(h_n^{\alpha})$. Classic results in data approximation \citep{Wendland_2004,Novak_2006} show that this rate is minimax optimal. Thus, we can conclude that RQK is theoretically reliable in the small-correlation regime. The next section will consider estimation strategies to optimize practical performance.

\section{Practical fitting strategies} \label{sec:fitting}

In the previous sections, we focused on the theoretical properties of the interpolating RQK estimator, although Theorem \ref{theorem:Weak_consistency} covers both deterministic and noisy data-generating processes. Nevertheless, for fitting noisy data, it is often beneficial to add a nugget term \begin{align}
    R(\mathbf{x}_i, \mathbf{x}_j) = \left(1 + \frac{\|\mathbf{x}_i - \mathbf{x}_j\|_{\mathbf{w}}^2}{\theta^2}\right)^{-p} + \gamma\mathbb{I}\{\mathbf{x}_i = \mathbf{x}_j\},
\end{align} where $\gamma \geq 0$ and $\mathbb{I}\{\cdot\}$ is the indicator function. For deterministic data, we set $\gamma$ to a small value ($10^{-8}$) for numerical stability \citep{Gramacy_Lee_2012}, whereas for noisy data, we treat $\gamma$ as an unknown hyperparameter to be estimated from the data.

For the remainder of this section, we describe the practical steps of hyperparameter estimation and uncertainty quantification. All strategies are implemented in the open-source R package \texttt{RQK}, available on \texttt{CRAN}\footnote{\url{https://cran.r-project.org/}}. 

\subsection{Initial weight fitting} \label{subsec:initial}

In Sections \ref{sec:RQK} and \ref{sec:convergence}, we established the connection between the RQK estimator and two other estimators: the kernel regression estimator and the IDW estimator. Beyond theoretical properties, the connections also motivate an initial fitting step for the anisotropic weights.

Here, we choose to work with the kernel regression estimator. Specifically, fix $\theta = \theta_0$ and $c_i = 1$ for all $i$ in equation \eqref{eqn:KR}, and evaluate the leave-one-out mean-squared error as a function of the weights $\mathbf{w}$: \begin{align}
    \mathrm{LOO}_{\mathrm{KR}}(\mathbf{w}) = \frac{1}{n} \sum_{i = 1}^n \left\{y_i - \frac{\sum_{j \neq i} R(\mathbf{x}_i, \mathbf{x}_j) y_j}{\sum_{l \neq i} R(\mathbf{x}_i, \mathbf{x}_l)}\right\}^2. \label{eqn:LOO_KR}
\end{align} Inspired by the convergence conditions in Section \ref{sec:convergence}, we choose $p = (d + 1)/2$ and $p = d/2$ as the default for deterministic and noisy data, respectively.

There are two benefits of using the loss \eqref{eqn:LOO_KR}. First, it is straightforward to derive the analytical gradient: \begin{align}
    \frac{\partial \mathrm{LOO}_{\mathrm{KR}}(\mathbf{w})}{\partial w_k} = \frac{2p}{n\theta_0^2} \sum_{i=1}^n \left\{\frac{y_i - \hat{y}_{\sim i, \mathrm{KR}}}{\sum_{l \neq i} R(\mathbf{x}_i, \mathbf{x}_l)} \sum_{j \neq i} R(\mathbf{x}_i, \mathbf{x}_j)^{\frac{p + 1}{p}} (x_{ik} - x_{jk})^2 (y_j - \hat{y}_{\sim i, \mathrm{KR}}) \right\}, \label{eqn:LOO_KR_grad}
\end{align} where $\hat{y}_{\sim i, \mathrm{KR}}$ is the leave-one-out prediction in \eqref{eqn:LOO_KR}. The computational complexity is governed by the calculation of $\mathbf{R}$, which is $\mathcal{O}(n^2)$. In contrast, both gradient-based and non-gradient optimization of the kriging predictor involve inverting the correlation matrix, which is inefficient with a complexity of $\mathcal{O}(n^3)$ and less stable. Second, it is straightforward to truncate the estimator to depend only on a few nearest neighbors, rather than all points other than $\mathbf{x}_i$. If we keep the number of neighbors $\mathcal{O}(\log n)$, the overall complexity can be further reduced to $\mathcal{O}(n \log n)$. Therefore, this initial estimation step is expected to provide a stable initial orientation for the weights $\mathbf{w}$, and does not add any computational burden.   

For practical implementation, we reparametrize $\mathbf{w}$ using the softmax function: \begin{align}
    w_k = \frac{\exp(u_k)}{\sum_{l = 1}^d \exp(u_l)}, \quad k = 1, \dots, d.
\end{align} Then, we can optimize over the unconstrained vector $\mathbf{u}$, where $\frac{\partial \mathrm{LOO}_{\mathrm{KR}}(\mathbf{u})}{\partial u_k}$ is easily obtained from equation \eqref{eqn:LOO_KR_grad} and the chain rule. The number of neighbors is taken to be $d + 1$ and $2d + 1$ for deterministic and noisy data, respectively. Starting with $u_k = 0$ for all $k$, we use the Method of Moving Asymptotes \citep{Svanberg_2002} implemented in the \texttt{Nlopt} library \citep{NLopt} for local gradient-based optimization. The minimizer $\hat{\mathbf{w}}_0$ will be used as the initial value in the next step in Section \ref{subsec:LOO}.

\subsection{Leave-one-out cross-validation} \label{subsec:LOO}

Proposition \ref{prop:RQK_limit} established that, for our specific choice of $\mathbf{c}$, the out-of-sample prediction uncertainty was finite as $\theta \to 0$. However, at the sampled locations, the likelihood under a Gaussian process assumption becomes ill-behaved since $R(\mathbf{x}_i, \mathbf{x}_i) = 1 + \gamma$ and $c_i \to \infty$. Thus, we cannot rely on a likelihood-based fit. We instead propose optimizing the hyperparameters $(\mathbf{w}, \theta, \gamma)$ using leave-one-out cross-validation.

Before we derive the leave-one-out formulas, we note that $c_i \to \infty$ introduces another source of instability in implementation. Therefore, we propose the following rescale: \begin{align}
    \tilde{\mathbf{r}}(\mathbf{x}) = \theta^{-2p} \mathbf{r}(\mathbf{x}), \quad \tilde{\mathbf{c}} = \theta^{2p} \mathbf{c}.
\end{align} Because \begin{align}
    \theta^{-2p} \left(1 + \frac{\|\mathbf{x}_i - \mathbf{x}_j\|_{\mathbf{w}}^2}{\theta^2}\right)^{-p} = \left(\theta^{-2p} + \|\mathbf{x}_i - \mathbf{x}_j\|_{\mathbf{w}}^2\right)^{-p},
\end{align} $\tilde{\mathbf{r}}(\mathbf{x})$ and $\tilde{\mathbf{c}}$ both become finite, which benefits the numerical stability of the calculations. Then, the point predictor and predictive variance become \begin{align}
    \hat{m}_{\text{RQK}}(\mathbf{x}) &= \frac{\tilde{\mathbf{r}}(\mathbf{x})^\top \mathbf{R}^{-1} \mathrm{diag}(\mathbf{R} \tilde{\mathbf{c}})\mathbf{y}}{\tilde{\mathbf{r}}(\mathbf{x})^\top \tilde{\mathbf{c}}}, \label{eqn:RQK_mean_rescale} \\
    s^2_{\text{RQK}}(\mathbf{x}) &= \frac{\nu^2}{\left\{\tilde{\mathbf{r}}(\mathbf{x})^\top \tilde{\mathbf{c}}\right\}^2}\left\{1 + \gamma - \theta^{4p} \tilde{\mathbf{r}}(\mathbf{x})^\top \mathbf{R}^{-1}\tilde{\mathbf{r}}(\mathbf{x})\right\} \label{eqn:RQK_sd_rescale}
\end{align} 

Here, we directly give the formulas for the leave-one-out point predictors and predictive variances. For a fixed $\mathbf{x}$ (can be either in-sample or out-of-sample), denote \begin{align}
    \mathbf{z} = \mathbf{R}^{-1} \tilde{\mathbf{r}}(\mathbf{x}), \quad \mathbf{\Lambda} = \mathbf{R}^{-1},
\end{align} then \begin{align}
    \hat m_{\sim i}(\mathbf{x}) &= \frac{K_1}{K_2}, \label{eqn:LOO_mean} \\
    s^2_{\sim i}(\mathbf{x}) &= \frac{\nu^2}{K_2^2}\left\{1 + \gamma - \theta^{4p} \tilde{\mathbf{r}}(\mathbf{x}) \mathbf{R}^{-1} \tilde{\mathbf{r}}(\mathbf{x}) + \frac{\theta^{4p} z_i^2}{\Lambda_{ii}}\right\} \label{eqn:LOO_sd}
\end{align} with \begin{align}
    K_1 &= \mathbf{z}^\top ([\mathbf{R} \tilde{\mathbf{c}}] \odot \mathbf{y}) - \frac{z_i [\mathbf{\Lambda} \mathrm{diag}(\mathbf{R} \tilde{\mathbf{c}}) \mathbf{y}]_i}{\Lambda_{ii}} - \tilde{c}_i \left(\sum_{j} z_j R_{ij} y_j - \frac{z_i \sum_{l} \Lambda_{il} R_{il} y_l}{\Lambda_{ii}}\right), \\
    K_2 &= \tilde{\mathbf{r}}(\mathbf{x})^{\top} \tilde{\mathbf{c}} - \theta^{-2p} R(\mathbf{x}, \mathbf{x}_i) \tilde{c}_i,
\end{align} where $[\cdot]_{i}$ denotes the $i$-th element of a vector, and $\odot$ denotes element-wise (Hadamard) product.

A key feature of the above formulas is that a single inversion of $\mathbf{R}$ is used towards the calculations of all $n$ leave-one-out predictions, which keeps the computational complexity at $\mathcal{O}(n^3)$. The key identity allowing the simplification is a direct consequence of the Sherman-Morrison formula: \begin{align}
    \mathbf{u}_{-i}^\top [\mathbf{R}_{-i, -i}]^{-1} \mathbf{v}_{-i} = \mathbf{u}^\top \mathbf{\Lambda} \mathbf{v} - \frac{[\mathbf{\Lambda} \mathbf{u}]_i [\mathbf{\Lambda} \mathbf{v}]_i}{\Lambda_{ii}},
\end{align} where $[\cdot]_{-i}$ denotes leaving out the $i$-th element of a vector, $\mathbf{R}_{-i, -i}$ denotes leaving out the $i$-th column and row of $\mathbf{R}$, and we defer the remaining details to the supplementary. In the appendix, we also provide a further simplification for in-sample $\mathbf{x}$, which can be used directly during training.

Given equation \eqref{eqn:LOO_mean}, we use the leave-one-out mean-squared error as the objective: \begin{align}
    (\hat{\mathbf{w}}, \hat \theta, \hat \gamma) = \argmin_{\mathbf{w}, \theta, \gamma} \frac{1}{n} \sum_{i = 1}^n \left\{y_i - \hat m_{\sim i}(\mathbf{x}_i)\right\}^2. \label{eqn:LOO_loss}
\end{align} Our implementation uses $\hat{\mathbf{w}}_0$ from Section \ref{subsec:initial} as the initial $\mathbf{w}$, and conducts a grid search for an initial $(\theta, \gamma)$. We use the BOBYQA method \citep{Powell_2009} in the \texttt{NLopt} library \citep{NLopt} for local non-gradient optimization, which provided a stable result in 100 iterations.

\subsection{Conformal prediction intervals} \label{subsec:conformal}


Kriging and Gaussian process modeling are closely related techniques, with the former viewed as a frequentist and the latter as a Bayesian approach \citep{sacks1989, currin1991bayesian}. Under the Gaussian process assumption for the stochastic process, the $(1 - \eta)$ confidence interval for the prediction can be obtained as \citep{Joseph_2025}
\begin{align}
    \hat m_{\mathrm{RQK}}(\mathbf{x}) \pm \Phi^{-1}(1 - \eta/2) s_{\mathrm{RQK}}(\mathbf{x}),
\end{align} where $\Phi^{-1}(\cdot)$ is the quantile function of the standard normal distribution. However, the Gaussian distributional assumption may not be valid, and the above confidence interval may not have the required coverage. Therefore, we instead adapt the leave-one-out cross-conformal prediction approach \citep{Vovk_2015} to build nonparametric confidence intervals for the predictions.

Fixing the hyperparameters $(\mathbf{w}, \theta, \gamma, \nu)$, we can construct the leave-one-out scores for the $n$ in-sample points: \begin{align}
    S_i = \frac{|y_i - \hat m_{\sim i}(\mathbf{x}_i)|}{s_{\sim i}(\mathbf{x}_i)}, i = 1, \dots, n.
\end{align}
\citet{Barber_2021} points out that, for a theoretical guarantee of the cross-conformal approach, we need to build scores from the leave-one-out models: \begin{align}
    S^{(i)}(\mathbf{x}, y) = \frac{|y - \hat m_{\sim i}(\mathbf{x})|}{s_{\sim i}(\mathbf{x})}, \quad i = 1, \dots, n,
\end{align} where $\hat m_{\sim i}(\mathbf{x})$ and $s_{\sim i}(\mathbf{x})$ were given in Equations \eqref{eqn:LOO_mean} and \eqref{eqn:LOO_sd}. Then, comparing $S^{(i)}(\mathbf{x}, y)$ with $S_i$, we can construct a continuous conformal prediction interval: \begin{align}
    \left[Q_{\eta}^{-} \left(\{\hat m_{\sim i}(\mathbf{x}) - S_i s_{\sim i}(\mathbf{x})\}_{i = 1}^n \right), Q_{1 - \eta}^{+} \left(\{\hat m_{\sim i}(\mathbf{x}) + S_i s_{\sim i}(\mathbf{x})\}_{i = 1}^n \right)\right]
    \label{eqn:conformal_interval}
\end{align} where, $Q_{\eta}^-(\cdot), Q_{1 - \eta}^+(\cdot)$ denote the $\lfloor \eta(n + 1)\rfloor$-th and $\lceil (1 - \eta)(n + 1)\rceil$-th empirical quantile, respectively. The hyperparameter $\nu$ cancels out in $S_i s_{\sim i}(\mathbf{x})$, thereby eliminating the need to estimate it.

\begin{proposition} \label{prop:conformal}
Suppose $\{(\mathbf{X}_i, Y_i)\}_{i \geq 1}$ are i.i.d. random vectors, and the hyperparameters $(\mathbf{w}, \theta, \gamma, \nu)$ are fixed before seeing the data. Then, the interval \eqref{eqn:conformal_interval} satisfies \begin{align}
    \mathbb{P}\{Y_{n + 1} \in \mathcal{C}_{\eta}(\mathbf{X}_{n + 1})\} \geq 1 - 2\eta.
\end{align}
\end{proposition}
The interval is a superset of the conformal prediction interval defined in \citet{Barber_2021}; hence, the result follows directly from Theorem 4 in their article\footnote{We also implement the original definition in our code, although we recommend our version for continuous prediction intervals. The supplementary provides further details on the relationship.}. They also noted that, while the theory guarantees a coverage probability of $1 - 2\eta$, empirical performance is often close to the desired $1 - \eta$ coverage.

The interval no longer relies on the Gaussian distributional assumption. Instead, the uncertainties $s_{\sim i}(\mathbf{x})$ are used as references for the \textit{relative uncertainty}. When data size is reasonably large, $s_{\sim i}(\mathbf{x})$ becomes close to $s(\mathbf{x})$. Therefore, the bandwidth depends on the location $\mathbf{x}$, and we preserve the football-shaped uncertainty quantification of kriging, as we saw in Figure \ref{fig:motivating_example_RQK}.

In our implementation, we fix $\nu = 1$. Because the $n$ observations are also used to estimate $(\mathbf{w}, \theta, \gamma)$, the exchangeability condition does not strictly hold. Nevertheless, our simulations in the next section show that the empirical coverage is close to $1 - \eta$ for both deterministic and noisy data. Therefore, we keep the leave-one-out conformal prediction procedure, and note that split conformal \citep{Papadopoulos_2002, Lei_2018} is an alternative if strict theoretical guarantees are desired.

\section{Numerical experiments} \label{sec:simulation}

This section presents numerical experiments on both deterministic simulation functions and real-world datasets with noise, demonstrating the consistent performance of the proposed RQK strategy across multiple benchmarks.

\subsection{Deterministic functions} \label{subsec:noiseless}

We consider four functions that are widely used as benchmarks for black-box function emulation: 
\begin{itemize}
\item Borehole function ($d_0 = 8$): \begin{align}
    m(\mathbf{x}) = \frac{2\pi x_3 (x_4 - x_6)}{\ln(x_2 / x_1) \left[ 1 + \frac{2 x_7 x_3}{\ln(x_2 / x_1) x_1^2 x_8} + \frac{x_3}{x_5} \right]}.
\end{align}

\item Dette \& Pepelyshev 8-dimensional function ($d_0 = 8$): \begin{align}
    m(\mathbf{x}) = 4(x_1 - 2 + 8x_2 - 8x_2^2)^2 + (3x_1 - 5 + 11x_2 - 11x_2^2)^2 + \sum_{j=3}^8 j \ln \left( 1 + \sum_{i=1}^j x_i \right).
\end{align}

\item Friedman function ($d_0 = 5$): \begin{align}
    m(\mathbf{x}) = 10 \sin(\pi x_1 x_2) + 20(x_3 - 0.5)^2 + 10x_4 + 5x_5.
\end{align}

\item OTL circuit function ($d_0 = 6$): \begin{align}
    m(\mathbf{x}) = \frac{\left( \frac{12x_2}{x_1 + x_2} + 0.74 \right) x_6 (x_5 + 9)}{x_6 (x_5 + 9) + x_3} + \frac{11.35 x_3}{x_6 (x_5 + 9) + x_3}.
\end{align}
\end{itemize}
We scale the inputs to lie on $[0, 1]$. See \citet{simulationlib} for details on the original scales of the inputs.

For each function, we consider two simulation settings. In the first setting, the input has the same dimension as the original dimension ($d = d_0$). Additionally, we augment the functions with $20 - d_0$ inert dimensions to create problems with $d = 20$ dimensions, making the emulation more difficult. We run experiments with $n = 20d, 40d, 60d$ i.i.d. training samples generated from a uniform distribution. For testing, we use a 10,000-point Sobol' sequence for out-of-sample evaluation.

We run the RQK approach with $p = (d + 1)/2$ and $\gamma = 10^{-8}$, and fit the parameters $(\mathbf{w}, \theta)$ using leave-one-out cross validation. For comparison, we use several other kriging implementations available in \texttt{R}: \begin{enumerate}
\item The \texttt{rkriging} package \citep{rkriging} implements maximum likelihood estimation of both ordinary kriging and rational kriging, which are denoted as ``rkriging (OK)'' and ``rkriging (RK)'' in the plots. The package supports the rational quadratic kernel, so we use it with $p = (d + 1)/2$ for comparison. 

The remaining packages do not have built-in support for the rational quadratic kernel, so we use the squared-exponential kernel. Also, all these packages implement ordinary kriging.

\item The \texttt{DiceKriging} package \citep{DiceKriging} supports both maximum likelihood and leave-one-out estimation of deterministic data. They are denoted as ``Dice (MLE)'' and ``Dice (LOO)''.

\item The \texttt{hetGP} package also supports both maximum likelihood and leave-one-out estimation. However, the leave-one-out error is unreasonably high in our simulations, so we present only the maximum likelihood result and denote it as ``hetGP''.

\item The \texttt{mlegp} package also uses maximum likelihood estimation. It is denoted as ``mlegp''.
\end{enumerate}
All the above comparison methods are run with the default settings provided in the package, except that the nugget $\gamma$ is explicitly set to $10^{-8}$. Each simulation is repeated over 30 random seeds for data generation, and all methods share the same data.

\begin{figure}[p]
\centering

\begin{subfigure}{\textwidth}
\centering
\includegraphics[width=0.4\textwidth]{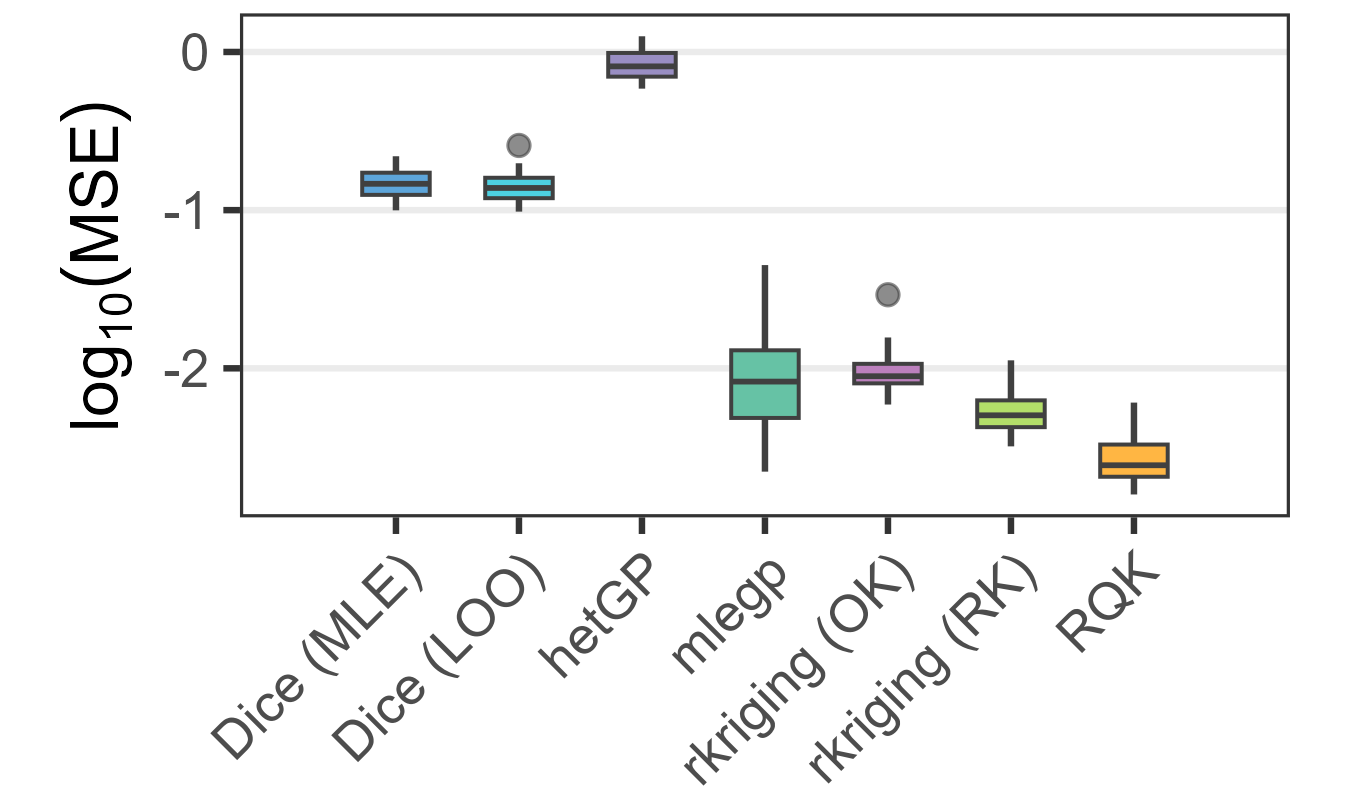}
\includegraphics[width=0.4\textwidth]{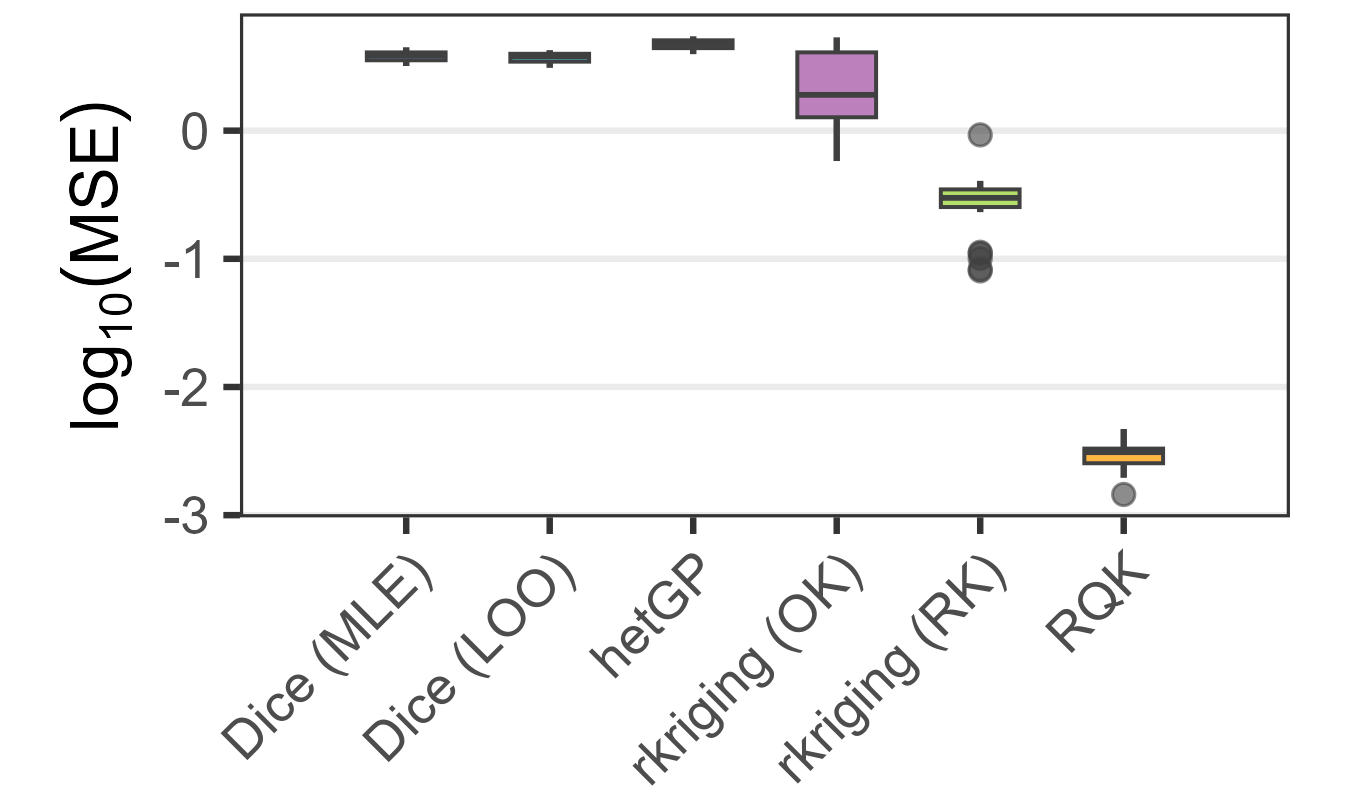}
\caption{Borehole function. Left: 8 dimensions. Right: 20 dimensions.}
\end{subfigure}%

\begin{subfigure}{\textwidth}
\centering
\includegraphics[width=0.4\textwidth]{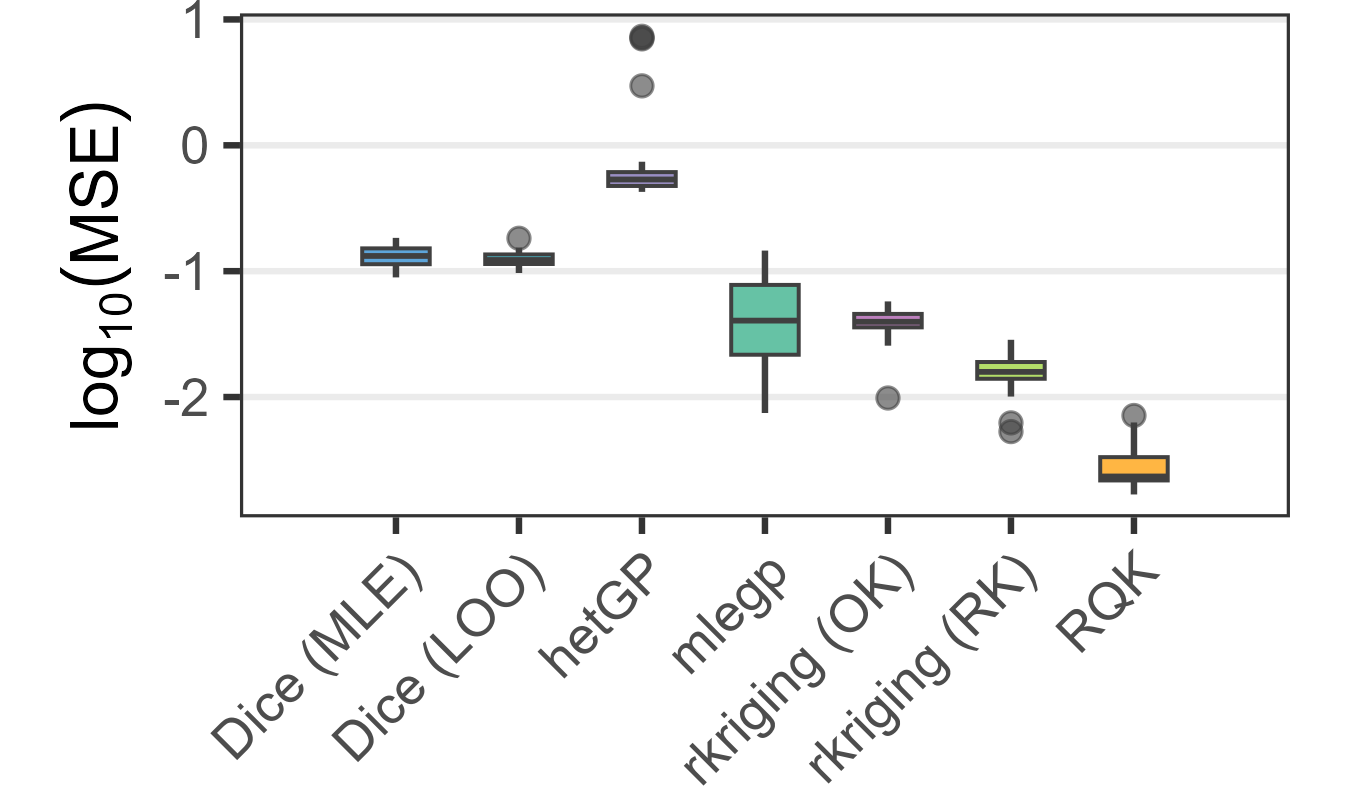}
\includegraphics[width=0.4\textwidth]{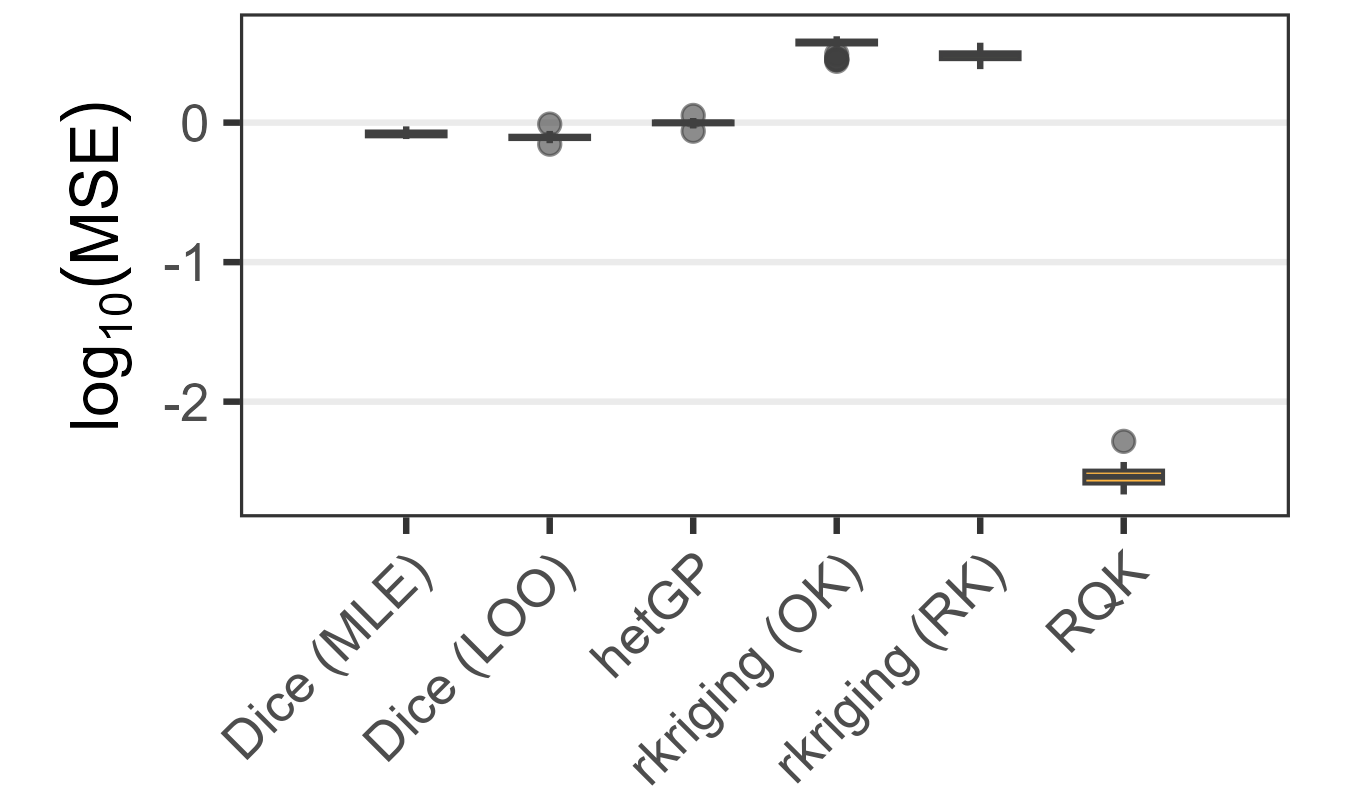}
\caption{Dette and Pepelyshev 8d function. Left: 8 dimensions. Right: 20 dimensions.}
\end{subfigure}%

\begin{subfigure}{\textwidth}
\centering
\includegraphics[width=0.4\textwidth]{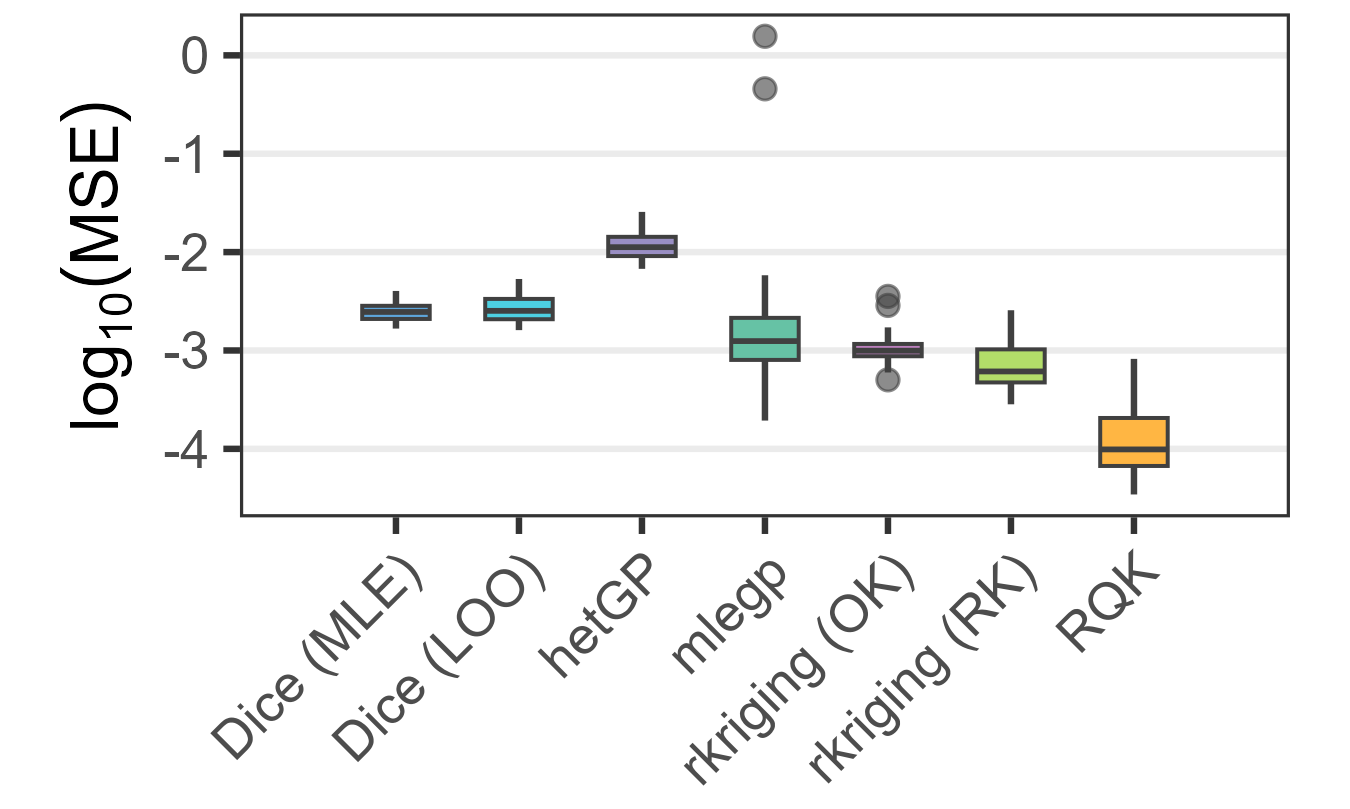}
\includegraphics[width=0.4\textwidth]{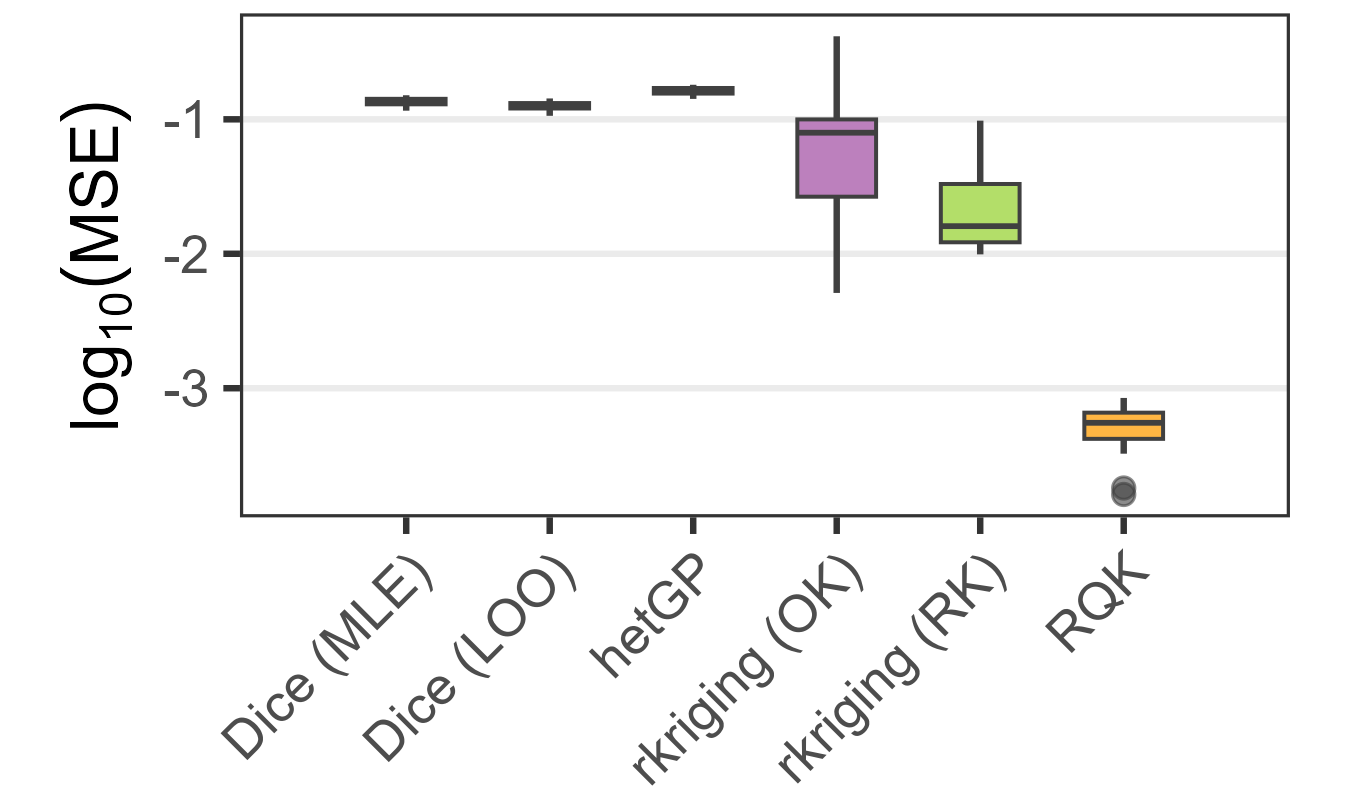}
\caption{Friedman function. Left: 5 dimensions. Right: 20 dimensions.}
\end{subfigure}%

\begin{subfigure}{\textwidth}
\centering
\includegraphics[width=0.4\textwidth]{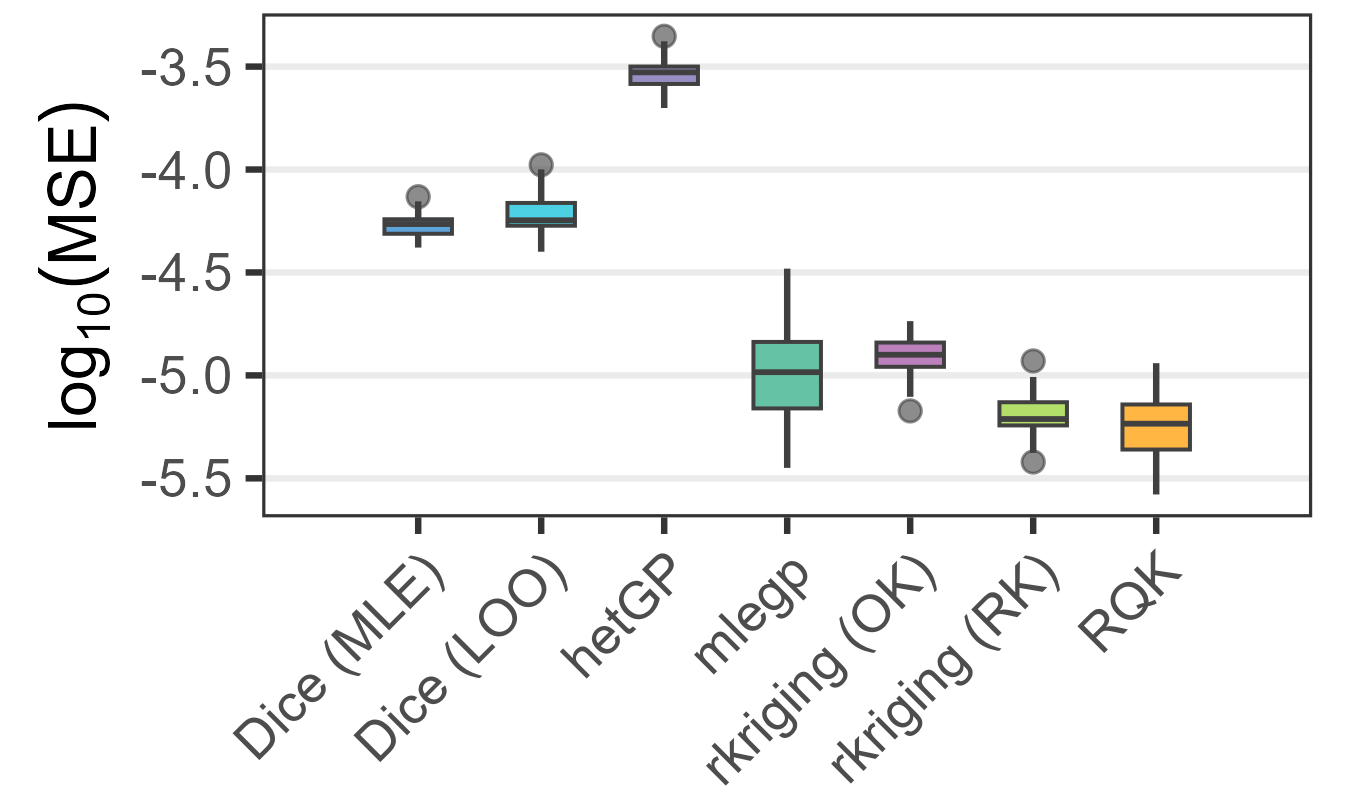}
\includegraphics[width=0.4\textwidth]{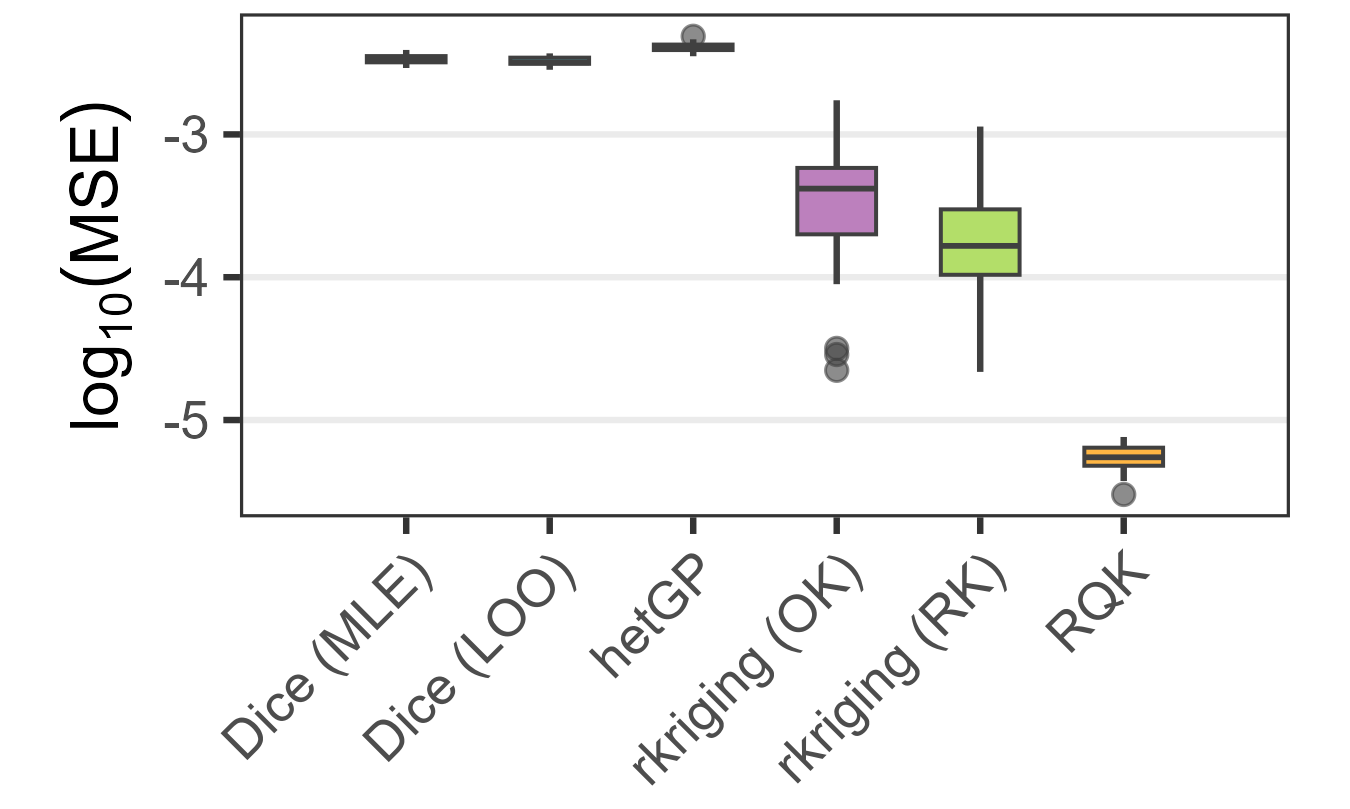}
\caption{OTL circuit function. Left: 6 dimensions. Right: 20 dimensions.}
\end{subfigure}%

\caption{Logarithm MSE plots for deterministic functions with $n = 40d$. The original functions are on the left, and functions with augmented inert variable are on the right. Each experiment is repeated over 30 random seeds for data generation.}
\label{fig:noiseless}

\end{figure}

For the quality of the point predictors, we evaluate the mean-squared error (MSE) on the test set: \begin{align}
    \mathrm{MSE} = \frac{1}{n_{\mathrm{test}}} \sum_{i = 1}^{n_{\mathrm{test}}} \{\hat m(\mathbf{x}_i) - m(\mathbf{x}_i)\}^2.
\end{align}
Figure \ref{fig:noiseless} presents the logarithm of the MSE for all four functions. Due to space constraints, we present results only for $n = 40d$ and defer the qualitatively similar $n = 20d$ and $60d$ to the supplementary. The plots clearly show that RQK consistently outperforms other methods across all functions, and its advantage is more pronounced when inert factors are present in the $d = 20$ case. 

\begin{figure}[p]
\centering

\begin{subfigure}{\textwidth}
\centering
\includegraphics[width=0.4\textwidth]{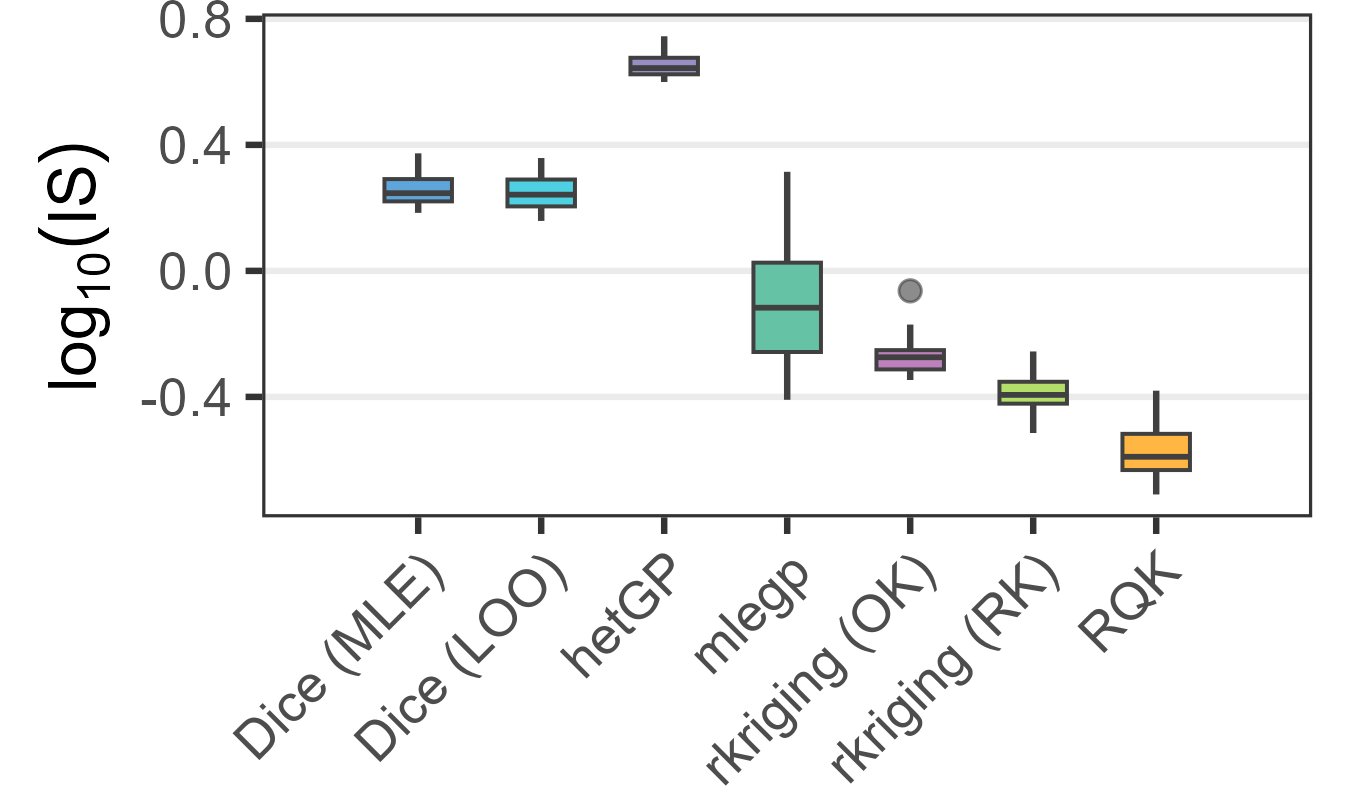}
\includegraphics[width=0.4\textwidth]{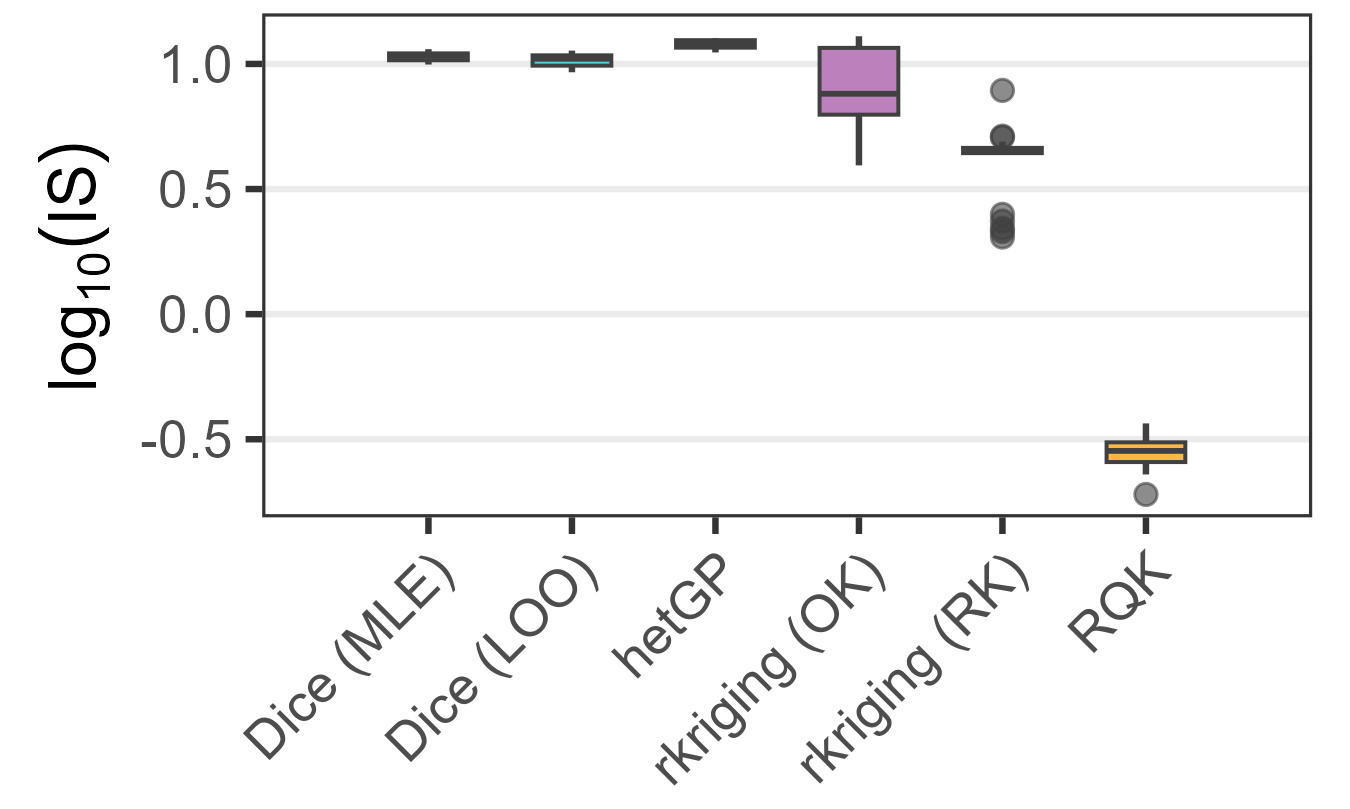}
\caption{Borehole function. Left: 8 dimensions. Right: 20 dimensions.}
\end{subfigure}%

\begin{subfigure}{\textwidth}
\centering
\includegraphics[width=0.4\textwidth]{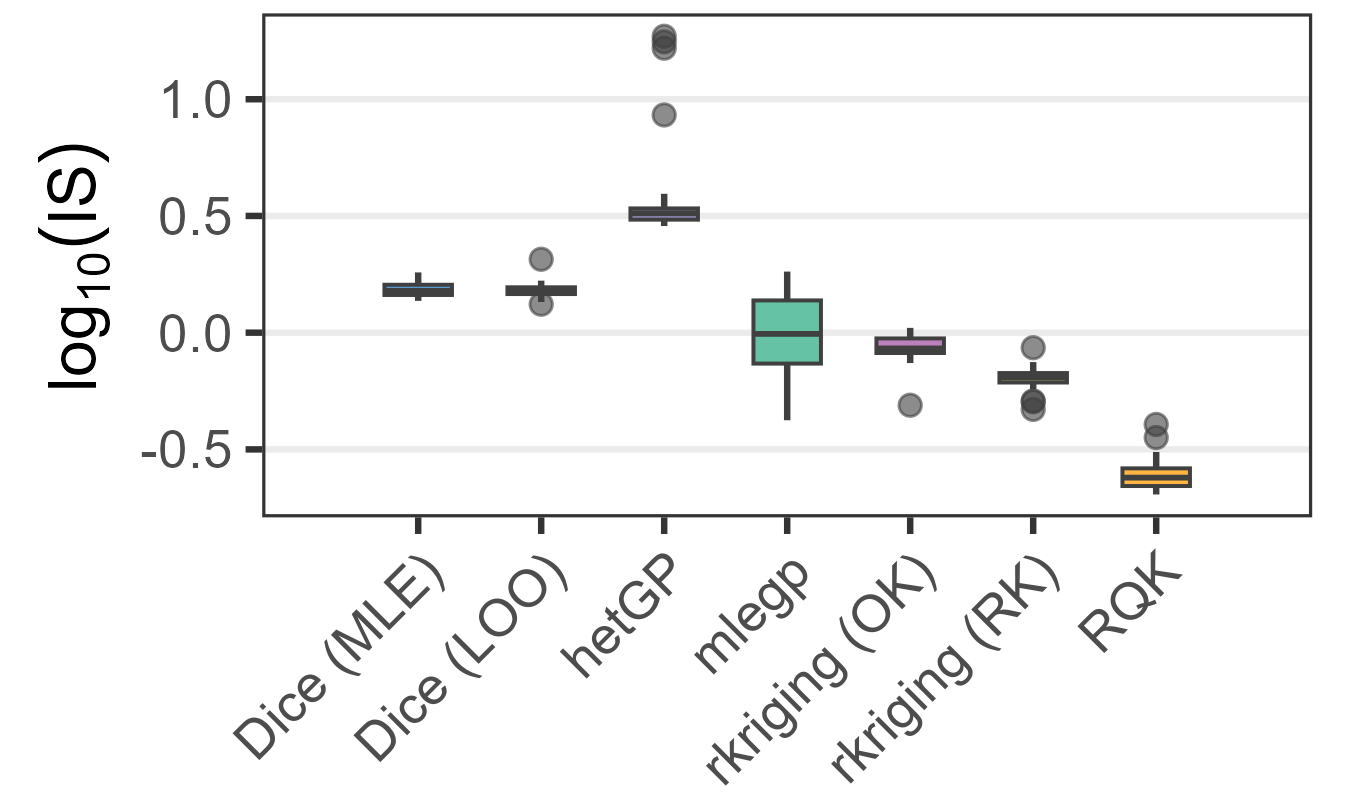}
\includegraphics[width=0.4\textwidth]{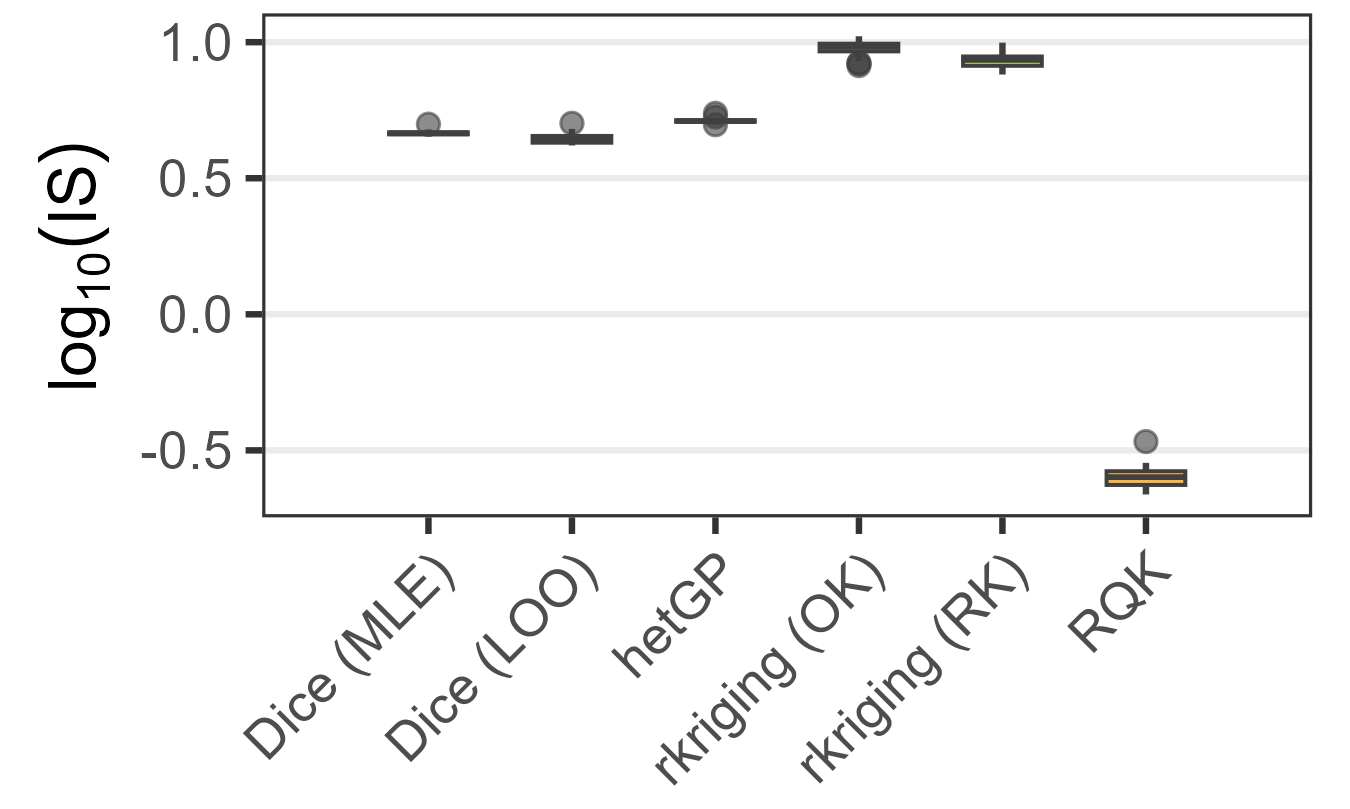}
\caption{Dette and Pepelyshev 8d function. Left: 8 dimensions. Right: 20 dimensions.}
\end{subfigure}%

\begin{subfigure}{\textwidth}
\centering
\includegraphics[width=0.4\textwidth]{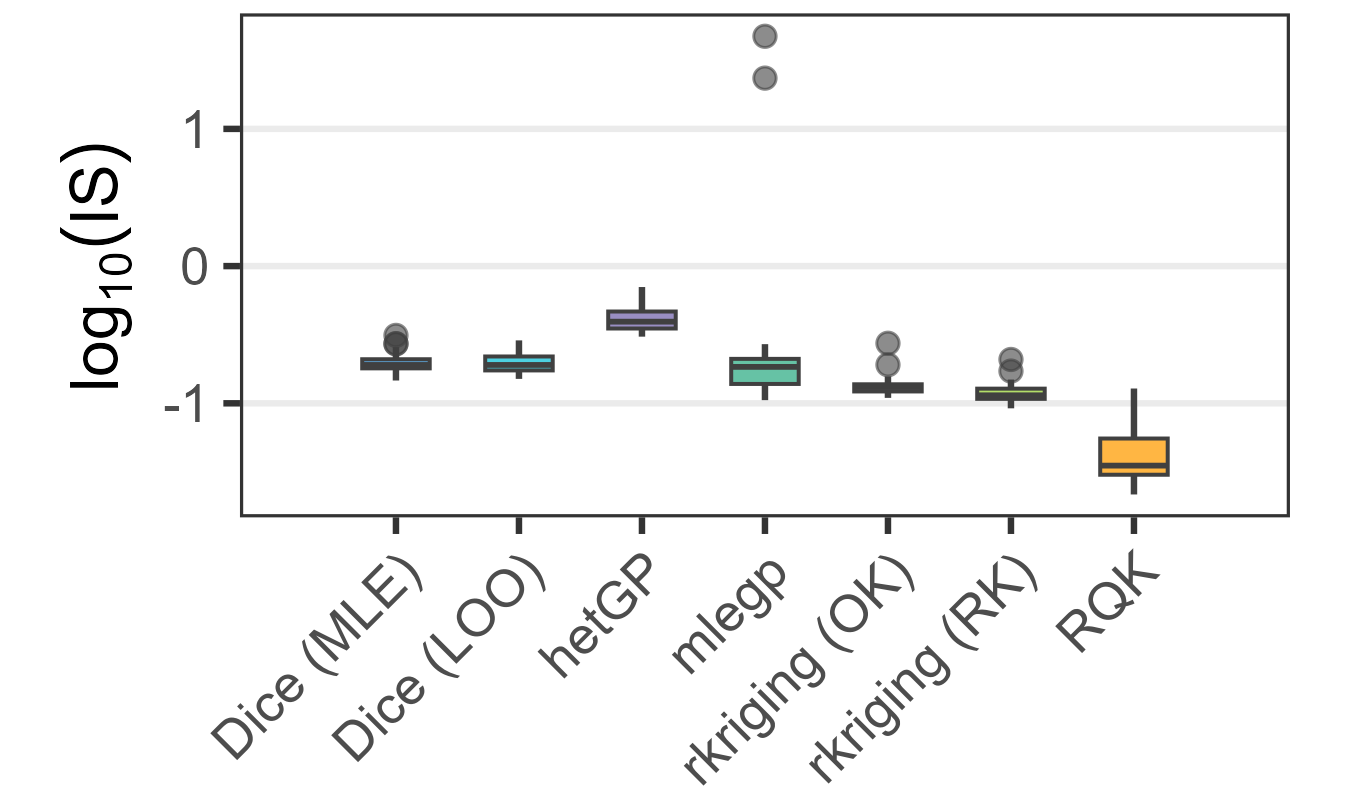}
\includegraphics[width=0.4\textwidth]{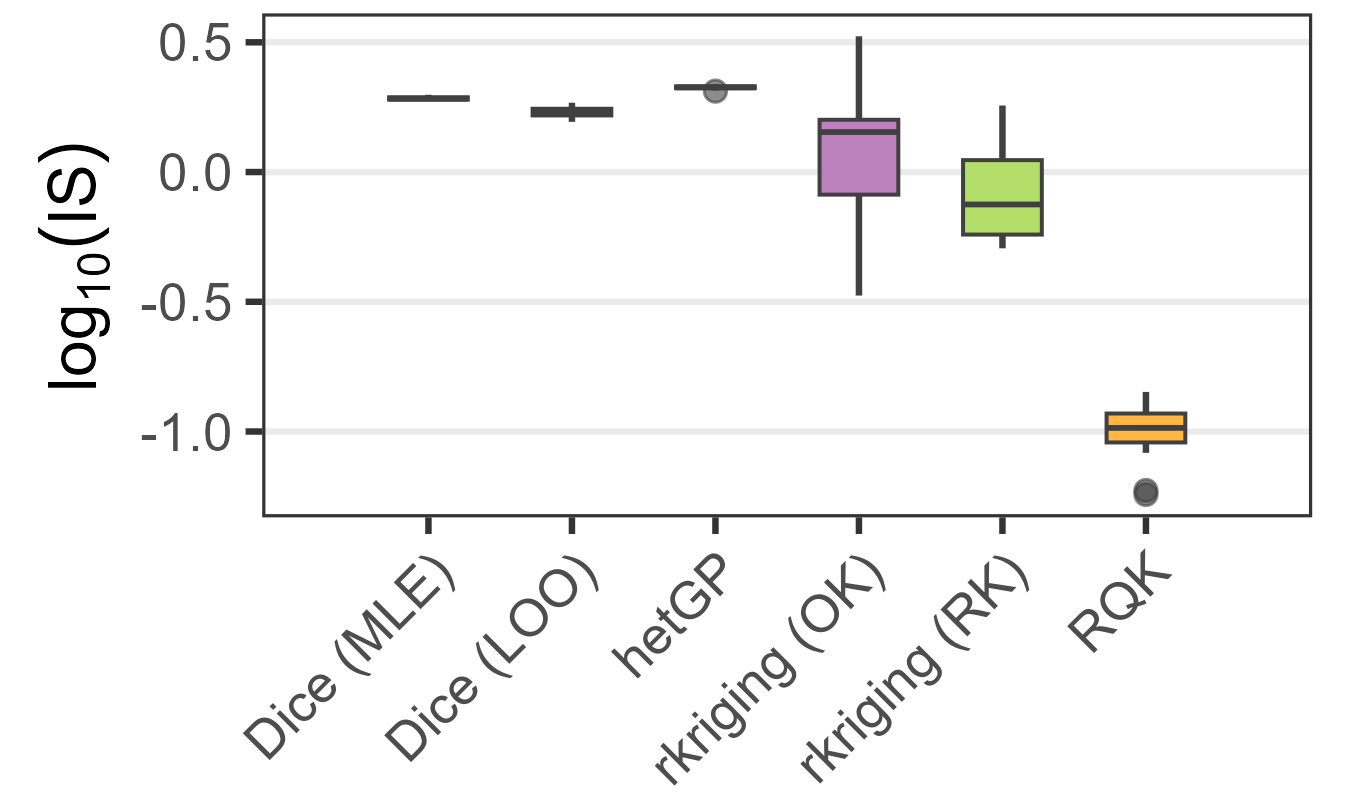}
\caption{Friedman function. Left: 5 dimensions. Right: 20 dimensions.}
\end{subfigure}%

\begin{subfigure}{\textwidth}
\centering
\includegraphics[width=0.4\textwidth]{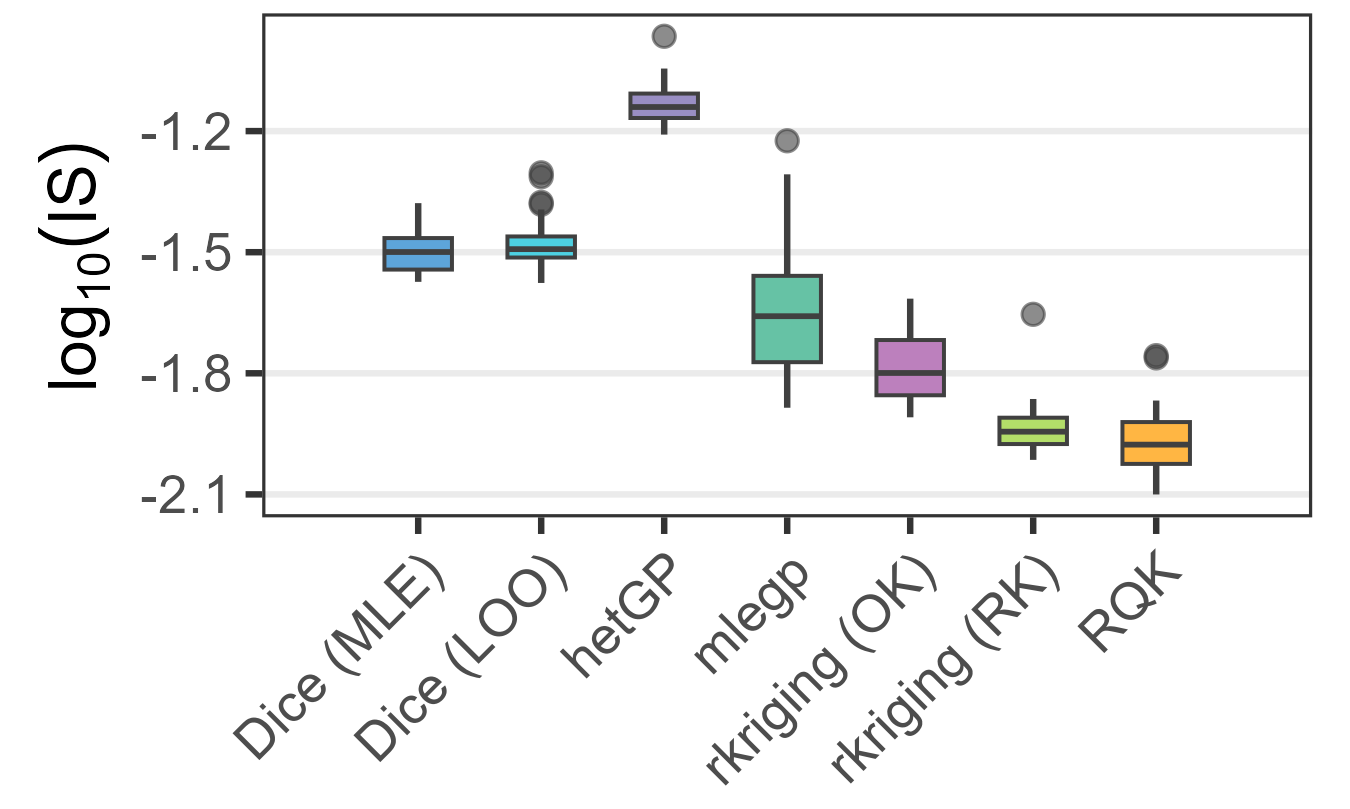}
\includegraphics[width=0.4\textwidth]{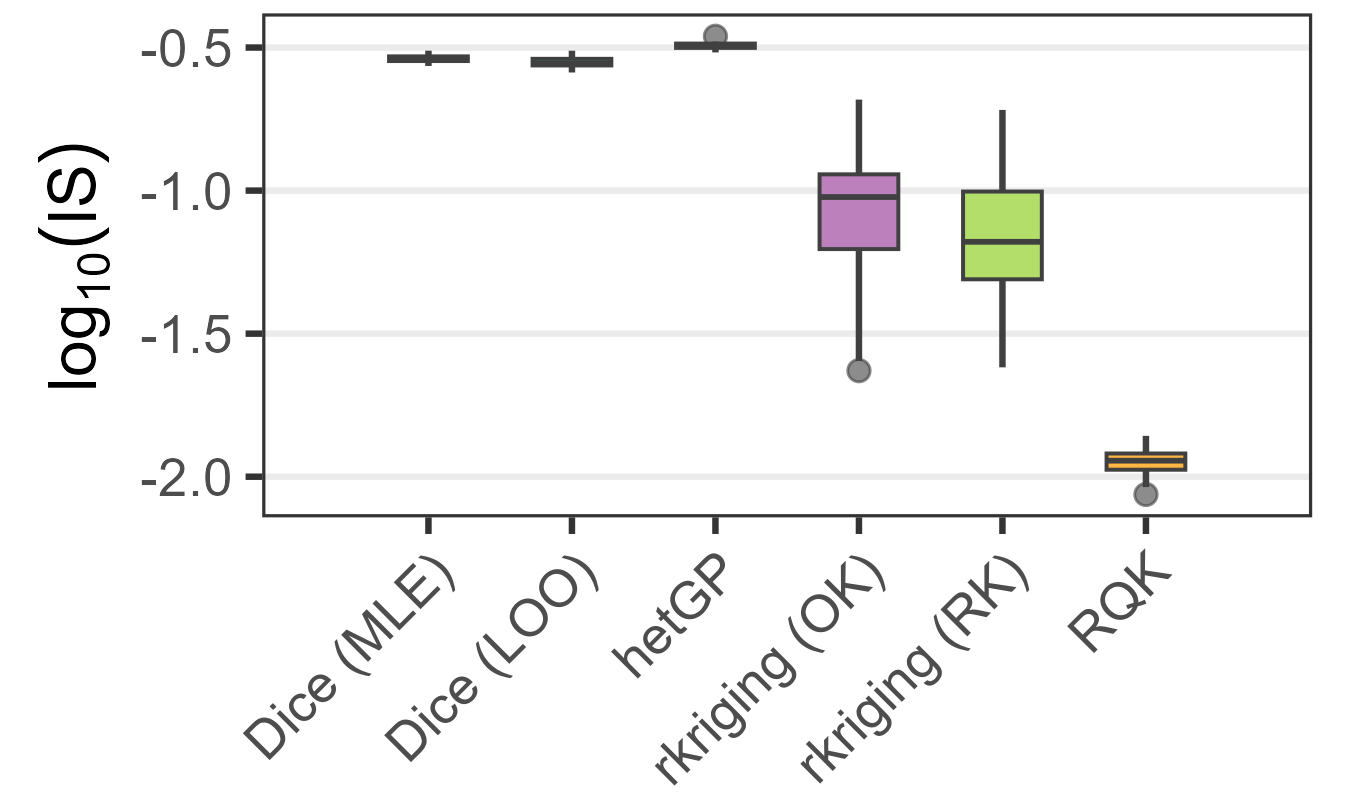}
\caption{OTL circuit function. Left: 6 dimensions. Right: 20 dimensions.}
\end{subfigure}%

\caption{Logarithm IS plots for deterministic functions with $n = 40d$. The original functions are on the left, and functions with augmented inert variable are on the right. Each experiment is repeated over 30 random seeds for data generation.}
\label{fig:noiseless_IS}

\end{figure}

Beyond point prediction, we also care about the uncertainty quantification of kriging methods. Here, we evaluate the interval score (IS, \citealp{Gneiting_Raftery_2007}): \begin{align}
    \mathrm{IS} = \frac{1}{n_{\mathrm{test}}} \sum_{i = 1}^{n_{\mathrm{test}}} \left\{(u - l) + \frac{2}{\eta} \Big([l - m(\mathbf{x}_i)]_+ + [m(\mathbf{x}_i) - u]_+\Big)\right\}, \label{eqn:interval_score}
\end{align} where $[\cdot]_+ = \max\{\cdot, 0\}$, and $[l, u]$ is the prediction interval. For RQK, the interval is constructed by the conformal prediction procedure in Section \ref{subsec:conformal}. We use the variant \eqref{eqn:conformal_interval}, as it is more conservative. For other kriging methods, the intervals are constructed by $[\hat m(\mathbf{x}) - 1.96 s(\mathbf{x}), \hat m(\mathbf{x}) + 1.96s(\mathbf{x}_i)]$. The metric favors short prediction intervals and penalizes coverage gaps by a factor of $2/\eta$ (recall that $1 - \eta$ is the desired coverage). Figure \ref{fig:noiseless_IS} depicts the logarithm IS for the four functions. The patterns are similar to Figure \ref{fig:noiseless}: more accurate methods also have shorter prediction intervals. For these relatively smooth functions, all methods provide coverage close to 95\% (not shown in the figures). In Section \ref{subsec:noisy}, we will see that coverage varies greatly for the noisy datasets.

Comparisons with existing methods show that no single component of the RQK procedure is the ``magic ingredient''. First, \texttt{DiceKriging} shows that leave-one-out estimation will not provide great improvements over maximum likelihood estimation. Using the rational quadratic kernel by \texttt{rkriging} also provides similar performance with the squared-exponential kernel. Finally, rational kriging yields lower error than ordinary kriging in many cases, but it still cannot outperform RQK. Rather, by combining the rational quadratic kernel with the rational kriging form, we can leverage the kernel regression relationship to obtain better predictors. 

\begin{figure}[tb]
\centering

\begin{subfigure}{\textwidth}
\centering
\includegraphics[width=0.48\textwidth,valign=c]{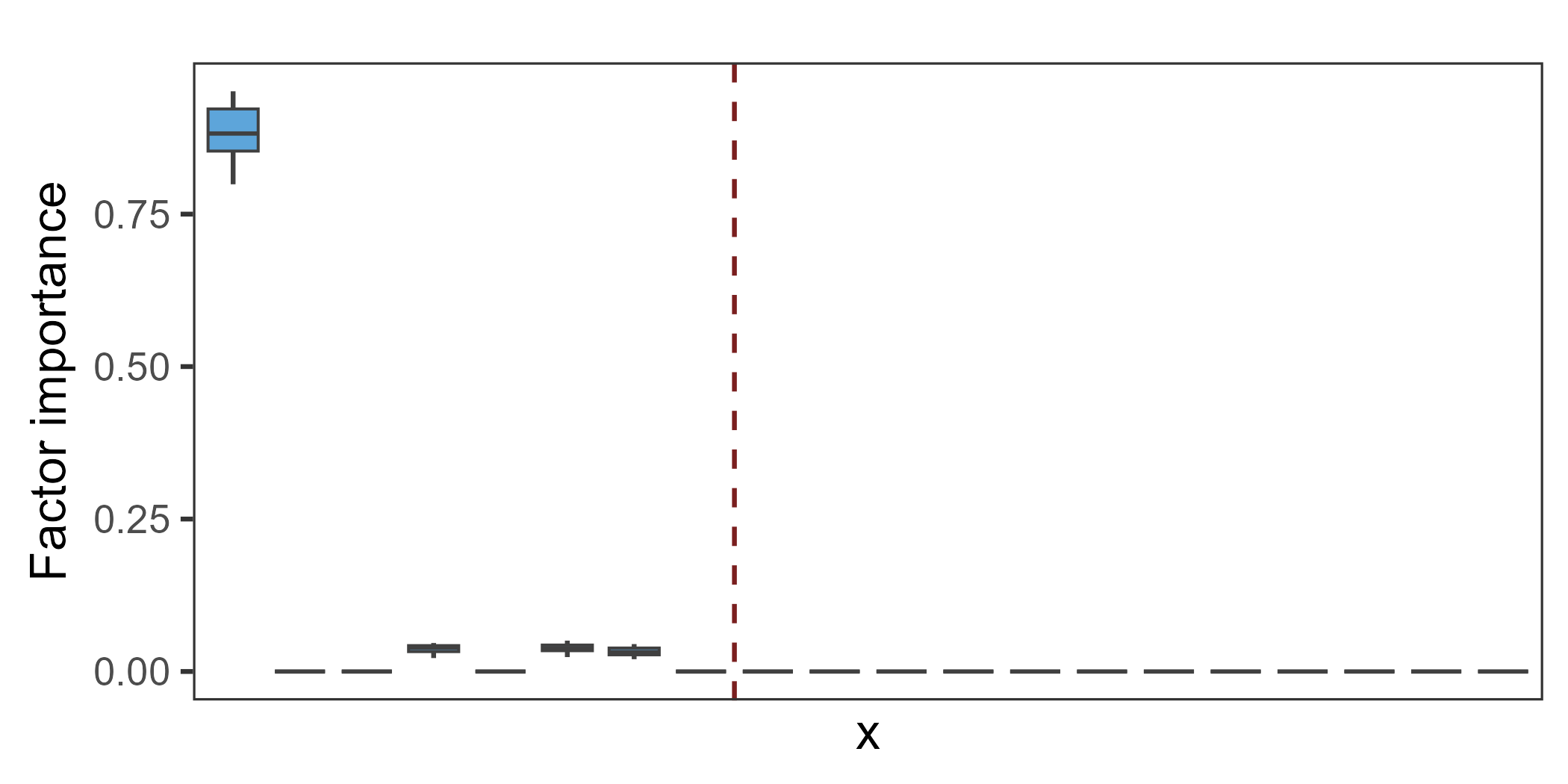}
\includegraphics[width=0.48\textwidth,valign=c]{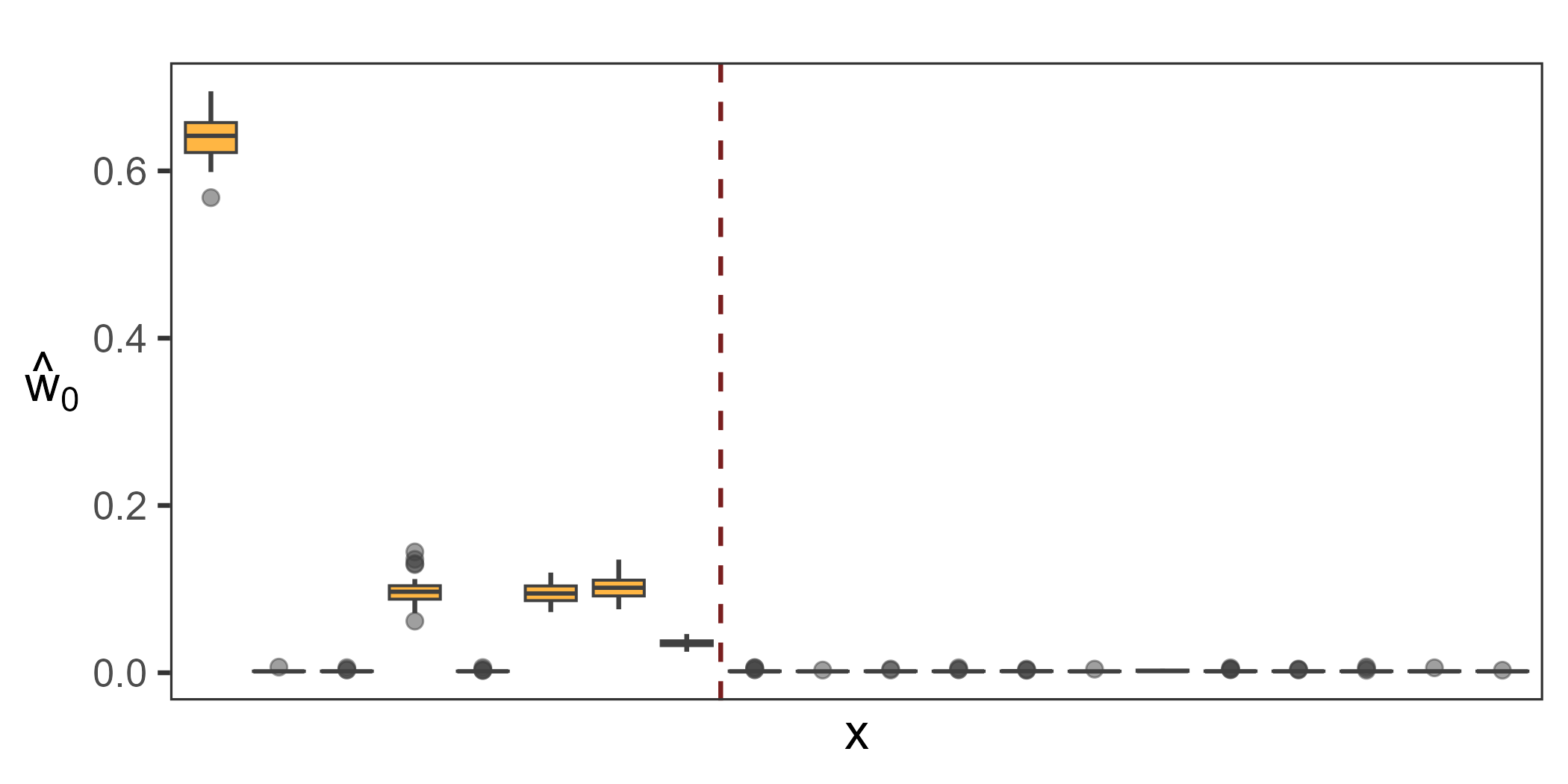}
\caption{Comparison of factor importance (left) and RQK initial weights (right).}
\end{subfigure}%

\begin{subfigure}{\textwidth}
\centering
\includegraphics[width=0.48\textwidth,valign=c]{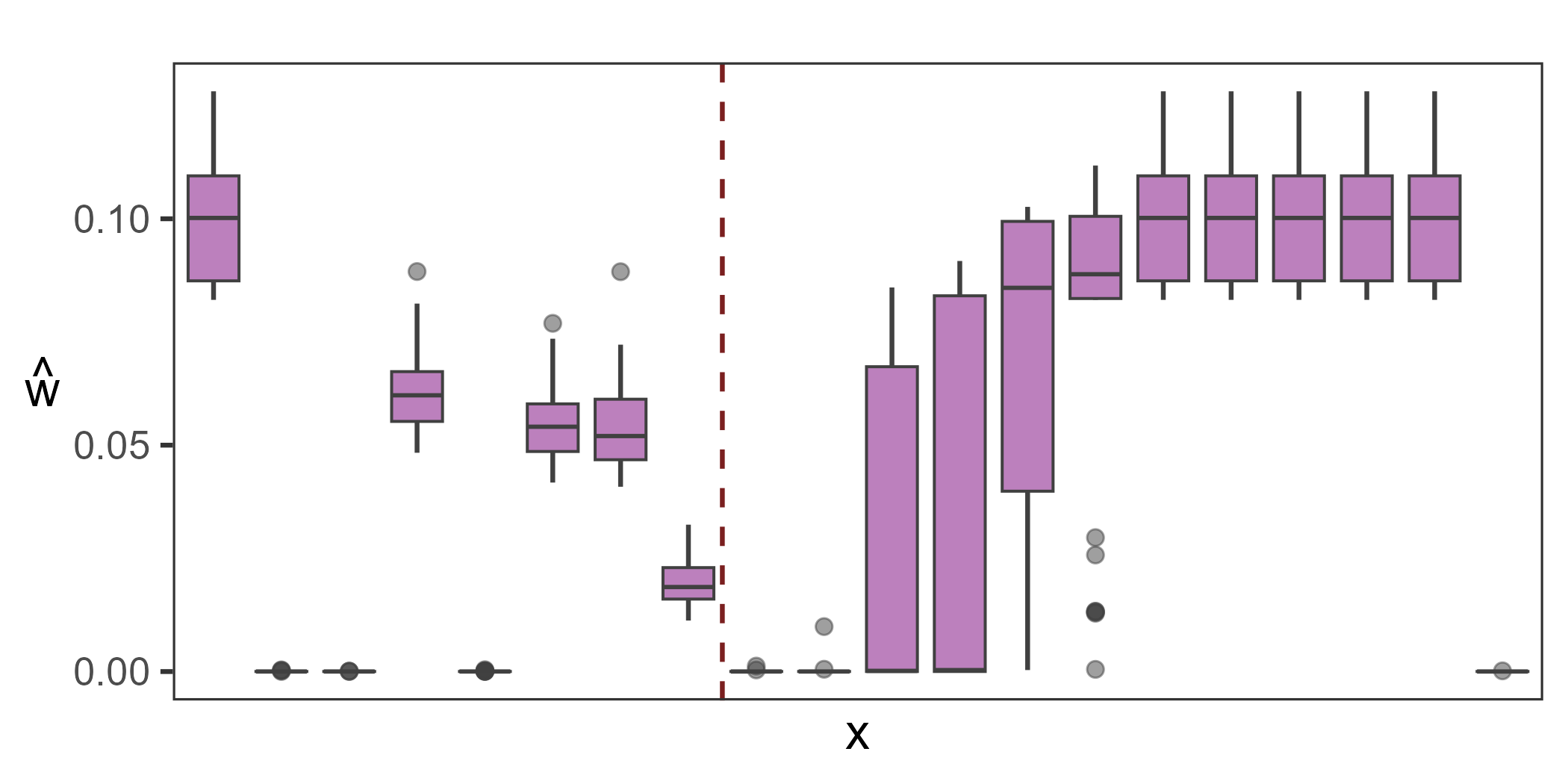}
\includegraphics[width=0.48\textwidth,valign=c]{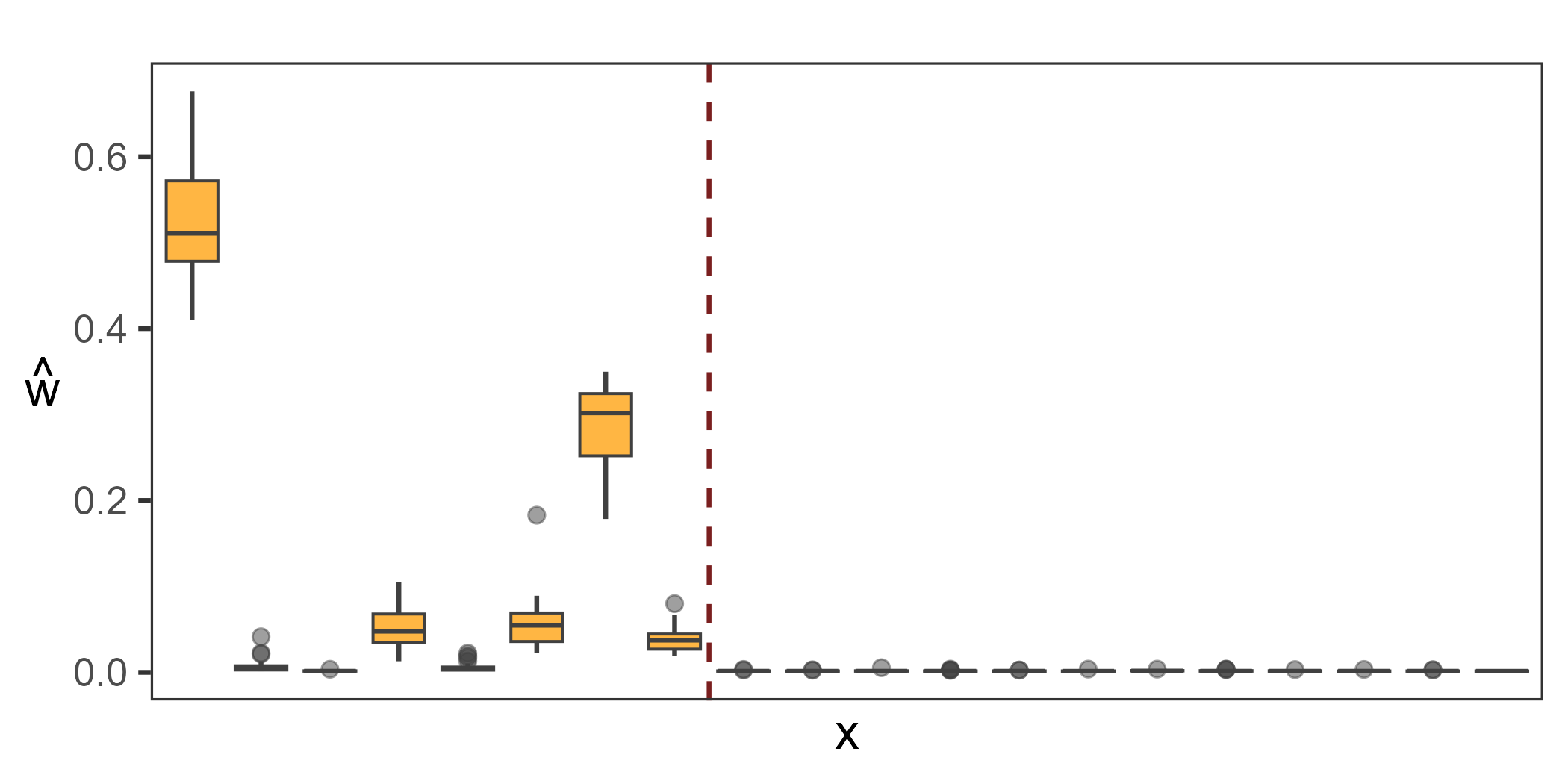}
\caption{Ordinary kriging weights (left) and RQK weights (right).}
\end{subfigure}%

\caption{Comparison of factor importance (estimated total Sobol' indices), RQK initial weights, and final weight estimations for the 20-dimensional Borehole problem. The experiment is repeated over 30 random seeds for data generation.}
\label{fig:noiseless_weight}

\end{figure}

To elaborate on the role of the initial weight estimation procedure outlined in Section \ref{subsec:initial}, we present an example in Figure \ref{fig:noiseless_weight} using the Borehole function. With $n = 800$ samples from the $d = 20$ setting, we can estimate factor importance from the data. Here, we use \texttt{FIRST}, introduced by \citet{Huang_2025}, which estimates total Sobol' indices from scattered data. The left side of panel (a) shows the factor importance summing to one given by \texttt{FIRST}. The result shows that $x_1, x_4, x_6, x_7$ are important, agreeing with the previous analysis of the Borehole function \citep{Moon_2010}. Furthermore, the additional inert variables ($x_9$ to $x_{20}$, right of the dashed line) have zero importance. Next, the right side of panel (a) presents the estimate of $\hat{\mathbf{w}}_0$ from kernel regression. We can see that, although the weights indicate relevance in the kernel distance and should not be directly interpreted as global sensitivity indices, they closely align with the factor importance. Therefore, the estimates serve as reasonable starting points for the expensive leave-one-out optimization. Panel (b) compares the ordinary kriging weights from \texttt{rkriging} and the RQK weights, showing that RQK is more stable and better reflects the relative importance, ultimately contributing to smaller approximation error.

\begin{table}[tb]
\centering
\caption{Median computation time (in seconds) over 30 runs for the Borehole function.}
\label{tab:noiseless_time}
\begin{tabular}{lcccccc}
\toprule
& \multicolumn{3}{c}{$p=8$} & \multicolumn{3}{c}{$p=20$} \\
\cmidrule(lr){2-4} \cmidrule(lr){5-7}
 & $n=160$ & $n=320$ & $n=480$ & $n=400$ & $n=800$ & $n=1200$ \\
\midrule
hetGP      & 0.07 & 0.28 & 0.58 & 0.77 & 2.54 & 6.50 \\
rkriging (OK)         & 0.12 & 0.38 & 0.80 & 0.69 & 2.70 & 6.50 \\
\textbf{RQK}        & \textbf{0.13} & \textbf{0.42} & \textbf{0.86} & \textbf{0.73} & \textbf{2.98} & \textbf{7.20} \\
Dice (MLE) & 0.12 & 0.35 & 0.77 & 1.40 & 5.42 & 11.27 \\
Dice (LOO) & 0.19 & 0.55 & 0.93 & 2.97 & 11.34 & 25.17 \\
rkriging (RK)         & 0.36 & 2.02 & 6.82 & 4.51 & 31.84 & 104.92 \\
mlegp      & 8.66 & 48.09 & 148.24 & 99.19 & $>$ 300 & $>$ 300 \\
\bottomrule
\end{tabular}
\end{table}

Finally, we compare the computational time of the packages. Table \ref{tab:noiseless_time} gives the median times (including estimation and prediction) over 30 runs on an Apple M1 chip with 16 GB RAM, sorted from fastest to slowest. We can make several observations: First, the rational kriging implementation in \texttt{rkriging} does not scale well, mainly due to the computation of the Perron eigenvector $\mathbf{c}$. In contrast, RQK obtains $\mathbf{c}$ directly through the correlation matrix $\mathbf{R}$. Second, leave-one-out estimation is more costly than maximum likelihood. Although both have complexity $\mathcal{O}(n^3)$, calculating leave-one-out errors requires more matrix solve operations and is typically slower. Nevertheless, our RQK implementation achieves a comparable computation time as the two fastest packages, \texttt{hetGP} and \texttt{rkriging}, without compromising accuracy and stability. There is one caveat: we perform efficient linear algebra on dense matrices using \texttt{RcppArmadillo} \citep{RcppArma_2025}, which will be slow using the \texttt{R} default BLAS (basic linear algebra subprogram). For optimal speed and replicating the results in this work, we strongly recommend using an external high-performance BLAS library. See (\citealp[Appendix A.3.1]{Rinstallation_2010}; \citealp{Eddelbuettel_2024}) for more details and best practices.

\subsection{Noisy data} \label{subsec:noisy}

We also apply the RQK methodology to real-world noisy datasets with varying sample sizes and dimensions. \begin{itemize}
\item Concrete compressive strength dataset \citep{Yeh_1998}. The dataset contains 1,030 samples of high-performance concrete. The objective is to use the eight input factors to predict concrete strength.

\item Appliances energy prediction dataset \citep{Energy_uci_2017}. The original dataset contains 138 daily time series sampled at 10-minute intervals. We limit our interest to samples in $[19:00, 20:00)$ each day, resulting in a dataset of 822 rows. There are 28 input factors for predicting energy usage in watts.

\item Meat spectroscopy dataset \citep{Thodberg_1993}. The dataset contains 215 meat samples, and the goal is to use a 100-channel spectrum of absorbance values to predict fat content. We note that \citet{Huang_2025} identified only three of the 100 channels as having a significant impact.
\end{itemize}

For RQK, we use $p = d/2$ and estimate $(\mathbf{w}, \theta, \gamma)$ using leave-one-out cross validation. Among the aforementioned kriging packages, \texttt{DiceKriging} (maximum likelihood only), \texttt{hetGP}, and \texttt{rkriging} (ordinary kriging only) support noisy data with an estimated nugget parameter. We also compare with other supervised machine learning methods: \begin{enumerate}
\item K-nearest neighbors (KNN), where the number of neighbors is the tuning parameter.

\item Random forest, implemented in the \texttt{randomForest} package \citep{randomForest}. We choose the maximum depth and the minimum number of samples per node as the tuning parameters.

\item Extreme gradient boosting (XGBoost), implemented in the \texttt{xgboost} package \citep{xgboost}. Similar to random forest, the tuning parameters are the maximum depth and minimum node weight.

\item Shallow neural networks. We use \texttt{nnet} \citep{nnet} to build neural networks containing a single hidden layer. The tuning parameters are the hidden-layer size and the weight decay. We also implement a Bayesian neural network (BNN), for which we tune the hidden-layer size and the prior precision on the weights.
\end{enumerate}
For the concrete strength and meat spectroscopy datasets, the samples are randomly split into 90\% for training and 10\% for testing. The appliance energy dataset is split by date, so that samples from the same date are either both in the training set or both in the test set. During training of the above methods with tunable hyperparameters, we perform light tuning (fewer than 20 combinations) via grid search for each method, using 10-fold cross-validation on the training set. 

\begin{figure}[p]
\centering

\begin{subfigure}{\textwidth}
\centering
\includegraphics[width=0.32\textwidth]{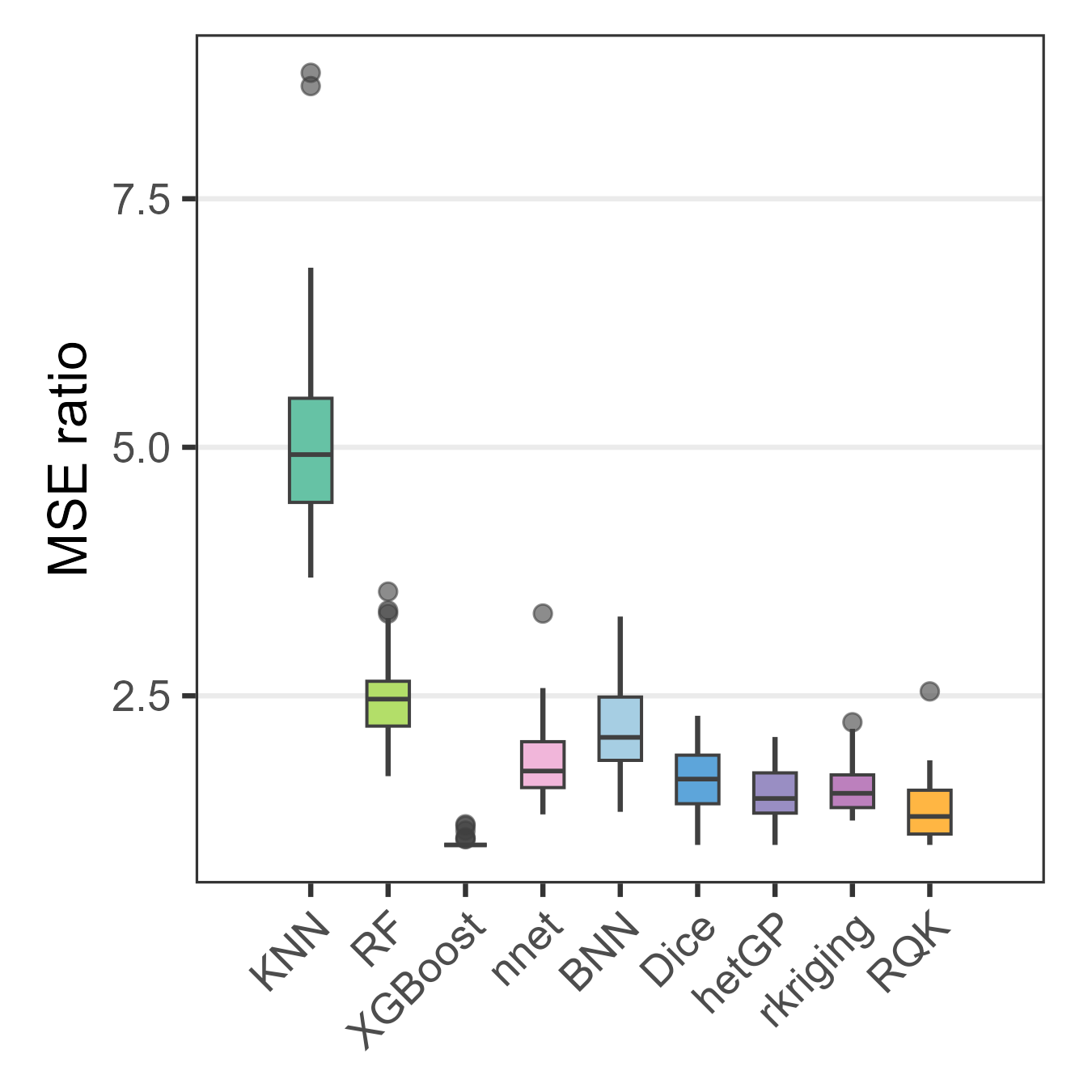}
\includegraphics[width=0.32\textwidth]{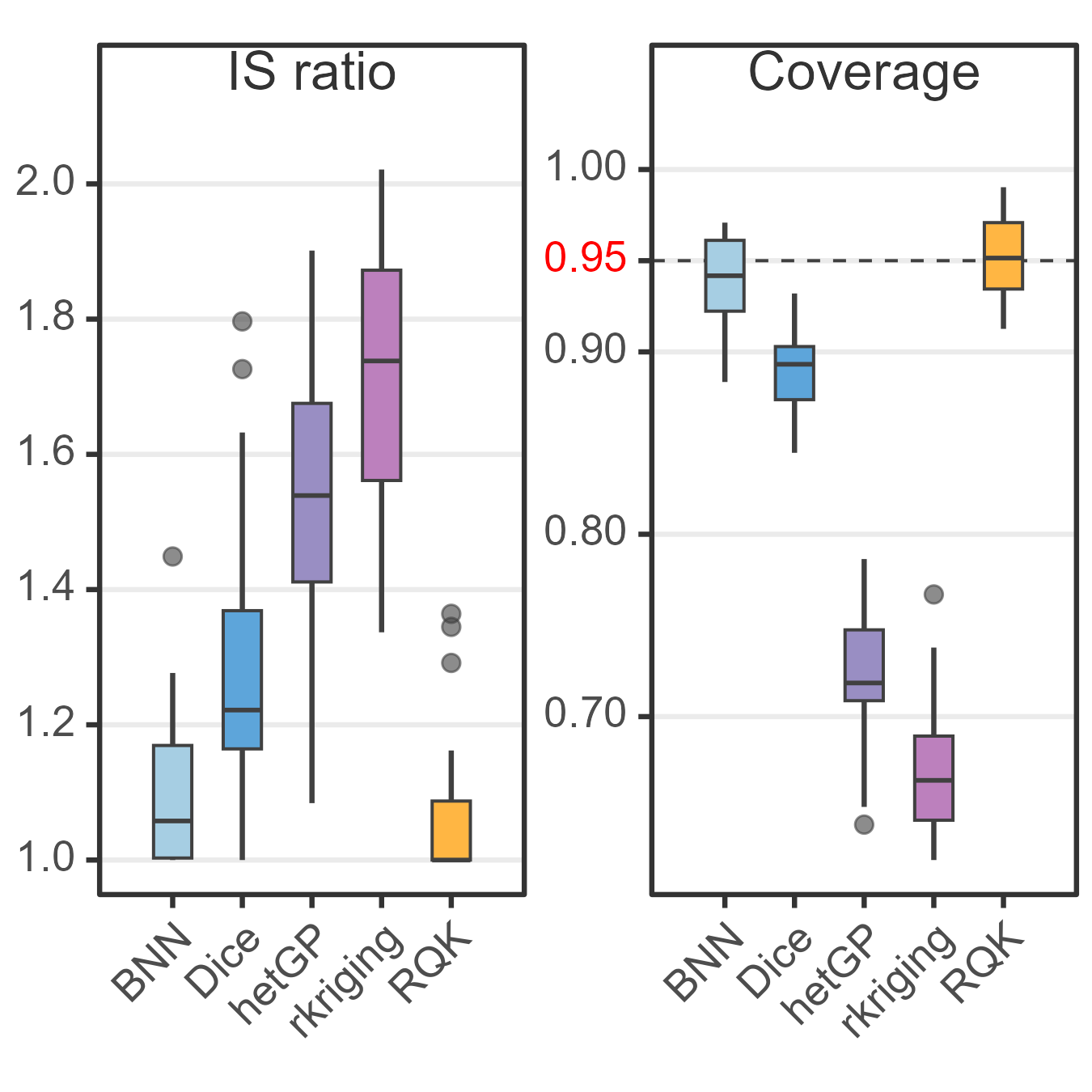}
\includegraphics[width=0.32\textwidth]{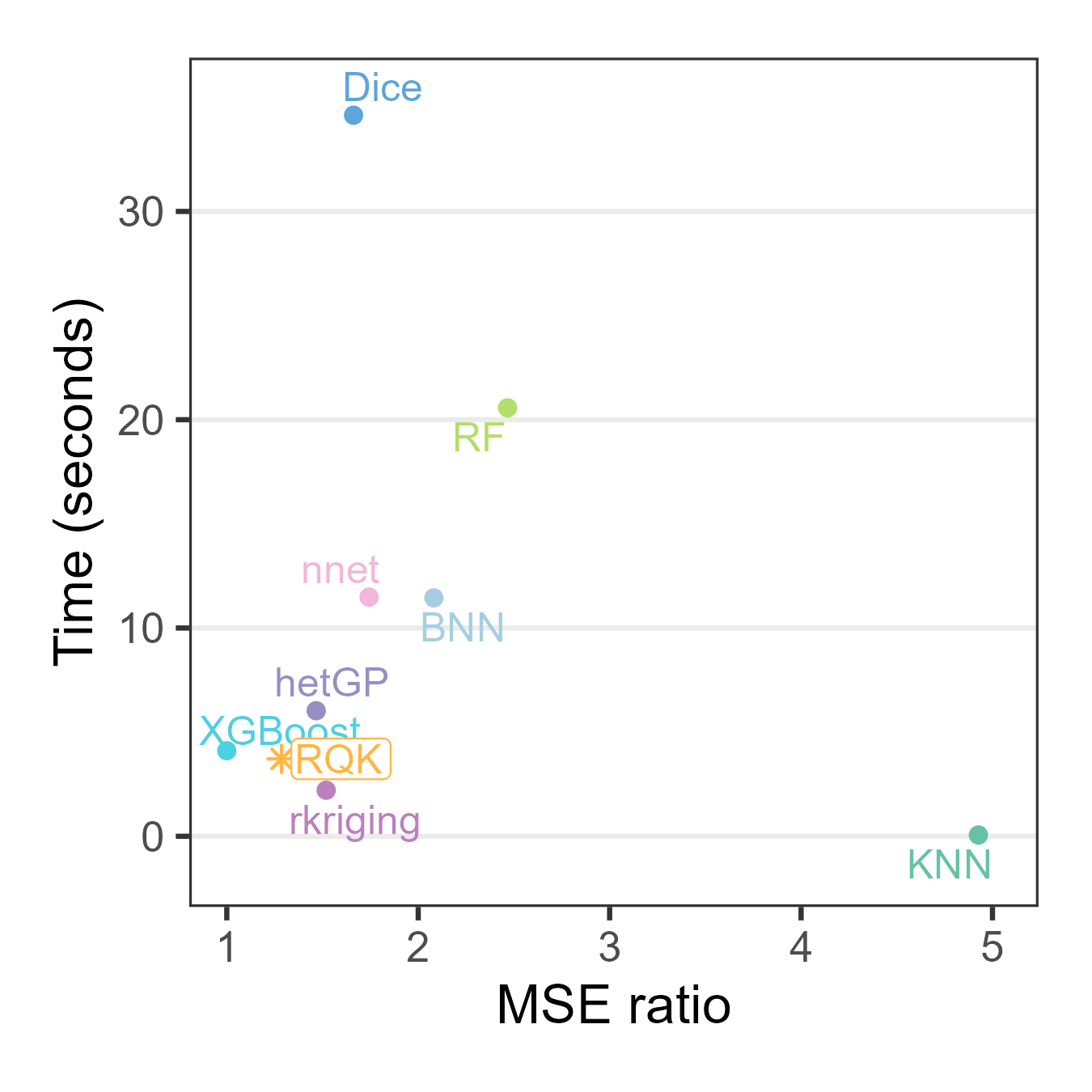}
\caption{Concrete strength dataset.}
\end{subfigure}%

\begin{subfigure}{\textwidth}
\centering
\includegraphics[width=0.32\textwidth]{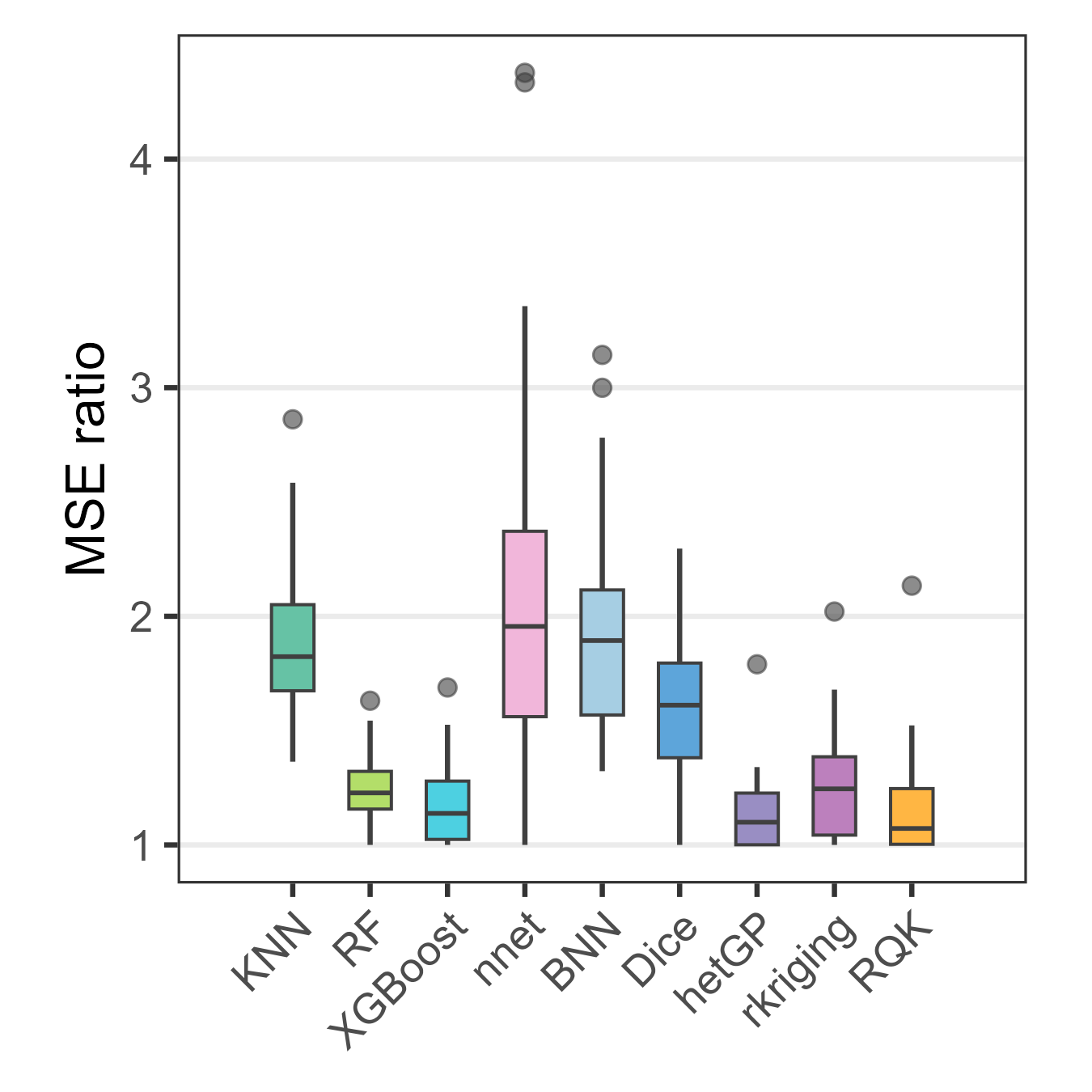}
\includegraphics[width=0.32\textwidth]{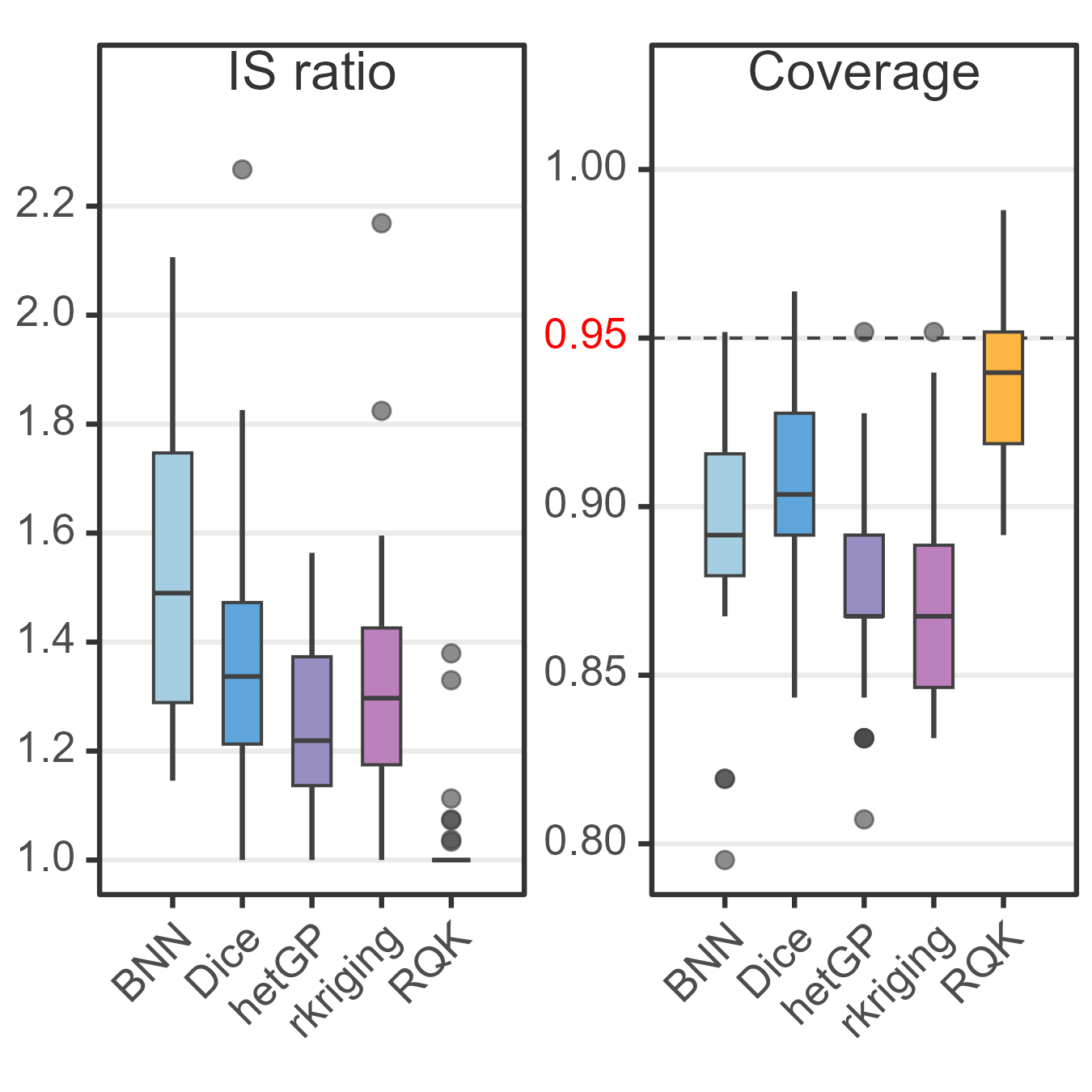}
\includegraphics[width=0.32\textwidth]{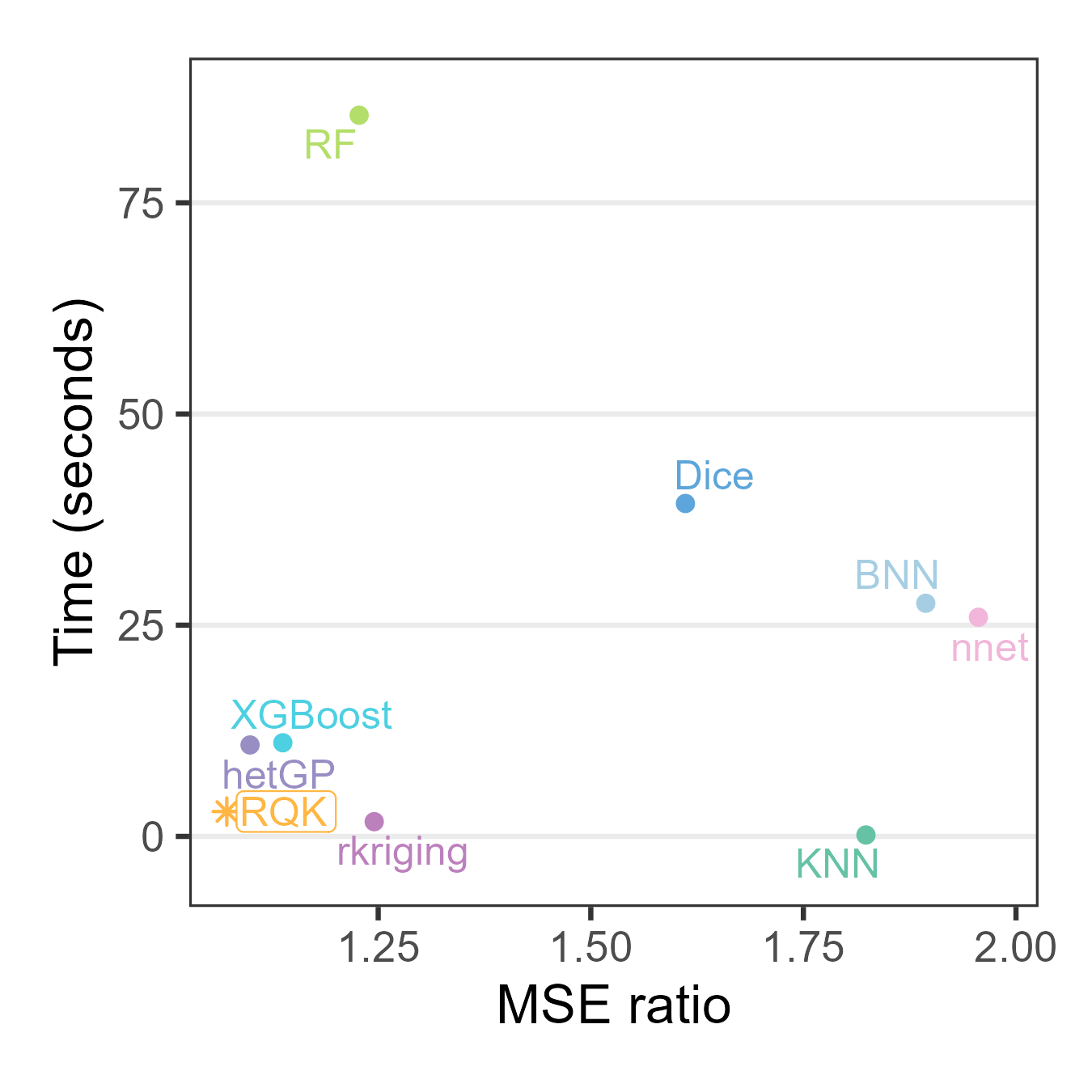}
\caption{Appliance energy prediction dataset.}
\end{subfigure}%

\begin{subfigure}{\textwidth}
\centering
\includegraphics[width=0.32\textwidth]{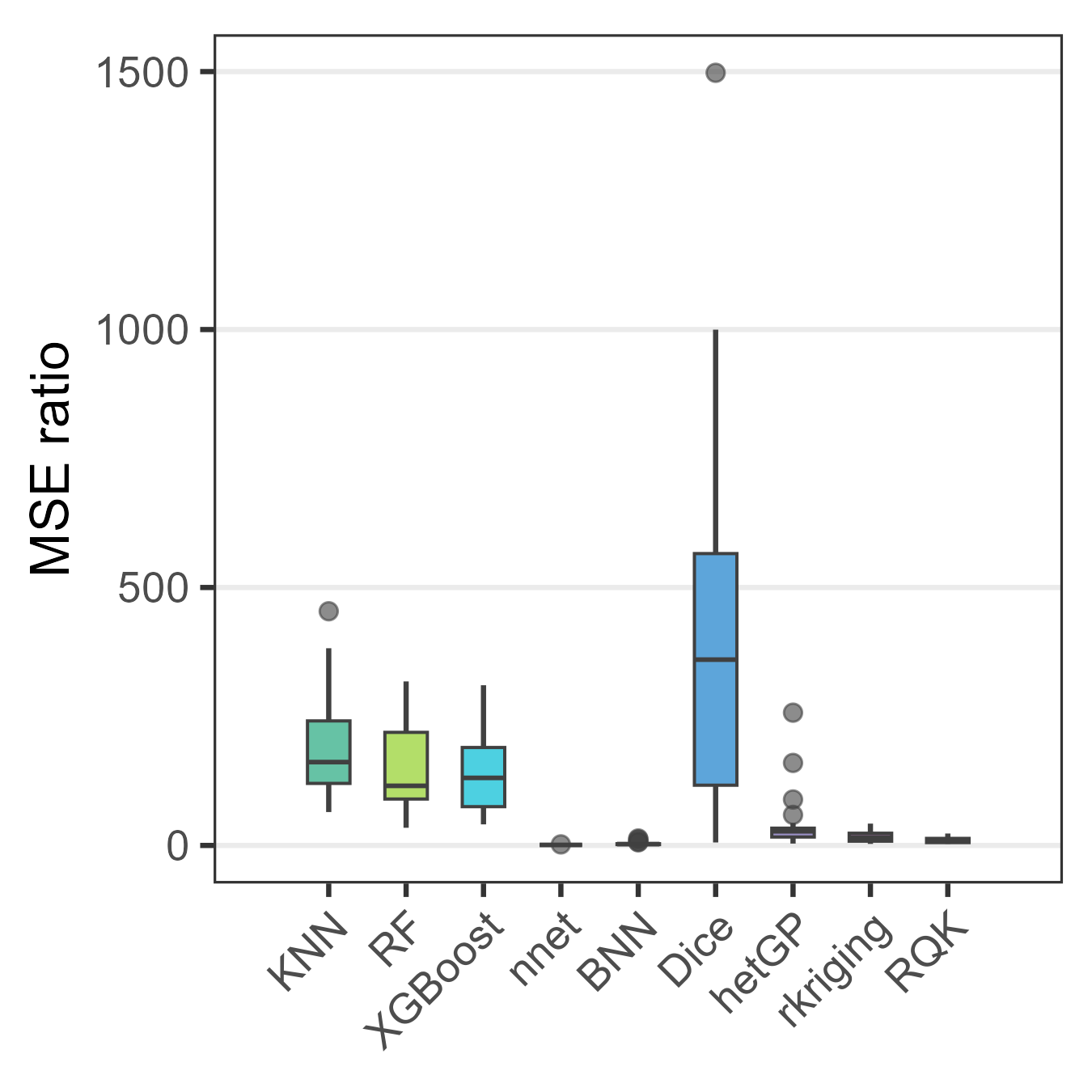}
\includegraphics[width=0.32\textwidth]{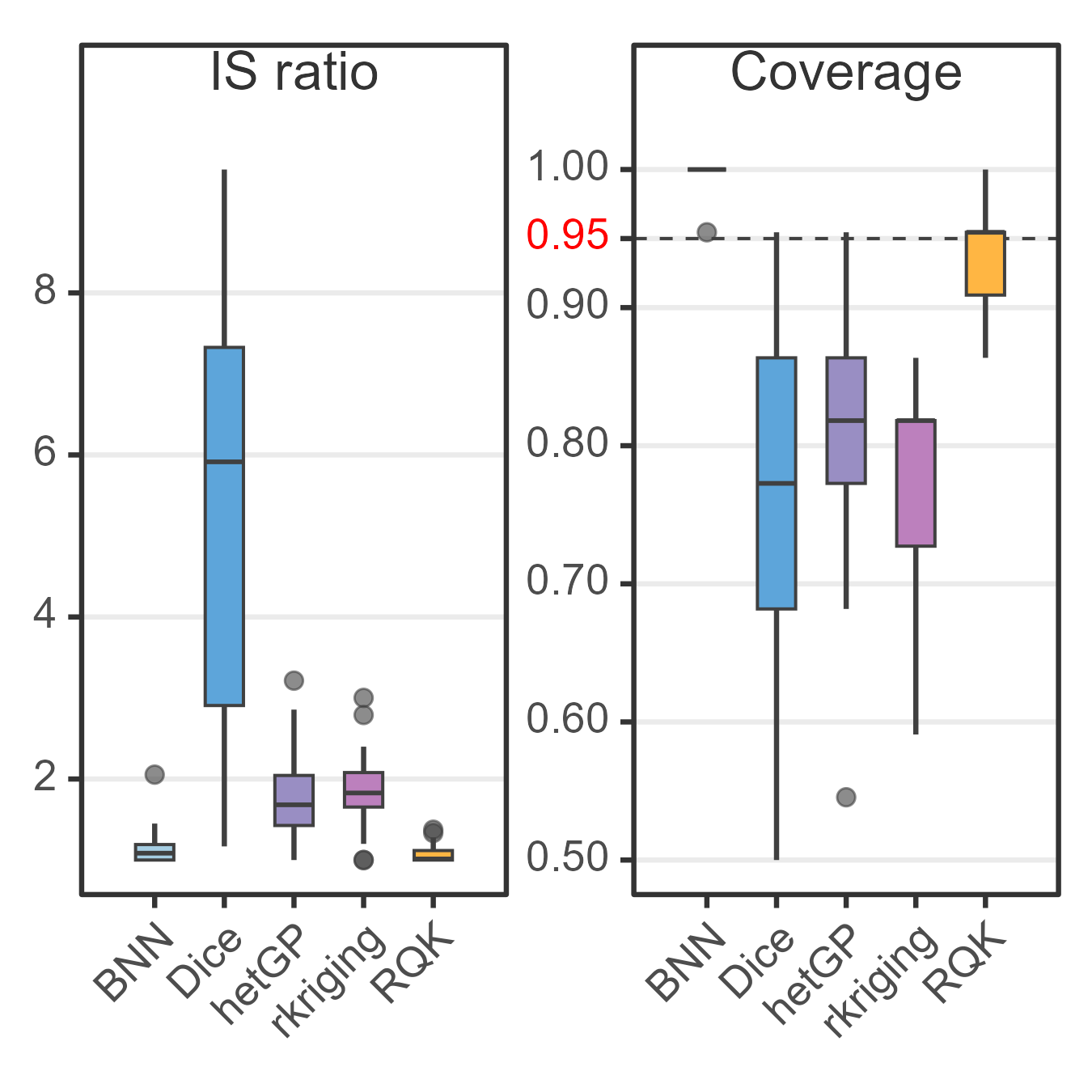}
\includegraphics[width=0.32\textwidth]{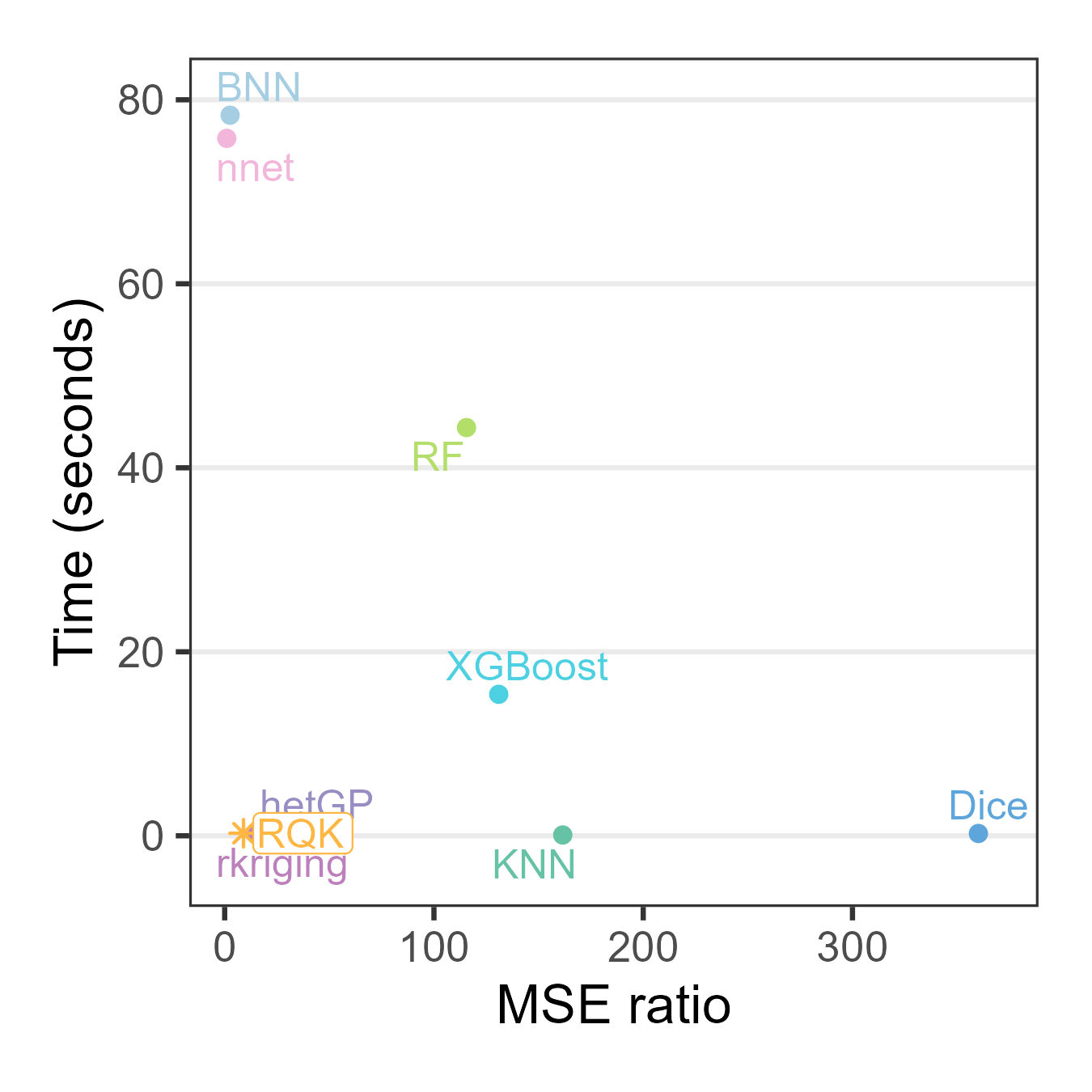}
\caption{Meat spectroscopy dataset.}
\end{subfigure}%

\caption{MSE (left), IS (middle), and computational time (right) comparisons for three noisy datasets. Each experiment is repeated over 30 random seeds for train-test splits.}
\label{fig:noisy}

\end{figure}

We repeat the fitting procedure 30 times using random train-test splits. Because the absolute value of the test error is sensitive to the specific split, we present the \textit{ratios} for the two metrics, MSE and IS: \begin{align}
    \frac{\text{metric}}{\text{best metric over all methods for the split}},
\end{align} so that closer to 1.0 indicates best performance across all methods. The left panels of Figure \ref{fig:noisy} show the MSE ratios. XGBoost performs well on the concrete and appliance energy datasets but fails on the 100-dimensional meat spectroscopy dataset. On the other hand, shallow neural networks have large errors for concrete and appliance energy but perform best for meat spectroscopy. RQK provides stable performance across all three datasets, either achieving the smallest error (appliance energy) or being close to the best-performing method (concrete and meat spectroscopy).

The middle panels present the uncertainty quantification performance for all kriging-based methods and the Bayesian neural network. Among these methods, RQK is the only one to achieve the desired coverage ratio through leave-one-out conformal prediction. Other methods generally lack coverage guarantees and often experience severe under-coverage. With correct coverage, it is unsurprising that RQK has the smallest IS.

Finally, the panels on the right plot the median computation time against the MSE ratio. We can see that the computation times of tree-based methods and neural networks grow quickly in $d$. In contrast, kriging-based methods can handle large $d$ well, but scale poorly with $n$, explaining the short computational time for the meat spectroscopy dataset ($n = 215, d = 100$). Nevertheless, RQK is always near the bottom-left of the figures, representing the Pareto front for trade-offs between efficiency and accuracy. We acknowledge that, with more extensive hyperparameter tuning, XGBoost achieves state-of-the-art performance on tabular data \citep{Shwartz-Ziv_Armon_2021}. Larger and deeper neural networks have also achieved success in predictive modeling across many real-world datasets. We are not claiming RQK is universally optimal. The take-home message is that, as an out-of-the-box machine learning method that requires minimal tuning, RQK provides both stable predictions and uncertainty quantification with limited computational resources. Therefore, we believe it will be useful for many practical learning and surrogate modeling tasks.

\section{Conclusion} \label{sec:conclusion}

This article proposes rational quadratic kriging as an improvement to existing kriging approaches. By making specific modeling choices and estimation procedures, we develop an estimator that is not only theoretically appealing for both estimation and uncertainty quantification but also numerically stable in implementation. Across a range of simulations with both deterministic functions and real-world noisy datasets, the method demonstrates reliable performance with a low computational burden.

We remark on the size of the problems appropriate for the method. Because we use a leave-one-out conformal prediction procedure for finding prediction intervals, the sample size should be reasonably large for reliable coverage. From our simulations, we find the rule of thumb $n \geq 10/\eta$ to be useful. Intuitively, the rule states that at least $10$ samples are in the tail to estimate the corresponding quantile bounds reliably. On the other hand, the approach still relies on $\mathcal{O}(n^3)$ computations. Although we never explicitly compute the matrix inverse $\mathbf{R}^{-1}$, the Cholesky factorization and linear solves still become expensive for big data. Thus, we recommend $n$ to be on the order of thousands. Extending the method to handle larger datasets using frameworks such as that described by \citet{Vakayil_Joseph_2024} will be a future direction.

Another aspect of the method that requires further understanding is the performance on high-dimensional data. Although we showed promising results on examples with $d$ ranging from $20$ to $100$, there is a need for theoretical justification for the performance on sparse and dense high-dimensional problems.

Finally, the rational quadratic kernel is a very specific choice that works well for smooth functions. However, other correlation kernels may be more appropriate for particular science and engineering problems with a known structure. Theoretical properties and empirical performance of more general kernel specifications in the rational kriging framework remain an interesting topic.

\vspace{.25in}

\FloatBarrier

\begin{center}
{\Large\bf Supplementary Materials}
\end{center}

\noindent The proofs of all theoretical results are provided in the supplementary.

\vspace{.25in}

\if0\blind{
\begin{center}
{\Large\bf Acknowledgments}
\end{center}

\noindent This research is supported by U.S. National Science Foundation grant DMS-2310637.

\vspace{.25in}
} \fi

\appendix

\setcounter{equation}{0}
\renewcommand{\theequation}{A.\arabic{equation}}
\renewcommand{\theHequation}{app.\theequation}

\setcounter{lemma}{0}
\renewcommand{\thelemma}{A.\arabic{lemma}}
\renewcommand{\theHlemma}{app.\thelemma}

\section{Proofs of theoretical results}

\subsection{Proof of Proposition 1}

\begin{proof}
Because $\mathbf{x} \neq \{\mathbf{x}_i\}_{i = 1}^n$, \begin{align}
    R(\mathbf{x}, \mathbf{x}_i) = \left(1 + \frac{\|\mathbf{x} - \mathbf{x}_i\|_{\mathbf{w}}^2}{\theta^2}\right)^{-p} = \theta^{2p} \|\mathbf{x} - \mathbf{x}_i\|_{\mathbf{w}}^{-2p}\{1 + o(1)\},
\end{align} when $\theta \to 0$. Likewise, \begin{align}
    \sum_{j \neq i} R(\mathbf{x}_i, \mathbf{x}_j) &= \theta^{2p} \sum_{j \neq i} \|\mathbf{x}_i - \mathbf{x}_j\|_{\mathbf{w}}^{-2p} \{1 + o(1)\}, \\
    \quad c_i &= \theta^{-2p} \left(\sum_{j \neq i} \|\mathbf{x}_i - \mathbf{x}_j\|_{\mathbf{w}}^{-2p}\right)^{-1} \{1 + o(1)\}.
\end{align}

Hence \begin{align}
    R(\mathbf{x}, \mathbf{x}_i) c_i = \frac{\|\mathbf{x} - \mathbf{x}_i\|_{\mathbf{w}}^{-2p}}{\sum_{j \neq i} \|\mathbf{x}_i - \mathbf{x}_j\|_{\mathbf{w}}^{-2p}} \{1 + o(1)\}. \label{eqn:rc_limit}
\end{align}

Since all off-diagonal entries of $\mathbf{R}$ vanish as $\theta \to 0$, we have $\mathbf{R} \to \mathbf{I}$ and $\mathbf{R}^{-1} \to \mathbf{I}$. In addition, \begin{align}
    [\mathbf{R} \mathbf{c}]_i = c_i + \sum_{j \neq i} R(\mathbf{x}_i, \mathbf{x}_j) c_j = c_i\{1 + o(1)\}
\end{align} because $c_i$ is $\mathcal{O}(\theta^{-2p})$ and the second term is $\mathcal{O}(1)$. Therefore, \begin{align}
    \mathbf{r}(\mathbf{x})^\top \mathbf{R}^{-1} \mathrm{diag}(\mathbf{R} \mathbf{c}) \mathbf{y} &= \sum_{i = 1}^n R(\mathbf{x}, \mathbf{x}_i) c_i y_i \{1 + o(1)\}, \\
    &\mathbf{r}(\mathbf{x})^\top \mathbf{c} = \sum_{i = 1}^n R(\mathbf{x}, \mathbf{x}_i) c_i.
\end{align}
Taking the ratio proves the result for $\hat m_{\mathrm{RQK}}(\mathbf{x})$.

Also, since $\mathbf{r}(\mathbf{x}) \to 0$, \begin{align}
    \lim_{\theta \to 0} \mathbf{r}(\mathbf{x})^\top \mathbf{R}^{-1} \mathbf{r}(\mathbf{x}) = 0.
\end{align} Combining with \eqref{eqn:rc_limit} shows that \begin{align}
    \lim_{\theta \to 0} s^2_{\mathrm{RQK}} = \nu^2 \left\{\sum_{i = 1}^n \frac{\|\mathbf{x} - \mathbf{x}_i\|_{\mathbf{w}}^{-2p}}{\sum_{j \neq i} \|\mathbf{x}_i - \mathbf{x}_j\|_{\mathbf{w}}^{-2p}}\right\}^{-2}
\end{align} and the value is $\mathcal{O}(1)$, determined by the weighted distances and hyperparameter $p$.
\end{proof}

\subsection{Proof of Lemma 1}

\begin{proof}
For a fixed $\mathbf{x} \in \mathcal{X}$. For simplicity, write $\mathbf{r} = \mathbf{r}(\mathbf{x})$. The difference between RQK weights and kernel regression weights is \begin{align} \label{eqn:RQK_KR_diff}
    \mathbf{a}_{\mathrm{RQK}}(\mathbf{x}) - \mathbf{a}_{\mathrm{KR}}(\mathbf{x}) = \frac{(\mathbf{r}^\top \mathbf{R}^{-1} - \mathbf{r}^\top) \mathbf{R} \mathbf{c}}{\mathbf{r}^\top \mathbf{c}} + \frac{\mathbf{r}^\top (\mathbf{R} \mathbf{c} - \mathbf{c})}{\mathbf{r}^\top \mathbf{c}}.
\end{align}

Let $\mathbf{E} := \mathbf{R} - \mathbf{I}$. The matrix infinity norm satisfies \begin{align}
    \|\mathbf{E}\|_{\infty} = \max_{1 \leq i \leq n} \sum_{j \neq i} R(\mathbf{x}_i, \mathbf{x}_j) = \zeta_n(\theta) < 1.
\end{align} Therefore, the Neumann series gives \begin{align}
    \mathbf{R}^{-1} = (\mathbf{I} + \mathbf{E})^{-1} = \mathbf{I} + \sum_{m = 1}^\infty (-\mathbf{E})^m,
\end{align} hence \begin{align}
    \|\mathbf{R}^{-1} - \mathbf{I}\|_\infty \leq \sum_{m = 1}^\infty \|\mathbf{E}\|_\infty^m \leq \frac{\zeta_n(\theta)}{1 - \zeta_n(\theta)}.
\end{align} Also, \begin{align}
    \|\mathbf{R} \mathbf{c} - \mathbf{c}\|_{\infty} = \|\mathbf{E} \mathbf{c}\|_{\infty} \leq \|\mathbf{E}\|_{\infty} \|\mathbf{c}\|_{\infty} \leq \zeta_n(\theta) \max_{1 \leq i \leq n} c_i.
\end{align} Finally, \begin{align}
    \mathbf{r}^\top \mathbf{c} \geq \min_{1 \leq i \leq n} c_i \sum_{j = 1}^n R(\mathbf{x}, \mathbf{x}_j).
\end{align}

We now bound the two terms of the right-hand side of \eqref{eqn:RQK_KR_diff} separately. First term: \begin{align}
    \left\|\frac{(\mathbf{r}^\top \mathbf{R}^{-1} - \mathbf{r}^\top) \mathbf{R} \mathbf{c}}{\mathbf{r}^\top \mathbf{c}}\right\|_1 &\leq \frac{\|\mathbf{r}^\top (\mathbf{R}^{-1} - \mathbf{I})\|_1 \|\mathbf{R} \mathbf{c}\|_\infty}{\mathbf{r}^\top \mathbf{c}} \\
    &\leq \frac{\|\mathbf{R}^{-1} - \mathbf{I}\|_\infty \sum_{j = 1}^n R(\mathbf{x}, \mathbf{x}_j) (\|\mathbf{c}\|_{\infty} + \|\mathbf{E} \mathbf{c}\|_\infty)}{\min_{1 \leq i \leq n} c_i \sum_{j = 1}^n R(\mathbf{x}, \mathbf{x}_j)} \\
    &\leq \frac{\zeta_n(\theta)}{1 - \zeta_n(\theta)} \frac{1 + \zeta_n(\theta) \max_{1 \leq i \leq n} c_i}{\min_{1 \leq i \leq n} c_i} \\
    &\leq C_1 \zeta_n(\theta),
\end{align} where $C_1$ only depends on $A_n(\theta)$ and $\zeta_n(\theta)$.

Second term: \begin{align}
    \left\|\frac{\mathbf{r}^\top (\mathbf{R} \mathbf{c} - \mathbf{c})}{\mathbf{r}^\top \mathbf{c}}\right\|_1 &\leq \frac{\sum_{j = 1}^n R(\mathbf{x}, \mathbf{x}_j) \|\mathbf{R} \mathbf{c} - \mathbf{c}\|_\infty}{\mathbf{r}^\top \mathbf{c}} \leq \zeta_n(\theta) \frac{\max_{1 \leq i \leq n} c_i}{\min_{1 \leq i \leq n} c_i} \leq A_n(\theta) \zeta_n(\theta).
\end{align}

Combining both inequalities proves \begin{align}
    \|\mathbf{a}_{\mathrm{RQK}}(\mathbf{x}) - \mathbf{a}_{\mathrm{KR}}(\mathbf{x})\|_1 \leq C \zeta_n(\theta).
\end{align}

The predictor bound immediately follows from \begin{align}
    |\hat m_{\mathrm{RQK}}(\mathbf{x}) - \hat m_{\mathrm{KR}}(\mathbf{x})| = \left|\left[\mathbf{a}_{\mathrm{RQK}}(\mathbf{x}) - \mathbf{a}_{\mathrm{KR}}(\mathbf{x})\right]^\top \mathbf{y}\right| \leq \|\mathbf{a}_{\mathrm{RQK}}(\mathbf{x}) - \mathbf{a}_{\mathrm{KR}}(\mathbf{x})\|_1 \|\mathbf{y}\|_{\infty}.
\end{align}

To get the uniform bound over $0 < \theta \leq \bar \theta$, directly use $A_n(\theta) \leq \bar A$ and $\zeta_n(\theta) \leq \zeta_n(\bar \theta)$ over all $0 < \theta \leq \bar \theta$.

\end{proof}

\subsection{Proof of Theorem 1}

\subsubsection{Additional lemmas for Theorem 1}

Having associated the RQK predictor with the kernel regression estimator, we further connect it to a weighted IDW estimator through the next lemma.

\begin{lemma}[Kernel regression converges to IDW] \label{lemma:KR_IDW}
Suppose $\{(\mathbf{X}_i, Y_i)\}_{i \geq 1}$ are i.i.d. random vectors, where $\mathbf{X}$ has a bounded density $f_{\mathbf{X}}$. Fix $\mathbf{x} \in \mathcal{X}$ such that $f_{\mathbf{X}}(\mathbf{x}) > 0$ and $\mathbb{P}\{\mathbf{X} = \mathbf{x}\} = 0$. Allow the lengthscale to depend on the sample size $n$, i.e., $\theta = \theta_n$. Define \begin{align}
    \hat m_{\mathrm{KR}}(\mathbf{x}) &= \sum_{i = 1}^n \frac{R(\mathbf{x}, \mathbf{X}_i) c_i}{\sum_{j = 1}^n R(\mathbf{x}, \mathbf{X}_j) c_j} Y_i, \label{eqn:KR_random} \\
    \hat m_{\mathrm{IDW}}(\mathbf{x}) &= \frac{\sum_{i = 1}^n c_i \|\mathbf{x} - \mathbf{X}_i\|_{\mathbf{w}}^{-2p} Y_i}{\sum_{j = 1}^n c_j \|\mathbf{x} - \mathbf{X}_j\|_{\mathbf{w}}^{-2p}}. \label{eqn:Weighted_IDW}
\end{align} where \eqref{eqn:KR_random} is the same as \eqref{eqn:KR} except that $\mathbf{X}_i$ is used to reflect randomness.

If $n \theta_n^d \to 0$, we have \begin{align}
    \hat m_{\mathrm{KR}}(\mathbf{x}) - \hat m_{\mathrm{IDW}}(\mathbf{x}) \overset{p}{\to} 0
\end{align} as $n \to \infty$.

If, in addition, $\sum_{n = 1}^\infty n\theta_n^d < \infty$, the convergence holds almost surely.
\end{lemma}

\begin{proof}
Choose $M_n := (n \theta_n^d)^{-1/(2d)}$, so that $M_n \to \infty$. Because $f_{\mathbf{X}}$ is bounded from above, \begin{align}
    \mathbb{P}\left\{\min_{1 \leq i \leq n} \|\mathbf{x} - \mathbf{X}_i\| \leq M_n \theta_n\right\} \leq Cn (M_n \theta_n)^d = C(n\theta_n^d)^{1/2} \to 0.
\end{align} Hence, with probability tending to 1, \begin{align}
    \frac{\|\mathbf{x} - \mathbf{X}_i\|_{\mathbf{w}}}{\theta_n} \geq M_n \text{ for all } i = 1, \dots, n.
\end{align}
On this event, \begin{align}
    R(\mathbf{x}, \mathbf{X}_i) = \left(1 + \frac{\|\mathbf{x} - \mathbf{X}_i\|_{\mathbf{w}}^2}{\theta_n^2}\right)^{-p} = \theta_n^{2p} \|\mathbf{x} - \mathbf{X}_i\|_{\mathbf{w}}^{-2p}(1 + \varepsilon_{i, n}(\mathbf{x})), \label{eqn:corr}
\end{align} with \begin{align}
    \max_{1 \leq i \leq n} |\varepsilon_{i, n}(\mathbf{x})| \to 0
\end{align}
Substituting \eqref{eqn:corr} into the definition of $\hat m_{\mathrm{KR}}(\mathbf{x})$, after cancelling out the common factor $\theta_n^{2p}$, the ratio differs from $\hat m_{\mathrm{IDW}}(\mathbf{x})$ by $o_p(1)$.

If $\sum_{i = 1}^n n \theta_n^d$, for every fixed $M > 0$, \begin{align}
    \sum_{n = 1}^\infty \mathbb{P}\left(\min_{1 \leq i \leq n} \|\mathbf{x} - \mathbf{X}_i\|_{\mathbf{w}} \leq M \theta_n\right) \leq C M^d \sum_{n = 1}^\infty n \theta_n^d < 0.
\end{align} By the Borel-Cantelli lemma, \begin{align}
    \frac{\min_{1 \leq i \leq n} \|\mathbf{x} - \mathbf{X}_i\|_{\mathbf{w}}}{\theta_n} \overset{a.s.}{\to} \infty.
\end{align} By the same argument, the convergence now holds almost surely.

\end{proof}

\begin{lemma}[Weak convergence of IDW]
\label{lemma:Weighted_IDW}
Suppose $\{(\mathbf{X}_i, Y_i)\}_{i \geq 1}$ are i.i.d. random vectors, where $\mathbf{X}$ has a density $f_{\mathbf{X}}$, and $Y$ is bounded. Let \begin{align}
    m(\mathbf{x}) := \mathbb{E}[Y | \mathbf{X} = \mathbf{x}].
\end{align}

In the weighted IDW estimator \eqref{eqn:Weighted_IDW}, let $p = d/2$, $c_i = c_i^{(n)}$ be positive random variables measurable with respect to $\sigma(\mathbf{X}_1, \dots, \mathbf{X}_n)$ and assume \begin{align}
    \frac{\max_{1 \leq i \leq n} c_i^{(n)}}{\min_{1 \leq i \leq n} c_i^{(n)}} = \mathcal{O}_p(1).
\end{align} Then as $n \to \infty$, \begin{enumerate}
    \item at almost all $\mathbf{x} \in \mathcal{X}$ with $f_{\mathbf{X}}(\mathbf{x}) > 0$, $\hat m_{\mathrm{IDW}}(\mathbf{x}) \overset{p}{\to} m(\mathbf{x})$;

    \item $\displaystyle \int_{\mathcal{X}} \left|\hat m_{\mathrm{IDW}}(\mathbf{x}) - m(\mathbf{x})\right| d \mu_{\mathbf{X}}(\mathbf{x}) \overset{p}{\to} 0$.
\end{enumerate}
\end{lemma}

\begin{proof}
Define the event \begin{align}
    E_{n, M} := \left\{\frac{\max_{1 \leq i \leq n} c_i^{(n)}}{\min_{1 \leq i \leq n} c_i^{(n)}} \leq M\right\}.
\end{align} Since $\frac{\max_{1 \leq i \leq n} c_i^{(n)}}{\min_{1 \leq i \leq n} c_i^{(n)}} = \mathcal{O}_p(1)$, for every $\xi > 0$ we can choose $M$ large enough so that \begin{align}
    \limsup_{n \to \infty} \mathbb{P}[E^c_{n, M}] < \xi.
\end{align}
On the event $E_{n, M}$, define $b_i^{(n)} = c_i^{(n)}/\min_{1 \leq i \leq n} c_i^{(n)}$, so that $1 \leq b_i^{(n)} \leq M$ for all $i$. The estimator can be rewritten as \begin{align}
    \hat m_{\mathrm{IDW}}(\mathbf{x}) := \frac{\sum_{i = 1}^n b_i^{(n)} \|\mathbf{x} - \mathbf{X}_i\|_{\mathbf{w}}^{-d} Y_i}{\sum_{j = 1}^n b_j^{(n)} \|\mathbf{x} - \mathbf{X}_j\|_{\mathbf{w}}^{-d}}.
\end{align}

Fix a point $\mathbf{x}$ that is a Lebesgue point of both $f_{\mathbf{X}}(\cdot)$ and $m(\cdot) f_{\mathbf{X}}(\cdot)$, and $f_{\mathbf{X}}(\mathbf{x}) > 0$. Let \begin{align}
    W_i := V_{d, \mathbf{w}} \|\mathbf{x} - \mathbf{X}_i\|_{\mathbf{w}}^d,
\end{align} where $V_{d, \mathbf{w}}$ is the weighted unit ball. The error can be decomposed into \begin{align}
    \hat m_{\mathrm{IDW}}(\mathbf{x}) - m(\mathbf{x}) = \frac{\sum_{i = 1}^n b_i^{(n)} [Y_i - m(\mathbf{X}_i)]/W_i}{\sum_{j = 1}^n b_j^{(n)}/W_j} + \frac{\sum_{i = 1}^n b_i^{(n)} [m(\mathbf{X}_i) - m(\mathbf{x})]/W_i}{\sum_{j = 1}^n b_j^{(n)}/W_i} =: I(\mathbf{x}) + II(\mathbf{x}).
\end{align}

For the first term: \begin{align}
    \mathbb{E}[I(\mathbf{x})^2 | \mathbf{X}_1, \dots, \mathbf{X}_{n}] \leq B^2 \sum_{i = 1}^n \left(\frac{b_i^{(n)}/W_i}{\sum_{j = 1}^n b_j^{(n)}/W_j}\right)^2 \leq B^2 M \max_{1 \leq i \leq n} \frac{1/W_i}{\sum_{j = 1}^n 1/W_j},
\end{align} where $|Y| \leq B$ from the bounded assumption. Lemma 2 of \citet{Devroye_Gyorfi_Krzyzak_1998} proves that $\max_{1 \leq i \leq n} \frac{1/W_i}{\sum_{j = 1}^n 1/W_j} \overset{p}{\to} 0$, so $I(\mathbf{x}) \overset{p}{\to} 0$ on the event $E_{n, M}$.

The second term \begin{align}
    |II(\mathbf{x})| \leq M \frac{\sum_{i = 1}^n |m(\mathbf{X}_i) - m(\mathbf{x})|/W_i}{\sum_{j = 1}^n 1/W_j},
\end{align} which also converges to zero in probability by \citet{Devroye_Gyorfi_Krzyzak_1998}.

Thus, $\hat m_{\mathrm{IDW}}(\mathbf{x}) \to m(\mathbf{x})$ holds on the event $E_{n, M}$. Unconditionally, for every $\varepsilon > 0$, \begin{align}
    \mathbb{P}[|\hat m_{\mathrm{IDW}}(\mathbf{x}) - m(\mathbf{x})| > \varepsilon] \leq \mathbb{P}[E_{n, M}^c] + \mathbb{P}[|\hat m_{\mathrm{IDW}}(\mathbf{x}) - m(\mathbf{x})| > \varepsilon, E_{n, m}].
\end{align} The first term can be made arbitrarily small by choosing a large $M$, and the second term goes to zero for fixed $M$. This proves pointwise convergence in probability.

For $L^1$ convergence, we use a standard argument from \citet{Devroye_1981}. Let \begin{align}
    Z_n(\mathbf{x}) := |\hat m_{\mathrm{IDW}}(\mathbf{x}) - m(\mathbf{x})|.
\end{align} Since $Y$ is bounded, $0 \leq Z_n(\mathbf{x}) \leq 2B$. For a fixed $\xi > 0$, \begin{align}
    \mathbb{E}[Z_n(\mathbf{x})] \leq \xi + 2B \mathbb{P}(Z_n(\mathbf{x}) > \xi).
\end{align} With $Z_n(\mathbf{x}) \overset{p}{\to} 0$, the inequality implies \begin{align}
    \limsup_{n \to \infty} \mathbb{E}[Z_n(\mathbf{x})] \leq \xi.
\end{align} Since $\xi > 0$ is arbitrary, we conclude $\mathbb{E}[Z_n(\mathbf{x})] \to 0$ at almost all $\mathbf{x}$. Furthermore, by the dominated convergence theorem, \begin{align}
    \mathbb{E}\left[\int_{\mathcal{X}} Z_n(\mathbf{x}) f_{\mathbf{X}}(\mathbf{x}) d\mathbf{x}\right] = \int_{\mathcal{X}}\mathbb{E}\left[ Z_n(\mathbf{x})\right] f_{\mathbf{X}}(\mathbf{x}) d\mathbf{x} \to 0
\end{align}

By Markov's inequality, for every $\varepsilon > 0$, \begin{align}
    \mathbb{P}\left[\int_{\mathcal{X}} Z_n(\mathbf{x}) f_{\mathbf{X}}(\mathbf{x}) d\mathbf{x} > \varepsilon \right] \leq \frac{1}{\varepsilon} \mathbb{E}\left[\int_\mathcal{X} Z_n(\mathbf{x}) f_{\mathbf{X}}(\mathbf{x}) d\mathbf{x}\right] \to 0,
\end{align} which allows us to conclude $L^1$ convergence.
\end{proof}

\subsubsection{Proof of Theorem 1}

\begin{proof}
Because $\zeta_n(\theta_n) \overset{p}{\to} 0$ and $\frac{\max_{1 \leq i \leq n} c_i}{\min_{1 \leq i \leq n} c_i} = \mathcal{O}_p(1)$, on the event \begin{align}
    E = \left\{\zeta_n(\theta_n) \leq 1/2, \frac{\max_{1 \leq i \leq n} c_i}{\min_{1 \leq i \leq n} c_i} \leq A\right\},
\end{align} Lemma \ref{lemma:RQK_KR} gives \begin{align}
    \sup_{\mathbf{x} \in \mathcal{X}} |\hat m_{\mathrm{RQK}}(\mathbf{x}) - \hat m_{\mathrm{KR}}(\mathbf{x})| \leq C(A, 1/2) \zeta_n(\theta_n) \|\mathbf{y}\|_\infty.
\end{align} Choose $A$ large enough so that $\mathbb{P}\left[\frac{\max_{1 \leq i \leq n} c_i}{\min_{1 \leq i \leq n} c_i} > A \right]$ is small, we have \begin{align}
    \hat m_{\mathrm{RQK}}(\mathbf{x}) \overset{p}{\to} \hat m_{\mathrm{KR}}(\mathbf{x})
\end{align} for each fixed $\mathbf{x}$.

Next, notice that a lower bound for $\zeta_n(\theta_n)$ is \begin{align}
    \zeta_n(\theta_n) \geq \left(1 + \frac{\min_{1 \leq i < j \leq n}\|\mathbf{x}_i - \mathbf{x}_j\|^2}{\theta_n^2}\right)^{-d}.
\end{align} Thus, $\zeta_n(\theta_n) \to 0$ implies $\frac{\theta_n}{\min_{1 \leq i < j \leq n}\|\mathbf{x}_i - \mathbf{x}_j\|} \to 0$.

A packing bound on compact $\mathcal{X}$ gives \begin{align}
    \min_{1 \leq i < j \leq n}\|\mathbf{x}_i - \mathbf{x}_j\| \leq C n^{1/d}.
\end{align} Furthermore,
\begin{align}
    n \theta_n^d = n \left(\frac{\theta_n}{\min_{1 \leq i < j \leq n}\|\mathbf{x}_i - \mathbf{x}_j\|}\right)^d \min_{1 \leq i < j \leq n}\|\mathbf{x}_i - \mathbf{x}_j\|^d
    \leq C^d \left(\frac{\theta_n}{\min_{1 \leq i < j \leq n}\|\mathbf{x}_i - \mathbf{x}_j\|}\right)^d,
\end{align} which implies $n \theta_n^d \to 0$. Under this condition, Lemma \ref{lemma:KR_IDW} gives \begin{align}
    \hat m_{\mathrm{KR}}(\mathbf{x}) \overset{p}{\to} \hat m_{\mathrm{IDW}}(\mathbf{x}).
\end{align}

Finally, Lemma \ref{lemma:Weighted_IDW} shows \begin{align}
    \hat m_{\mathrm{IDW}}(\mathbf{x}) \overset{p}{\to} m(\mathbf{x}) \quad \mu_{\mathbf{X}} \text{-a.e.}
\end{align} Combining the three approximation steps yields pointwise convergence.

$L^1$ convergence follows from $L^1$ convergence of $\hat m_{\mathrm{IDW}}(\mathbf{x})$ and the triangle inequality.
\end{proof}

\subsection{Proof of Theorem 2}

In this proof and later proofs, we use the notation $A_n \asymp B_n$ to indicate there exist constants $C_{\min}, C_{\max}$ that do not depend on $n$, such that \begin{align}
    C_{\min} B_n \leq A_n \leq C_{\max} B_n.
\end{align} We refer to this as ``$A_n$ is comparable to $B_n$''. We also use $A_n \lesssim B_n$ if the right-hand side holds.

\subsubsection{Additional lemma for Theorem 2}

\begin{lemma}[Strong convergence of IDW]
\label{lemma:IDW_continuity}
Suppose $\{\mathbf{X}_i\}_{i \geq 1}$ are i.i.d. random vectors, where $\mathbf{X}$ has a bounded density $f_{\mathbf{X}}$ on $\mathcal{X}$. $Y = m(\mathbf{X})$ for some continuous and bounded function $m: \mathcal{X} \to \mathbb{R}$. 

In the weighted IDW estimator \eqref{eqn:Weighted_IDW}, let $p \geq d/2$, $c_i = c_i^{(n)}$ be positive random variables measurable with respect to $\sigma(\mathbf{X}_1, \dots, \mathbf{X}_n)$ and assume \begin{align}
    \frac{\max_{1 \leq i \leq n} c_i^{(n)}}{\min_{1 \leq i \leq n} c_i^{(n)}} = \mathcal{O}_p(1).
\end{align} Then as $n \to \infty$, \begin{enumerate}
    \item at almost all $\mathbf{x} \in \mathcal{X}$ with $f_{\mathbf{X}}(\mathbf{x}) > 0$, $\hat m_{\mathrm{IDW}}(\mathbf{x}) \overset{p}{\to} m(\mathbf{x})$;

    \item $\displaystyle \int_{\mathcal{X}} \left|\hat m_{\mathrm{IDW}}(\mathbf{x}) - m(\mathbf{x})\right| d \mu_{\mathbf{X}}(\mathbf{x}) \overset{p}{\to} 0$.
\end{enumerate}

If, in addition, \begin{align}
   \sup_{n \geq 1} \frac{\max_{1 \leq i \leq n} c_i^{(n)}}{\min_{1 \leq i \leq n} c_i^{(n)}} < \infty \text{ almost surely},
\end{align} then
\begin{enumerate}
    \item at almost all $\mathbf{x} \in \mathcal{X}$ with $f_{\mathbf{X}}(\mathbf{x}) > 0$, $\hat m_{\mathrm{IDW}}(\mathbf{x}) \overset{a.s.}{\to} m(\mathbf{x})$;

    \item $\displaystyle \int_{\mathcal{X}} \left|\hat m_{\mathrm{IDW}}(\mathbf{x}) - m(\mathbf{x})\right| d \mu_{\mathbf{X}}(\mathbf{x}) \overset{a.s.}{\to} 0$.
\end{enumerate}
\end{lemma}

\begin{proof}
Let \begin{align}
    r_{(1),n}(\mathbf{x}) \leq \cdots \leq r_{(n),n}(\mathbf{x})
\end{align} be the ordered distances $\|\mathbf{x} - \mathbf{X}_i\|_{\mathbf{w}}$, and let $a_{(i), n}(\mathbf{x})$ denote the corresponding normalized weights of $\hat m_{\mathrm{IDW}}(\mathbf{x})$. For i.i.d. sampling with $f_{\mathbf{X}}(\mathbf{x}) > 0$, $r_{(i), n}$ is comparable to $(i/n)^{1/d}$. This also indicates that, for any fixed $i$ and $n \to \infty$, $r_{(i), n} \to 0$ almost surely. See \citet[Chapter 2]{Biau_2015} for more discussion.

Because the coefficients $c_i$ are positive and their ratio is bounded in probability (or almost surely), the unnormalized weights $c_i \|\mathbf{x} - \mathbf{X}_i\|_{\mathbf{w}}^{-2p}$ are comparable to $r_{(i),n}(\mathbf{x})^{-2p}$ in the same mode. This indicates for any $n_0$, the normalized weights satisfy \begin{align}
    \sum_{i > n_0} a_{(i), n}(\mathbf{x}) \lesssim \frac{\sum_{i > n_0} i^{-2p/d}}{\sum_{j = 1}^n j^{-2p/d}}, 
\end{align} on a high-probability event when $\frac{\max_{1 \leq i \leq n} c_i^{(n)}}{\min_{1 \leq i \leq n} c_i^{(n)}}$ is $\mathcal{O}_p(1)$, or eventually on a probability-one event when $\sup_{n \geq 1}\frac{\max_{1 \leq i \leq n} c_i^{(n)}}{\min_{1 \leq i \leq n} c_i^{(n)}} < \infty$ almost surely.

Next, we divide our discussion for $p > d/2$ and $p = d/2$.

\emph{Case 1: $p > d/2$.}
$2p/d > 1$ gives $\sum_{i \geq 1} i^{-2p/d} < \infty$. Hence, for every $\varepsilon > 0$, we can choose $n_0$ so large that $\sum_{i > n_0} i^{-2p/d}$ can be made arbitrarily small, while $\sum_{j = 1}^n j^{-2p/d}$ is bounded away from from zero for finite $p$.

Since $m$ is continuous at $\mathbf{x}$ and $r_{(i),n}(\mathbf{x}) \to 0$, \begin{align}
    \max_{1 \leq i \leq n_0} |m(\mathbf{X}_{(i), n}) - m(\mathbf{x})| \to 0.
\end{align} Therefore, \begin{align}
    |\hat m_{\mathrm{IDW}}(\mathbf{x})-m(\mathbf{x})| \leq \sum_{i = 1}^{n_0} a_{(i), n}(\mathbf{x})|m(\mathbf{X}_{(i),n})-m(\mathbf{x})| + 2\|m\|_\infty \sum_{i > n_0} a_{(i), n}(\mathbf{x}), \label{eqn:IDW_gap}
\end{align} and the claim follows either in probability or almost surely.

\emph{Case 2: $p = d/2$.}
Set \begin{align}
    \ell_n := \left\lfloor\frac{n}{\log n}\right\rfloor.
\end{align} Then $\ell_n \to \infty$ and $\ell_n / n \to 0$, so $r_{(\ell_n),n}(\mathbf{x})\to 0$. Moreover, the normalized weights are comparable to $i^{-1}/\sum_{j = 1}^n j^{-1}$.

Therefore, \begin{align}
    \sum_{i > \ell_n} a_{(i), n} (\mathbf{x}) \lesssim \frac{\sum_{i > \ell_n} i^{-1}}{\sum_{j = 1}^n j^{-1}} \asymp \frac{\log\log n}{\log n} \to 0
\end{align} either in probability or almost surely. By \eqref{eqn:IDW_gap}, we again conclude either mode of convergence under the corresponding conditions.

The argument from pointwise convergence to $L^1$ convergence is the same as the argument in \ref{lemma:Weighted_IDW}.
\end{proof}

\subsubsection{Proof of Theorem 2}

\begin{proof}
By Lemma \ref{lemma:RQK_KR}, $\zeta_n(\theta_n) \overset{p}{\to} 0$ and $\frac{\max_{1\leq i \leq n} c_i}{\min_{1 \leq i \leq n} c_i} = \mathcal{O}_p(1)$ gives \begin{align}
    \hat m_{\mathrm{RQK}}(\mathbf{x})-\hat m_{\mathrm{KR}}(\mathbf{x}) \overset{p}{\to} 0;
\end{align} $\zeta_n(\theta_n) \overset{a.s.}{\to} 0$ and $\sup_{n \geq 1} \frac{\max_{1\leq i \leq n} c_i}{\min_{1 \leq i \leq n} c_i} < \infty$ gives \begin{align}
    \hat m_{\mathrm{RQK}}(\mathbf{x})-\hat m_{\mathrm{KR}}(\mathbf{x}) \overset{a.s.}{\to} 0;
\end{align}

Next, as in the proof of Theorem \ref{theorem:Weak_consistency}, $\zeta_n(\theta_n) \to 0$ implies $n\theta_n^d to 0$. By Lemma \ref{lemma:KR_IDW}, \begin{align}
    \hat m_{\mathrm{KR}}(\mathbf{x}) - \hat m_{\mathrm{IDW}}(\mathbf{x}) \overset{p}{\to} 0.
\end{align} Additionally, Lemma \ref{lemma:KR_IDW} shows if $\sum_{n = 1}^\infty n\theta_n^d < \infty$, \begin{align}
    \hat m_{\mathrm{KR}}(\mathbf{x}) - \hat m_{\mathrm{IDW}}(\mathbf{x}) \overset{a.s.}{\to} 0.
\end{align}

Finally, apply Lemma \ref{lemma:IDW_continuity} with $p \geq d/2$ and $\frac{\max_{1\leq i \leq n} c_i}{\min_{1 \leq i \leq n} c_i} = \mathcal{O}_p(1)$ gives \begin{align}
    \hat m_{\mathrm{IDW}}(\mathbf{x}) \overset{p}{\to} m(\mathbf{x});
\end{align} and $\sup_{n \geq 1} \frac{\max_{1\leq i \leq n} c_i}{\min_{1 \leq i \leq n} c_i} < \infty$ gives \begin{align}
    \hat m_{\mathrm{IDW}}(\mathbf{x}) \overset{a.s.}{\to} m(\mathbf{x});
\end{align}

Combining the three steps proves pointwise convergence, and $L^1$ convergence again follows from the standard argument.
\end{proof}

\subsection{Proof of Theorem 3}

\subsubsection{Additional lemmas for Theorem 3}

\begin{lemma}[Ball counts under quasi-randomness]
\label{lemma:ball_counts}
Let $\{\mathbf{x}_i\}_{i = 1}^n$ be a quasi-random design on $\mathcal{X}$ with mesh ratio bounded by $\rho < \infty$. Then there exist constants $0 < C_m \leq C_M < \infty$ depending only on $(d,\rho, w)$, such that for every $\mathbf{x} \in \mathcal{X}$ and $r \geq 2q_n$, \begin{align}
    C_m \left(\frac{r}{q_n}\right)^d \leq \#\{i:\|\mathbf{x} - \mathbf{x}_i\|_{\mathbf{w}} \leq r\} \leq C_M \left(\frac{r}{q_n}\right)^d,
\end{align} where $\#$ indicates the cardinality of the set.
\end{lemma}

\begin{proof}
Upper bound: Consider weighted balls with the norm $\|\cdot\|_{\mathbf{w}}$. The balls $B_w(\mathbf{x}_i, q_n)$ are disjoint, and all such balls with centers satisfying $\|\mathbf{x} - \mathbf{x}_i\|_{\mathbf{w}} \leq r$ are contained in $B_w(\mathbf{x}, r + q_n)$. Comparing volumes gives \begin{align}
    \#\{i:\|\mathbf{x} - \mathbf{x}_i\|_{\mathbf{w}} \leq r\} \leq \left(\frac{r + q_n}{q_n}\right)^d \leq C_M \left(\frac{r}{q_n}\right)^d.
\end{align} 

Lower bound: Let $V_i$ be the Voronoi cell of $\mathbf{x}_i$. By the definition of the fill
distance, each cell is contained in $B_w(\mathbf{x}_i, h_n)$, so $|V_i| \leq |B_w(\mathbf{x}_i, h_n)|$.

If $r > h_n$, every point of $B_w(\mathbf{x}, r-h_n)$ belongs to a Voronoi cell whose nucleus lies within distance $r$ of $\mathbf{x}$. Therefore, volume comparison gives \begin{align}
    |B_w(\mathbf{x}, r - h_n)| \leq \sum_{\|\mathbf{x} - \mathbf{x}_i\|_{\mathbf{w}}\le r} |V_i| \leq |B_w(\mathbf{x}_i, h_n)| \#\{i:\|\mathbf{x}-\mathbf{x}_i\|_{\mathbf{w}} \leq r\}.
\end{align} It follows that when $r > 2h_n$ \begin{align}
    \#\{i:\|\mathbf{x} - \mathbf{x}_i\|_{\mathbf{w}} \leq r\} \geq \left(\frac{r - h_n}{h_n}\right)^d \geq 2^{-d} \left(\frac{r}{h_n}\right)^d \geq (2\rho)^{-d} \left(\frac{r}{b_n}\right)^d.
\end{align} On the other hand, when $2q_n \leq r \leq 2h_n$, notice that \begin{align}
    \#\{i:\|\mathbf{x} - \mathbf{x}_i\|_{\mathbf{w}} \leq r\} \geq 1 \geq (2\rho)^{-d} \left(\frac{r}{q_n}\right)^d.
\end{align}
\end{proof}

\begin{lemma}[Constant-order preweights under quasi-randomness]
\label{lemma:c_quasirandom}
Let $\{\mathbf{x}_i\}_{i = 1}^n$ be a quasi-random design on $\mathcal{X}$ with mesh ratio bounded by $\rho < \infty$. Then there exists a constant $A = A(d, p, \rho, w) \geq 1$ such that \begin{align}
    \frac{\max_{1 \leq i \leq n} c_i}{\min_{1 \leq i \leq n} c_i} \leq A
\end{align} for all $n$ and $\theta > 0$.
\end{lemma}

\begin{proof}
By Lemma \ref{lemma:ball_counts}, there exists $0 < C_m \leq C_M < \infty$ such that \begin{align}
    C_m \left(\frac{r}{q_n}\right)^d \leq \#\{j:\|\mathbf{x}_j -\mathbf{x}_i\|_{\mathbf{w}} \leq r\} \leq C_M \left(\frac{r}{q_n}\right)^d
\end{align} for every $i$ and $r \geq 2q_n$.

Choose $\lambda > \max\{2, (C_M/C_m)^{1/d}\}$ and let \begin{align}
    r_m := 2\lambda^m q_n.
\end{align} Let $M_n$ be the smallest integer with $r_{M_n} \geq D$, and define shells \begin{align}
    \mathcal{S}_{m}(\mathbf{x}_i) := \{j \neq i: r_m < \|\mathbf{x}_j - \mathbf{x}_i\|_{\mathbf{w}} \leq r_{m + 1}\}, m = 0, \dots, M_n - 1.
\end{align} Then \begin{align}
    \sum_{j \neq i} R(\mathbf{x}_i,\mathbf{x}_j) = \sum_{m = 0}^{M_n - 1} \sum_{j \in \mathcal{S}_{m}(\mathbf{x}_i)} \left(1 + \frac{\|\mathbf{x}_i - \mathbf{x}_j\|_{\mathbf{w}}^2}{\theta^2}\right)^{-p}.
\end{align}

Using the ball-count bounds at radii $r_m$ and $r_{m + 1}$, we obtain constants $0 < \underline C \leq \overline C < \infty$ depending only on $(d, \rho, w)$ such that \begin{align}
    \underline C \lambda^{md} \leq \# \{\mathcal{S}_{m}(\mathbf{x}_i)\} \leq \overline C \lambda^{md}
\end{align} for all $i, m$.

Since every distance in $\mathbf{S}_{m}(\mathbf{x}_i)$ lies between $r_m$ and $r_{m + 1}$, \begin{align}
    \left(1 + \frac{r_{m + 1}^2}{\theta^2}\right)^{-p} \leq \left(1 + \frac{\|\mathbf{x}_i - \mathbf{x}_j\|_{\mathbf{w}}^2}{\theta^2}\right)^{-p} \leq \left(1 + \frac{r_m^2}{\theta^2}\right)^{-p}.
\end{align} Therefore, \begin{align}
    \underline C \sum_{m = 0}^{M_n - 1} \lambda^{md} \left(1 + \frac{r_{m + 1}^2}{\theta^2}\right)^{-p} \leq \sum_{j \neq i} R(\mathbf{x}_i,\mathbf{x}_j) \leq \overline C \sum_{m = 0}^{M_n - 1} \lambda^{md} \left(1 + \frac{r_m^2}{\theta^2}\right)^{-p}.
\end{align} The two sums are comparable up to a constant depending only on $(d, p, \rho, w)$, so
$\sum_{j \neq i} R(\mathbf{x}_i, \mathbf{x}_j)$ for all $i$ are comparable. Taking reciprocals proves the claim.
\end{proof}

\begin{lemma}[Diagonal dominance under quasi-randomness]
\label{lemma:zeta_quasirandom}
Under the assumptions of Lemma \ref{lemma:c_quasirandom}, there exists
$C = C(d, p, \rho, w)$ such that \begin{align}
    \zeta_n(\theta) \leq \begin{cases}
    C \left(\dfrac{\theta}{q_n}\right)^{2p}, & p > d/2, \\
    C \left(\dfrac{\theta}{q_n}\right)^d \log \left(\dfrac{D}{q_n}\right), & p = d/2.
\end{cases}
\end{align}
In particular, if \begin{align}
    \bar\theta_n = c q_n^{1 + \delta} \text{ for some } c > 0, \delta > 0,
\end{align}
$\zeta_n(\bar \theta_n) \to 0$ as $n \to \infty$.
\end{lemma}

\begin{proof}
Using the shell decomposition from the proof of Lemma \ref{lemma:c_quasirandom}: \begin{align}
    \zeta_n(\theta) = \max_{1 \leq i \leq n} \sum_{j \neq i} R(\mathbf{x}_i, \mathbf{x}_j) \leq \bar C \sum_{m = 0}^{M_n - 1} \lambda^{md} \left(1 + \frac{r_m^2}{\theta^2}\right)^{-p}.
\end{align}

Because \begin{align}
    \left(1 + \frac{r_m^2}{\theta^2}\right)^{-p} \leq \left(\frac{\theta}{r_m}\right)^{2p},
\end{align} we have \begin{align}
    \zeta_n(\theta) \lesssim \left(\frac{\theta}{q_n}\right)^{2p} \sum_{m = 1}^{M_n - 1} \lambda^{m(d - 2p)}.
\end{align}

If $p > d/2$, the series converges and we have $\zeta_n(\theta) \lesssim \left(\frac{\theta}{q_n}\right)^{2p}$.

If $p = d/2$, the exponent vanishes and we have \begin{align}
    \zeta_n(\theta) \lesssim \left(\frac{\theta}{q_n}\right)^d M_n \asymp \left(\frac{\theta}{q_n}\right)^d \log \left(\frac{D}{q_n}\right),
\end{align} which proves the claim.
\end{proof}

\subsubsection{Proof of Theorem 3}

\begin{proof}
By Lemma \ref{lemma:c_quasirandom} and Lemma \ref{lemma:zeta_quasirandom}, the assumptions of
Lemma \ref{lemma:RQK_KR} hold uniformly for $0 < \theta \leq \bar\theta_n$. Therefore, \begin{align}
    \sup_{\mathbf{x} \in \mathcal{X}} |\hat m_{\mathrm{RQK}}(\mathbf{x}) - \hat m_{\mathrm{KR}}(\mathbf{x})| \leq C \zeta_n(\bar\theta_n) \|m\|_\infty. \label{eqn:sup_RQK_KR}
\end{align}

Fix $\mathbf{x} \in \mathcal{X}$ and $r \in [h_n, D]$. Since the kernel regression weights are nonnegative and sum to one, \begin{align}
    |\hat m_{\mathrm{KR}}(\mathbf{x}) -m(\mathbf{x})| \leq \sum_{i = 1}^n a_{i, \mathrm{KR}} |m(\mathbf{x}_i) - m(\mathbf{x})|. \label{eqn:bounded_KR}
\end{align}
Split the sum into ``nearby'' and ``far-away'' parts: \begin{align}
    |\hat m_{\mathrm{KR}}(\mathbf{x}) -m(\mathbf{x})| \leq \varphi_m(r) + 2\|m\|_{\infty} \sum_{\|\mathbf{x}_i - \mathbf{x}\|_{\mathbf{w}} > r} a_{i, \mathrm{KR}}(\mathbf{x}). \label{eqn:sup_KR}
\end{align}

Applying the same argument as Lemma \ref{lemma:zeta_quasirandom}, if $p > d/2$ and $r_{m_r} < r$,
\begin{align}
    \sum_{\|\mathbf{x}_i - \mathbf{x}\|_{\mathbf{w}} > r} a_{i, \mathrm{KR}}(\mathbf{x}) \asymp \frac{\sum_{\|\mathbf{x}_i - \mathbf{x}\|_{\mathbf{w}} > r} R(\mathbf{x}, \mathbf{x}_i)}{\sum_{j = 1}^n R(\mathbf{x}, \mathbf{x}_j)} \lesssim \frac{\sum_{m > m_r}^{M_n - 1} \lambda^{m(d - 2p)}}{\sum_{m = 0}^{M_n - 1} \lambda^{m(d - 2p)}} \asymp \left(\frac{r}{h_n}\right)^{d - 2p}.
\end{align} On the other hand, when $p = d/2$, \begin{align}
    \frac{\sum_{m > m_r}^{M_n - 1} \lambda^{m(d - 2p)}}{\sum_{m = 0}^{M_n - 1} \lambda^{m(d - 2p)}} \asymp \frac{\log(D/r)}{\log(D/h_n)}.
\end{align}

Putting together \eqref{eqn:sup_RQK_KR} and \eqref{eqn:sup_KR} and taking the supremum over $\mathbf{x} \in \mathcal{X}$ gives \begin{align}
    \sup_{0 < \theta \leq \bar\theta_n}\sup_{\mathbf{x} \in \mathcal{X}} |\hat m_{\mathrm{RQK}}(\mathbf{x}) - m(\mathbf{x})| \leq C \zeta_n(\bar \theta_n) \|m\|_{\infty} + C \inf_{r \in [h_n, D]}\{\varphi_m(r) + 2\|m\|_{\infty} \psi_{p, n}(r)\}.
\end{align}
\end{proof}

\subsubsection{Proof of Corollary 1}

\begin{proof}
Following the proof of Theorem \ref{theorem:quasirandom}, in \eqref{eqn:bounded_KR}, instead of splitting, directly use the H\"older property: \begin{align}
    |\hat m_{\mathrm{KR}}(\mathbf{x}) - m(\mathbf{x})| \lesssim \frac{\sum_{i = 1}^n R(\mathbf{x}, \mathbf{x}_i) \|\mathbf{x} - \mathbf{x}_i\|_{\mathbf{w}}^{\alpha}}{\sum_{j = 1}^n R(\mathbf{x}, \mathbf{x}_i)} \asymp \frac{\sum_{i = 1}^n \|\mathbf{x} - \mathbf{x}_i\|_{\mathbf{w}}^{\alpha - 2p}}{\sum_{j = 1}^n \|\mathbf{x} - \mathbf{x}_i\|_{\mathbf{w}}^{-2p}}.
\end{align}

Again using the shell-bound argument, \begin{align}
    \|\mathbf{x} - \mathbf{x}_i\|_{\mathbf{w}}^{\alpha - 2p} &\asymp h_n^{\alpha - 2p} \sum_{m = 0}^{M_n - 1} \lambda^{m(d + \alpha - 2p)}, \\
    \|\mathbf{x} - \mathbf{x}_i\|_{\mathbf{w}}^{-2p} &\asymp h_n^{-2p} \sum_{m = 0}^{M_n - 1} \lambda^{m(d - 2p)}.
\end{align} Taking the ratio yields: \begin{itemize}
    \item When $p > (d + \alpha)/2$, both series converge, yielding $\mathcal{O}(h_n^{\alpha})$;

    \item When $d/2 < p \leq (d + \alpha)/2$, the denominator converges, but the numerator is dominated by the last shell, yielding $\mathcal{O}(h_n^{2p -d})$;

    \item When $p = 2$, the denominator is logarithmic, yielding $\mathcal{O}(1/\log(D/h_n))$,
\end{itemize} which gives the rates of $\phi_n$.

Meanwhile, $\bar \theta_n = c q_n^{1 + \delta}$ gives \begin{align}
    \zeta_n(\bar \theta_n) \lesssim \begin{cases}
        q_n^{2p\delta}, & p > d/2, \\
        q_n^{d\delta} \log(D/q_n), & p = d/2.
    \end{cases} 
\end{align} Comparing with the rates of $\phi_n$, we need \begin{enumerate}
    \item [(i)] $2p\delta \geq \alpha$ when $p > (d + \alpha)/2$, and
    \item [(ii)] $2p\delta \geq 2p - d$ when $d/2 < p \leq (d + \alpha)/2$
\end{enumerate} for $\phi_n$ to dominate. This is exactly satisfied by the condition $\delta \geq \frac{1}{2p} \min\{\alpha, 2p - d\}$.
\end{proof}

\section{Derivation of leave-one-out expressions}

Here, we derive the expressions for leave-one-out point estimators and prediction errors given in Section 4.2. In addition to $\mathbf{z} = \mathbf{R}^{-1} \tilde{\mathbf{r}}(\mathbf{x}), \mathbf{\Lambda} = \mathbf{R}^{-1}$ specified in the paper, denote $\mathbf{d} = \mathbf{R} \tilde{\mathbf{c}}$.  

First, consider the point estimator. When sample $i$ is left out from the training set, the vectors $\tilde{\mathbf{r}}(\mathbf{x})$ and $\tilde{\mathbf{c}}$ have length $n - 1$. The denominator becomes \begin{align}
    [\tilde{\mathbf{r}}(\mathbf{x})]_{-i}^{\top} \tilde{\mathbf{c}}_{-i} = \tilde{\mathbf{r}}(\mathbf{x})^{\top} \tilde{\mathbf{c}} - \theta^{-2p} R(\mathbf{x}, \mathbf{x}_i) \tilde{c}_i. \label{eqn:LOO_denom}
\end{align}

For the numerator, first consider $\mathbf{d}$. Because \begin{align}
    d_j = \sum_{k \neq i} R_{jk} \tilde{c}_k + R_{ji} \tilde{c}_i,
\end{align} when $i$ is left out, \begin{align}
    \mathbf{d}_{\sim i} = \mathbf{d}_{-i} - \tilde{c}_i \mathbf{R}_{-i, i},
\end{align} where $\mathbf{R}_{-i, i}$ denotes the $i$-th column, but leaving out the $i$-th row. Notice the difference between $\mathbf{d}_{\sim i}$ (the new $\mathbf{d}$ when sample $i$ is left out) and $\mathbf{d}_{-i}$ (vector $\mathbf{d}$ with the $i$-th element left out). This gives us the leave-one-out numerator: \begin{align}
    & [\mathbf{\tilde r}(\mathbf{x})]_{-i}^\top [\mathbf{R}_{-i, -i}]^{-1} \mathrm{diag}(\mathbf{d}_{\sim i}) \mathbf{y}_{-i} \\
    =& [\mathbf{\tilde r}(\mathbf{x})]_{-i}^\top [\mathbf{R}_{-i, -i}]^{-1} (\mathbf{d}_{-i} \odot \mathbf{y}_{-i}) - \tilde{c}_i [\mathbf{\tilde r}(\mathbf{x})]_{-i}^\top [\mathbf{R}_{-i, -i}]^{-1} (\mathbf{R}_{-i, i} \odot \mathbf{y}_{-i}) \\
    =& \left[\mathbf{z} (\mathbf{d} \odot \mathbf{y}) - \frac{z_i [\mathbf{\Lambda} \mathrm{diag}(\mathbf{d}) \mathbf{y}]_i}{\Lambda_{ii}}\right] - \tilde{c}_i \left[\sum_{j} z_j R_{ij} y_j - \frac{z_i[\sum_{l} \Lambda_{il} R_{il} y_l]}{\Lambda_{ii}} \right],
\end{align} where both expressions are from the application of the Sherman-Morrison formula.

As for the leave-one-out prediction variance, the denominator is just \eqref{eqn:LOO_denom} squared. For the quadratic term in the numerator, again using Sherman-Morrison: \begin{align}
    [\mathbf{\tilde r}(\mathbf{x})]_{-i}^\top [\mathbf{R}_{-i, -i}]^{-1} [\mathbf{\tilde r}(\mathbf{x})]_{-i} = \mathbf{\tilde r}(\mathbf{x})^\top \mathbf{R}^{-1} \mathbf{\tilde r}(\mathbf{x}) - \frac{z_i^2}{\Lambda_{ii}},
\end{align} and we have the desired expressions that do not involve the $-i$ notation. 

During the estimation of $(\mathbf{w}, \theta, \gamma)$, we only need the expressions for $\{\mathbf{x_i}\}_{i = 1}^n$. For $\mathbf{x} = \mathbf{x}_i$, we have the relationships \begin{align}
    \tilde{\mathbf{r}}(\mathbf{x_i}) \tilde{\mathbf{c}} = \theta^{-2p} d_i, \quad 
    [\mathbf{R}_{-i, -i}]^{-1} [\tilde{\mathbf{r}}(\mathbf{x}_i)]_{-i} = -\frac{\theta^{-2p} \mathbf{\Lambda}_{-i, i}}{\Lambda_{ii}}.
\end{align} Then the numerator becomes \begin{align}
    & -\frac{\theta^{-2p}}{\Lambda_{ii}} \Big[\mathbf{\Lambda}_{i, -i} (\mathbf{d}_{-i} \odot \mathbf{y}_{-i}) - \tilde{c}_i \mathbf{\Lambda}_{i, -i} (\mathbf{R}_{-i, i} \odot \mathbf{y}_{-i})\Big] \\
    =& -\frac{\theta^{-2p}}{\Lambda_{ii}} \Big[ \sum_{j \neq i} \Lambda_{ij} d_j y_j - \tilde{c}_i \sum_{j \neq i} \Lambda_{ij} R_{ij} y_j \Big] \\
    =& -\frac{\theta^{-2p}}{\Lambda_{ii}} \Big[[\mathbf{\Lambda} (\mathbf{d} \odot \mathbf{y})]_i - \Lambda_{ii} (\mathbf{d} \odot \mathbf{y}) - \tilde{c}_i \big\{[(\mathbf{\Lambda} \odot \mathbf{R}) \mathbf{y}]_{i} - \Lambda_{ii} R_{ii} y_i \big\}\Big].
\end{align}
The final expression: \begin{align}
    \hat m_{\mathrm{RQK}}(\mathbf{x}_i) = -\frac{[\mathbf{\Lambda} (\mathbf{d} \odot \mathbf{y})]_i - \Lambda_{ii} (\mathbf{d} \odot \mathbf{y}) - \tilde{c}_i \big\{[(\mathbf{\Lambda} \odot \mathbf{R}) \mathbf{y}]_{i} - \Lambda_{ii} R_{ii} y_i \big\}}{\Lambda_{ii} (d_i - R_{ii} \tilde{c}_i)}.
\end{align}

The variance term has \begin{align}
    R_{ii} - \mathbf{R}_{-i, i}^\top \Lambda_{-i, -i} \mathbf{R}_{-i, i} = \frac{1}{\Lambda_{ii}}.
\end{align} Hence, \begin{align}
    s^2_{\mathrm{RQK}}(\mathbf{x}_i) = \frac{1}{\Lambda_{ii} (d_i - R_{ii} \tilde{c}_i)^2}.
\end{align}

\section{Variants of conformal prediction}

\begin{figure}[tb]
\centering

\begin{subfigure}{0.49\textwidth}
\centering
\includegraphics[width=\textwidth]{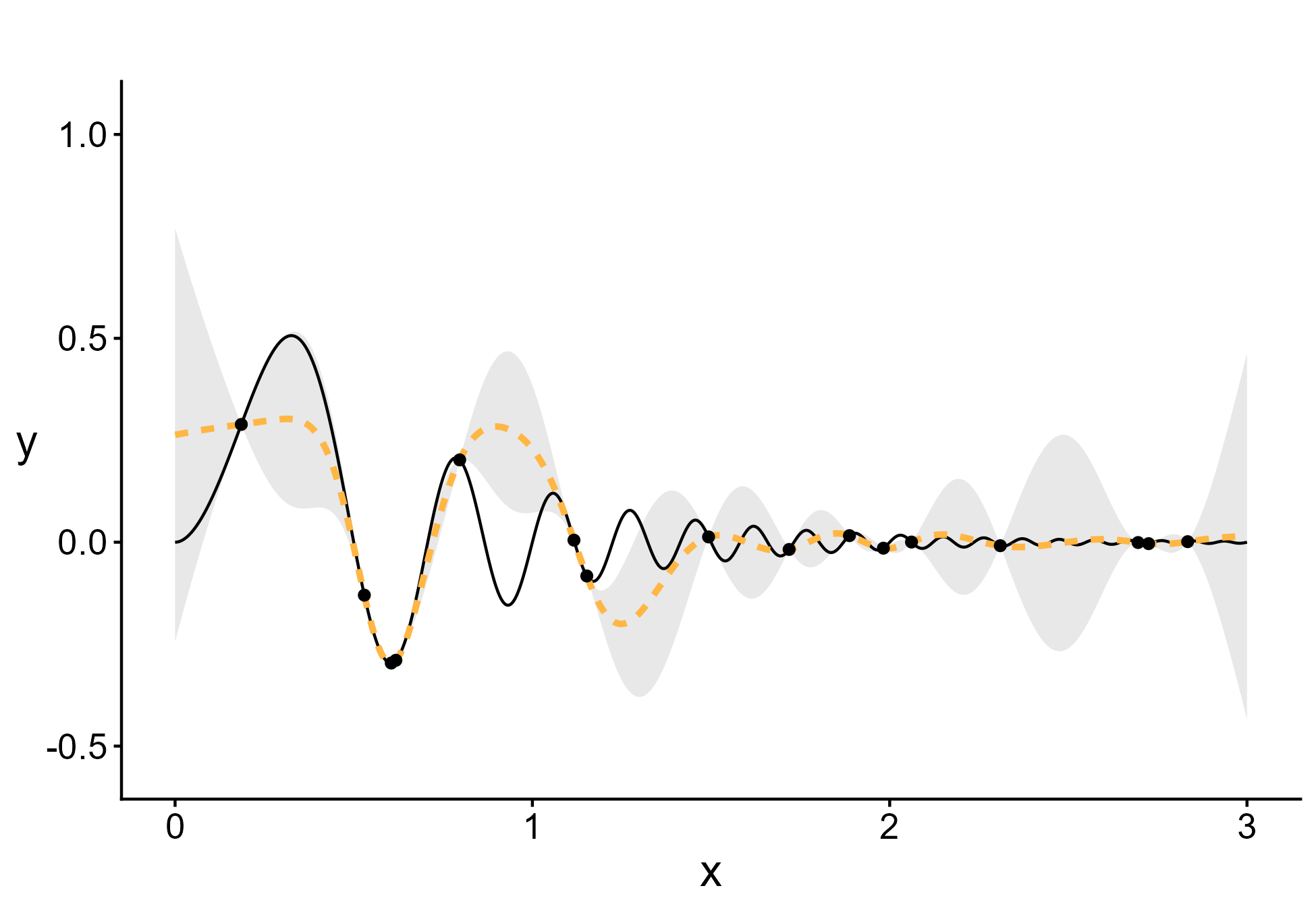}
\caption{Type 1.}
\end{subfigure}%
\begin{subfigure}{0.49\textwidth}
\centering
\includegraphics[width=\textwidth]{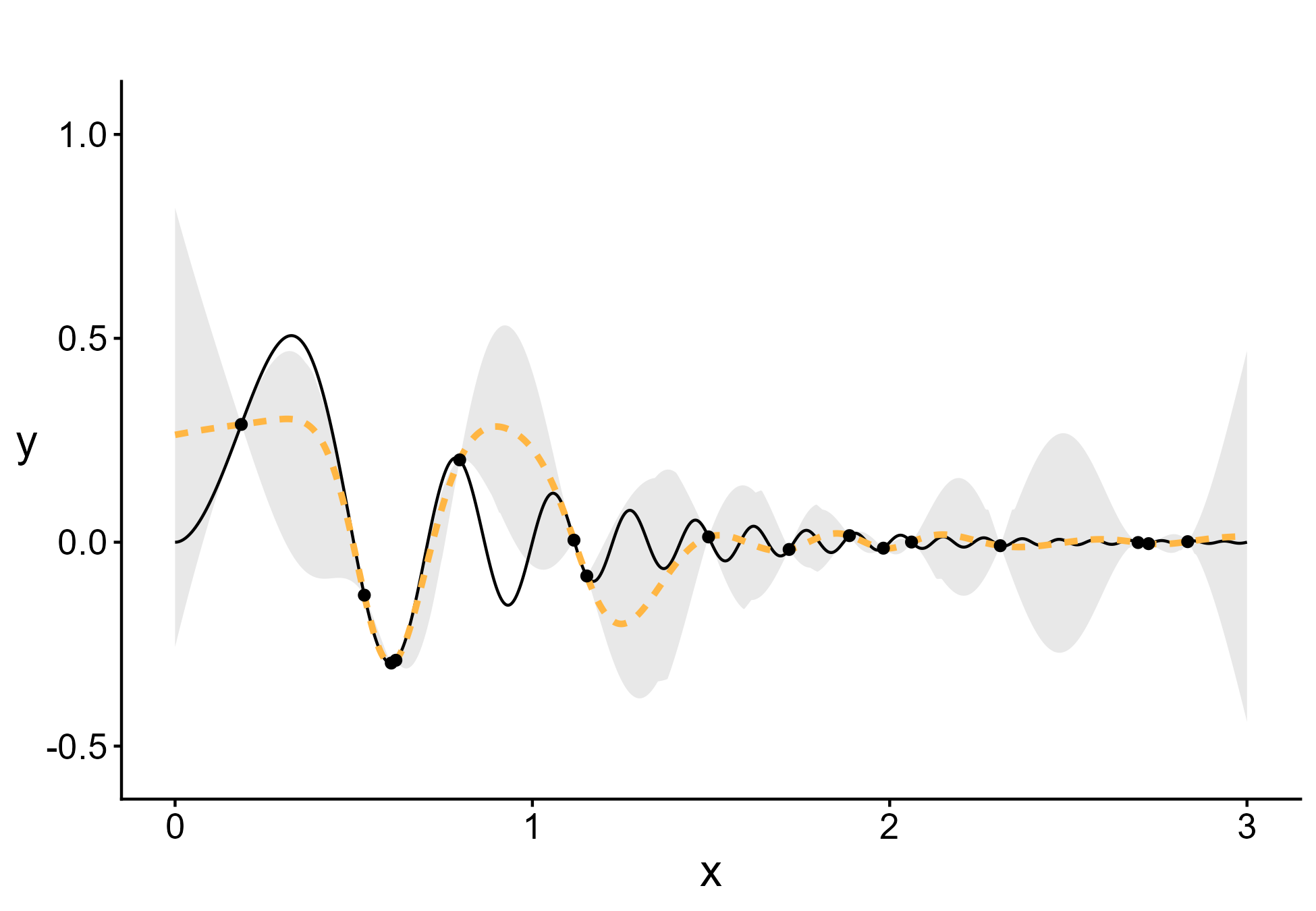}
\caption{Type 4.}
\end{subfigure}%
\caption{Comparison of different approaches to constructing conformal prediction intervals.}
\label{fig:toy_example_conformal}

\end{figure}

In Section 4.3 of the main text, we provided a construction of the conformal prediction interval. Here, we present other variants, all of which are implemented in our package.

\textit{Type 1}: Using the score $|y - \hat m_{\mathrm{RQK}}(\mathbf{x})|/s_{\mathrm{RQK}}(\mathbf{x})$ for a new $\mathbf{x}$, we can construct a confidence interval by \begin{align}
    \hat m_{\mathrm{RQK}}(\mathbf{x}) \pm Q_{1 - \eta}(\mathbf{S}) s_{\mathrm{RQK}}(\mathbf{x}), \label{eqn:raw_interval}
\end{align} where $Q_{1 - \eta}(\mathbf{S})$ is the $1 - \eta$ empirical quantile function of the scores $\mathbf{S} = (S_1, \dots, S_n)$. This interval is guaranteed to be continuous and symmetric about $\hat m(\mathbf{x})$, but there are no coverage guarantees.

\textit{Type 2}: With the scores $S^{(i)}(\mathbf{x}, y)$ defined in the main text, the conformal set given by \citet{Barber_2021} is \begin{align}
    \mathcal{C}(\mathbf{x}) = \left\{y \in \mathbb{R}: \sum_{i = 1}^n \mathbb{I}\{S^{(i)}(\mathbf{x}, y) > S_i\} < (1 - \eta)(n + 1) \right\}.
\end{align} Writing $L_i = \hat m_{\sim i}(\mathbf{x}) - S_i s_{\sim i}(\mathbf{x}), U_i = \hat m_{\sim i}(\mathbf{x}) - S_i s_{\sim i}(\mathbf{x})$, the expression simplifies to \begin{align}
    \bigcup_{t = 1}^{n - k + 1} [L_{[n - t + 1]}, U_{[t]}]
\end{align} where $k = \lfloor\eta(n + 1)\rfloor$, $L_{[\cdot]}, U_{[\cdot]}$ denotes the order statistics, and $[L_{[n - t + 1]}, U_{[t]}] = \emptyset$ if $L_{[n - t + 1]} > U_{[t]}$. This set has the $1 - 2\eta$ coverage guarantee, but it is not guaranteed to be a continuous interval.

\textit{Type 3}: To make the above set continuous, we simply write \begin{align}
    \mathrm{Conv}\left\{\bigcup_{t = 1}^{n - k + 1} [L_{[n - t + 1]}, U_{[t]}]\right\}.
\end{align}

\textit{Type 4}: This is the version presented in the main text. In the current notations, the interval can also be written as \begin{align}
    [L_{[n - k + 1]}, U_{[n - k + 1]}].
\end{align} It is not difficult to see Type 2 $\subseteq$ Type 3 $\subseteq$ Type 4, so they all possess the same $1 - 2\eta$ coverage guarantee, with Type 4 being the most conservative. In practice, these three sets are often identical, and typically wider than Type 1 because leave-one-out increases uncertainty. 

We now present a one-dimensional example for illustration: the function \begin{align}
    m(x) = \exp(-2x) \sin(4\pi x^2)
\end{align} is observed at 16 random locations in $[0, 3]$. In Figure \ref{fig:toy_example_conformal}, we plot the 90\% prediction intervals for Types 1 and 4. As expected, Type 1 produces smooth, symmetric football-shaped intervals, similar to those under a Gaussian process assumption. On the other hand, Type 4 produces upper and lower bounds that are continuous but not smooth in $\mathbf{x}$, as seen by the small ``kinks'' in the figure. 

\section{Additional simulation results}

Here, we present additional results for the deterministic functions in Section 5.1. Figures \ref{fig:noiseless_20d} and \ref{fig:noiseless_IS_20d} present the logarithm of MSE and IS, respectively, for $n = 20d$. The conclusions are similar to the main text. However, the relative advantage of RQK is smaller than that for $n = 40$. This is because RQK relies on the leave-one-out estimates, so a small dataset will be a limitation. For example, $n = 100$ for the Friedman function with no inert factors, which makes the estimates more volatile.

Figures \ref{fig:noiseless_60d} and \ref{fig:noiseless_IS_60d} present the results for $n = 60d$. Here, RQK is not only the most accurate, but also the least volatile. We specifically comment on the performance of \texttt{mlegp}: the package sometimes provides very accurate estimations (OTL circuit function with $n = 360, d = 6$), but can produce outliers in other examples. Also, it is the only method that produces prediction intervals with coverage much lower than the target 95\%, causing its IS to drop. Furthermore, the computation time is about one hour for the $n = 1200, d = 20$, which we view as unreasonable. The RQK methodology and code implementation details specifically aim to avoid these computational issues, making RQK a desirable choice for quick, accurate, and stable surrogate modeling.

\begin{figure}[p]
\centering

\begin{subfigure}{\textwidth}
\centering
\includegraphics[width=0.4\textwidth]{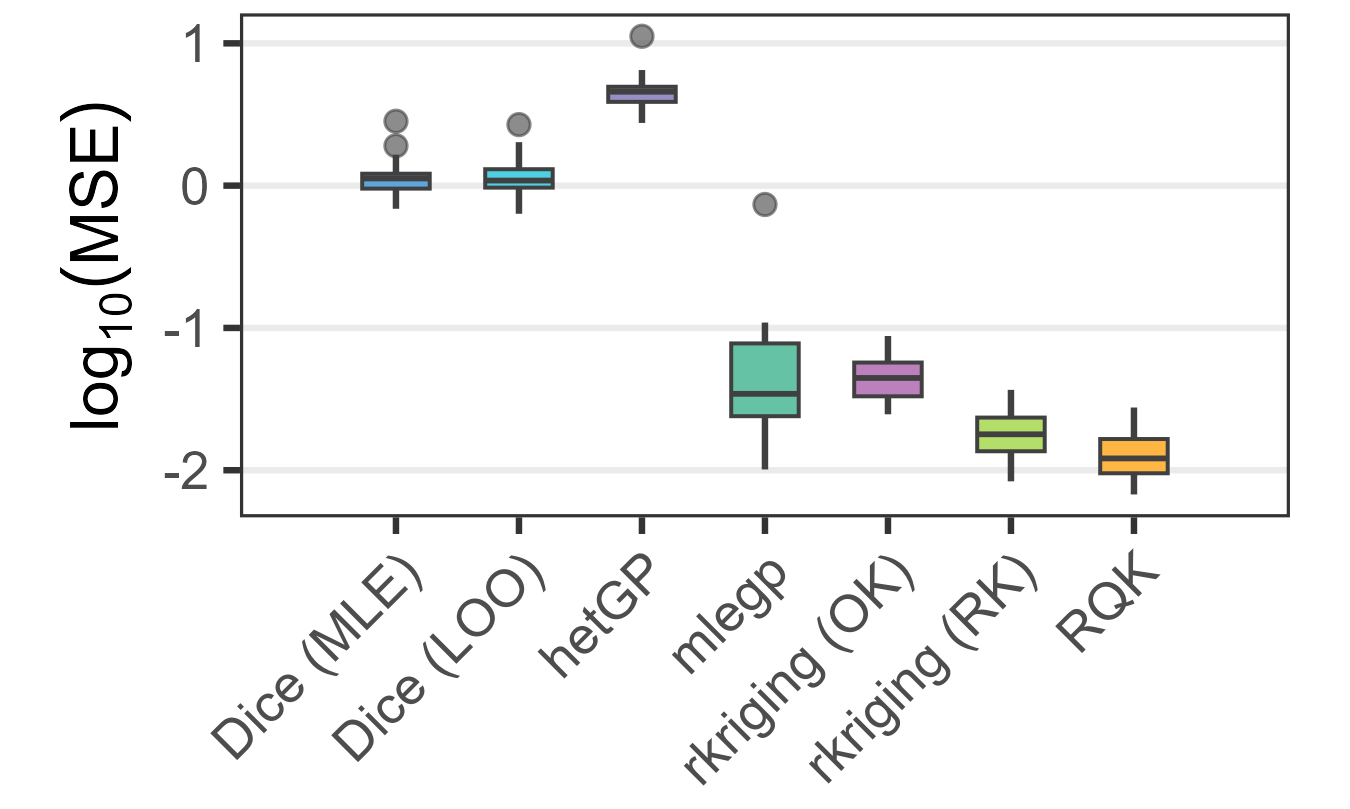}
\includegraphics[width=0.4\textwidth]{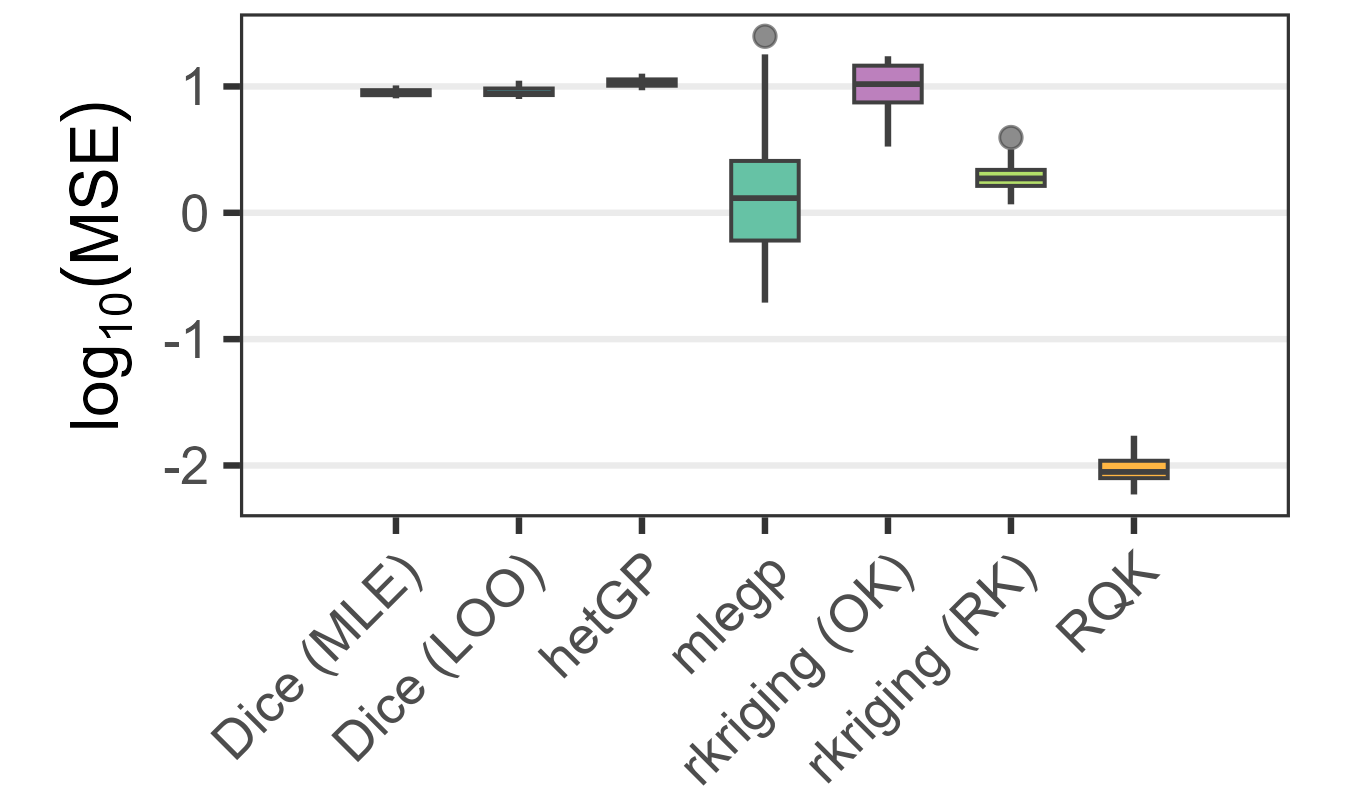}
\caption{Borehole function. Left: 8 dimensions. Right: 20 dimensions.}
\end{subfigure}%

\begin{subfigure}{\textwidth}
\centering
\includegraphics[width=0.4\textwidth]{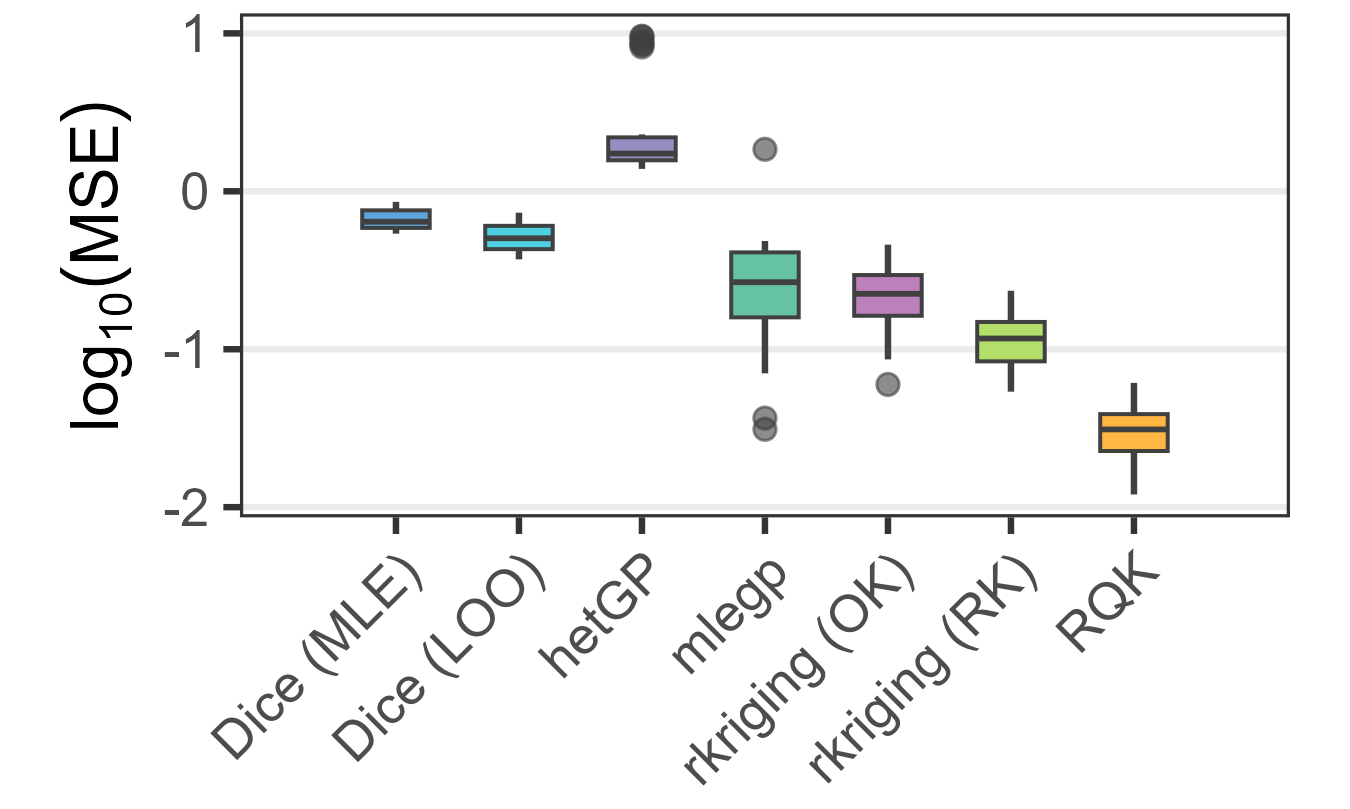}
\includegraphics[width=0.4\textwidth]{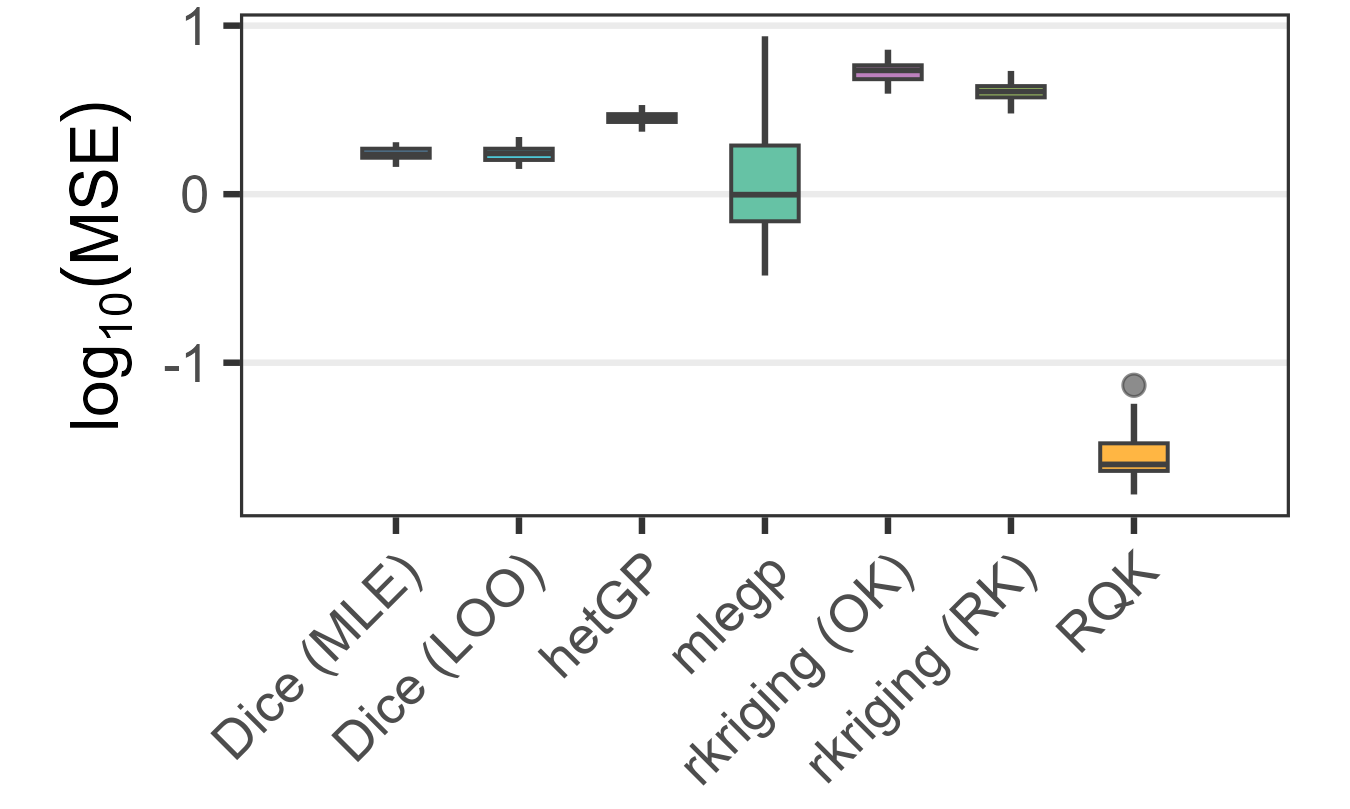}
\caption{Dette and Pepelyshev 8d function. Left: 8 dimensions. Right: 20 dimensions.}
\end{subfigure}%

\begin{subfigure}{\textwidth}
\centering
\includegraphics[width=0.4\textwidth]{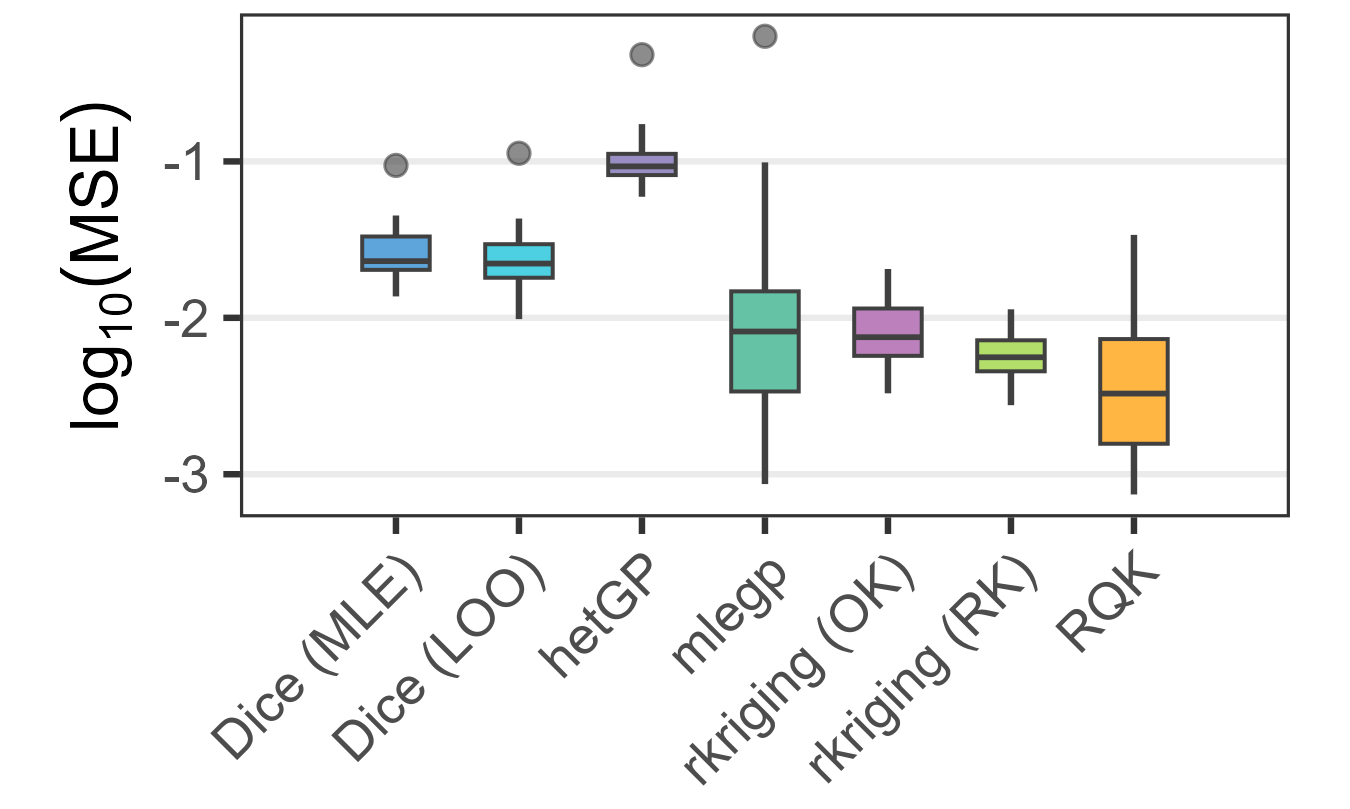}
\includegraphics[width=0.4\textwidth]{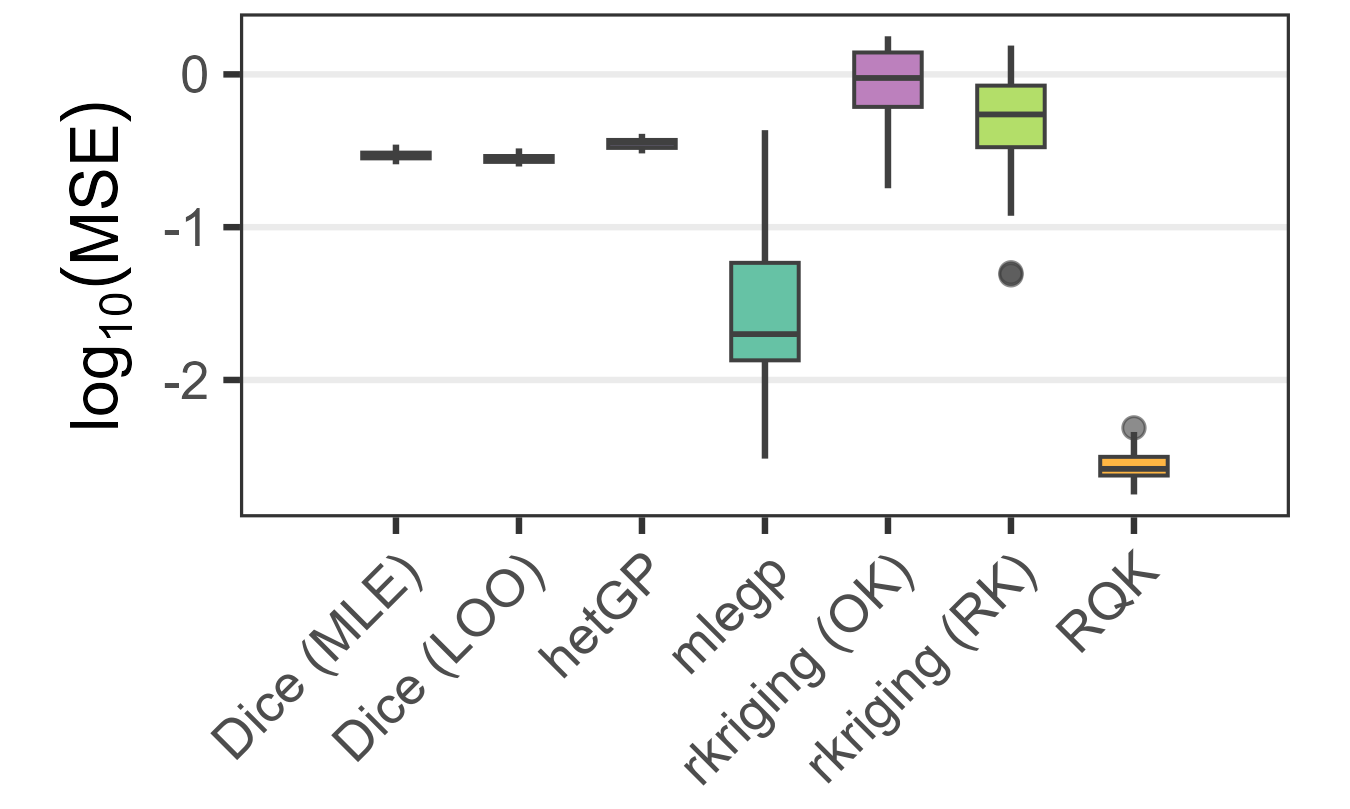}
\caption{Friedman function. Left: 5 dimensions. Right: 20 dimensions.}
\end{subfigure}%

\begin{subfigure}{\textwidth}
\centering
\includegraphics[width=0.4\textwidth]{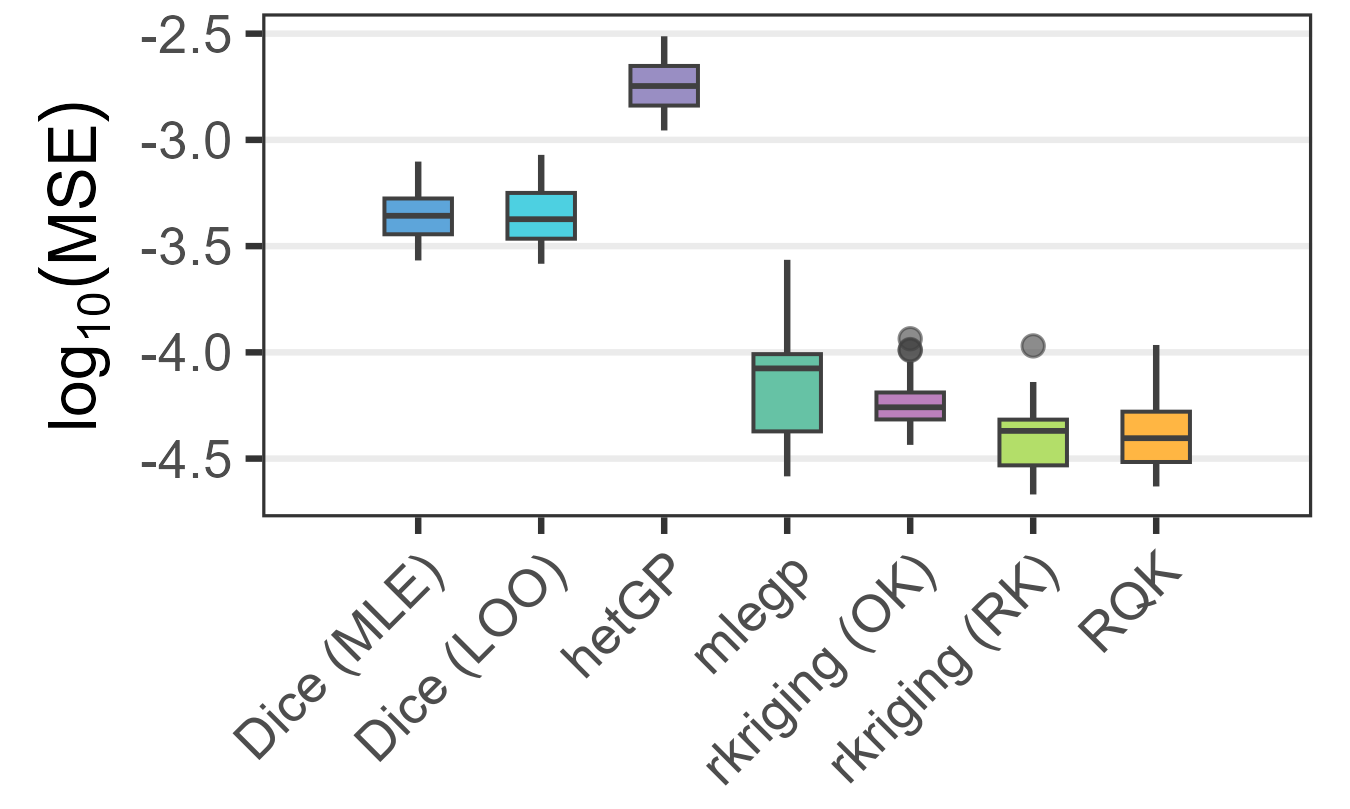}
\includegraphics[width=0.4\textwidth]{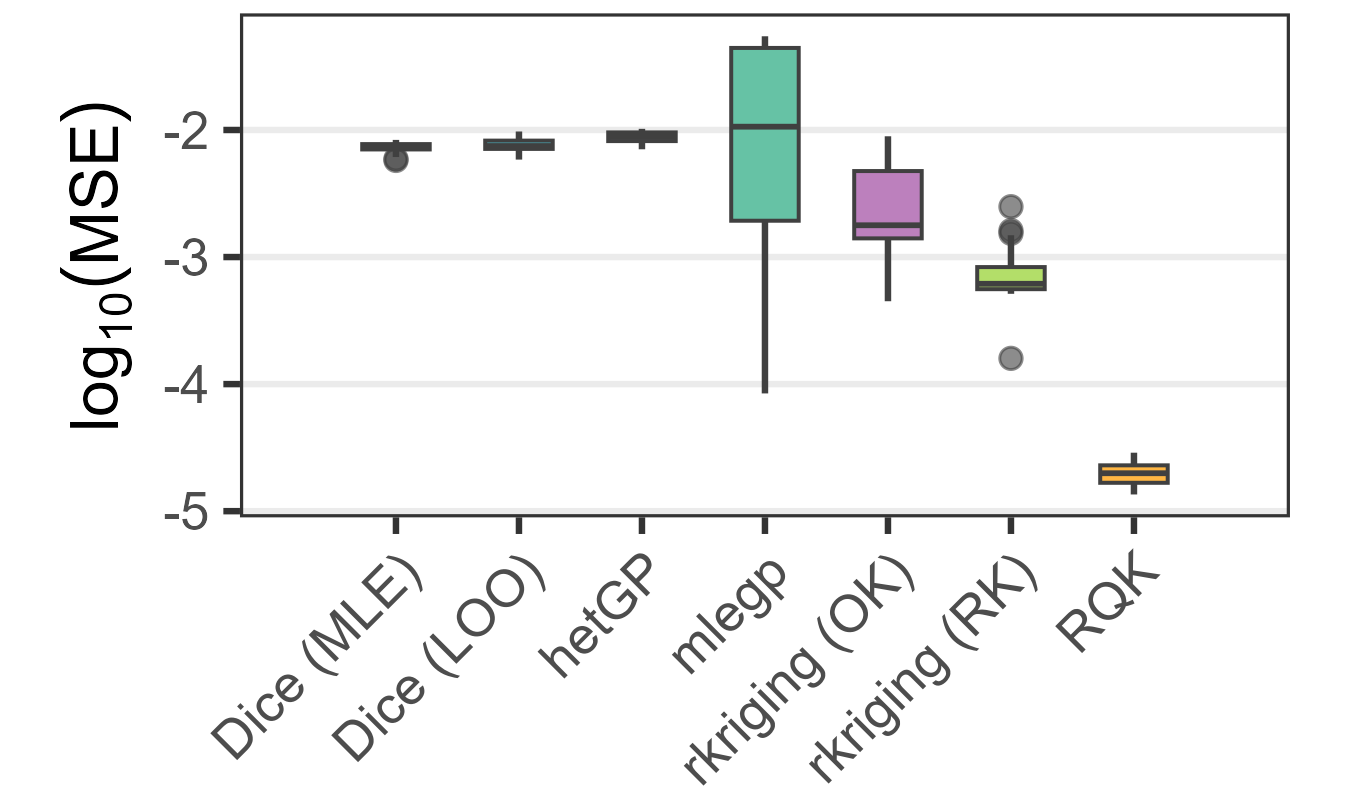}
\caption{OTL circuit function. Left: 6 dimensions. Right: 20 dimensions.}
\end{subfigure}%

\caption{Logarithm MSE plots for deterministic functions with $n = 20d$. The original functions are on the left, and functions with augmented inert variable are on the right. Each experiment is repeated over 30 random seeds for data generation.}
\label{fig:noiseless_20d}

\end{figure}

\begin{figure}[p]
\centering

\begin{subfigure}{\textwidth}
\centering
\includegraphics[width=0.4\textwidth]{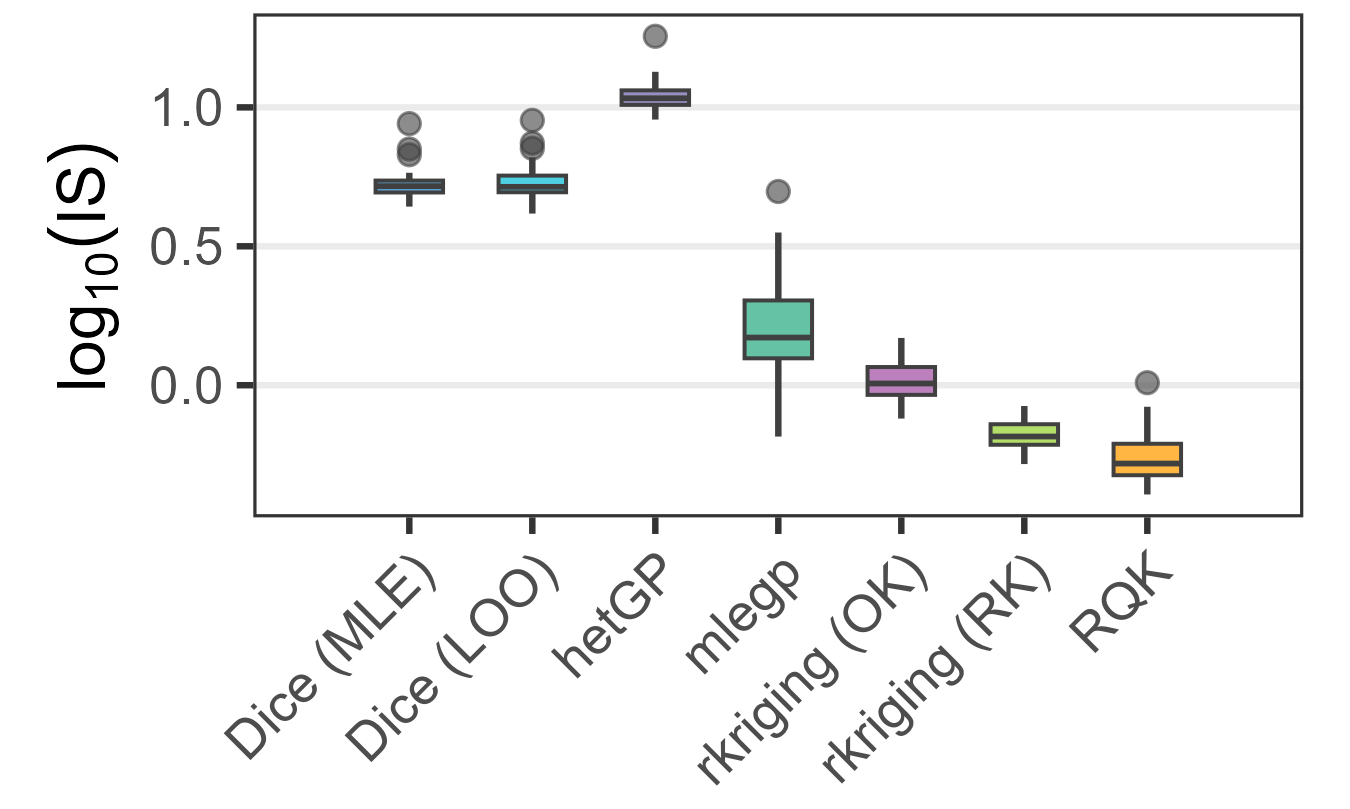}
\includegraphics[width=0.4\textwidth]{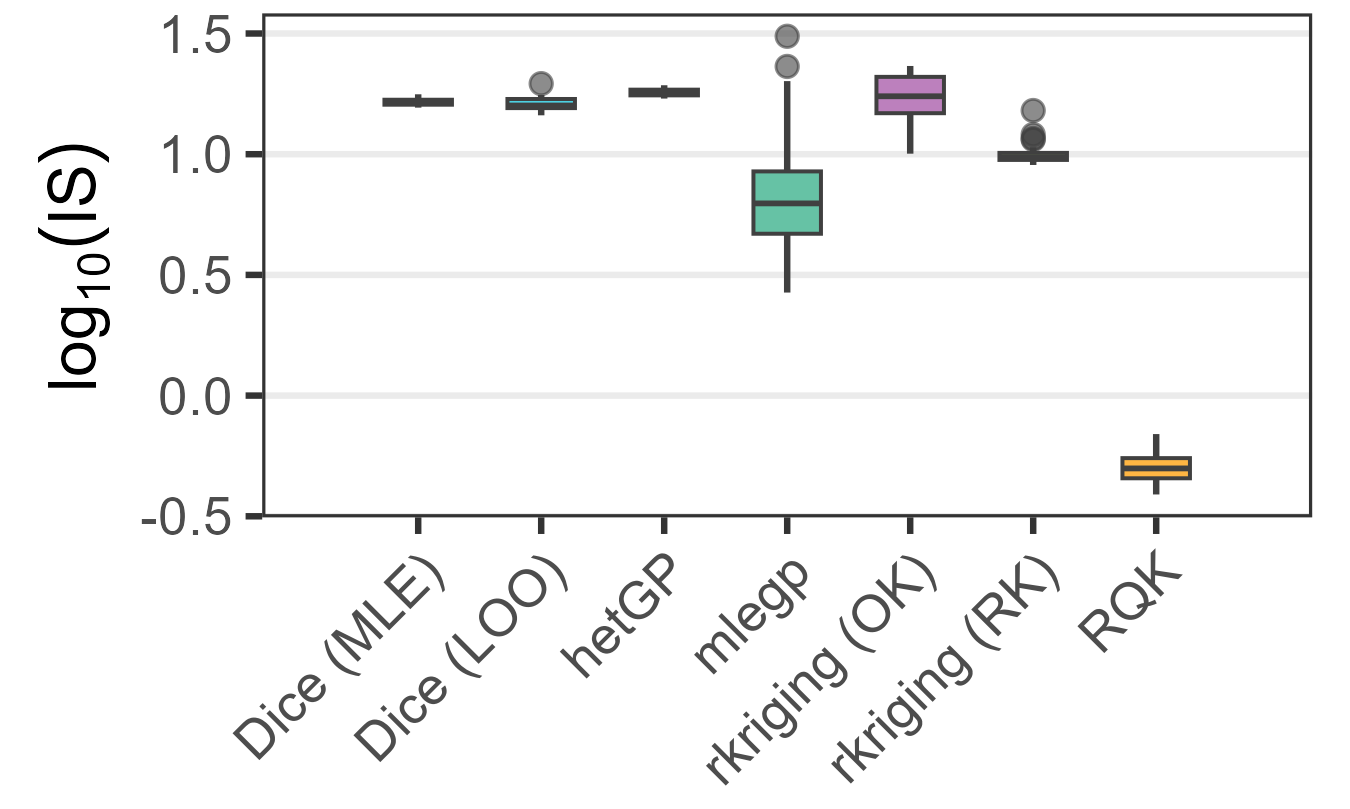}
\caption{Borehole function. Left: 8 dimensions. Right: 20 dimensions.}
\end{subfigure}%

\begin{subfigure}{\textwidth}
\centering
\includegraphics[width=0.4\textwidth]{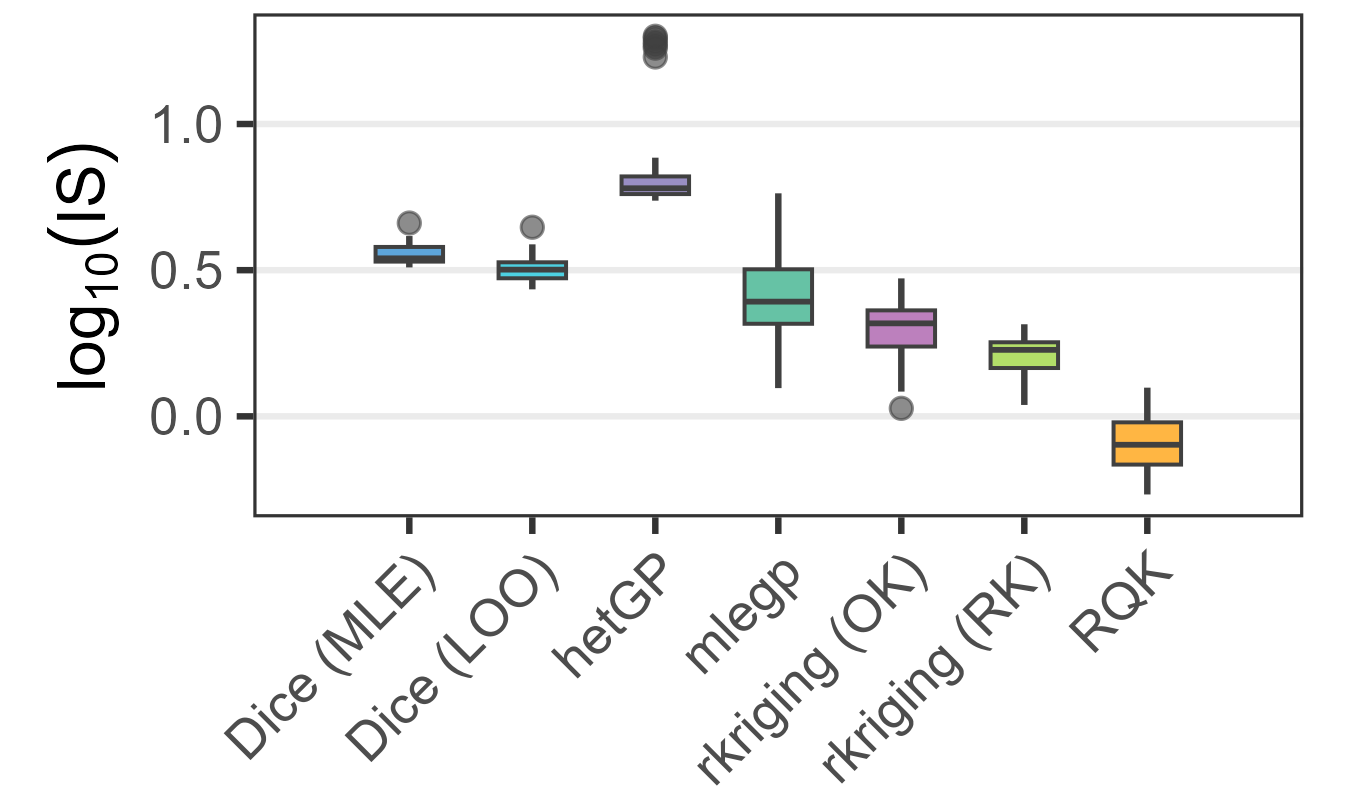}
\includegraphics[width=0.4\textwidth]{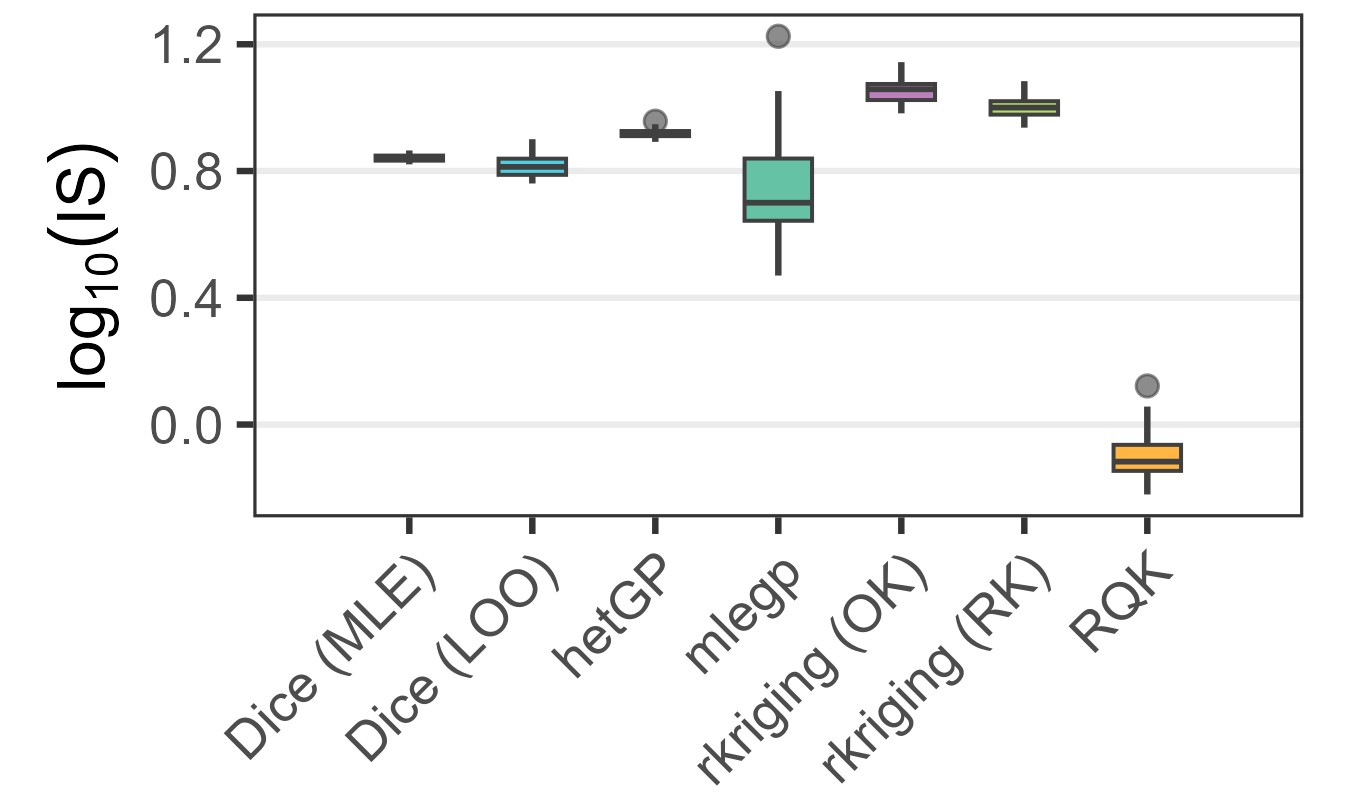}
\caption{Dette and Pepelyshev 8d function. Left: 8 dimensions. Right: 20 dimensions.}
\end{subfigure}%

\begin{subfigure}{\textwidth}
\centering
\includegraphics[width=0.4\textwidth]{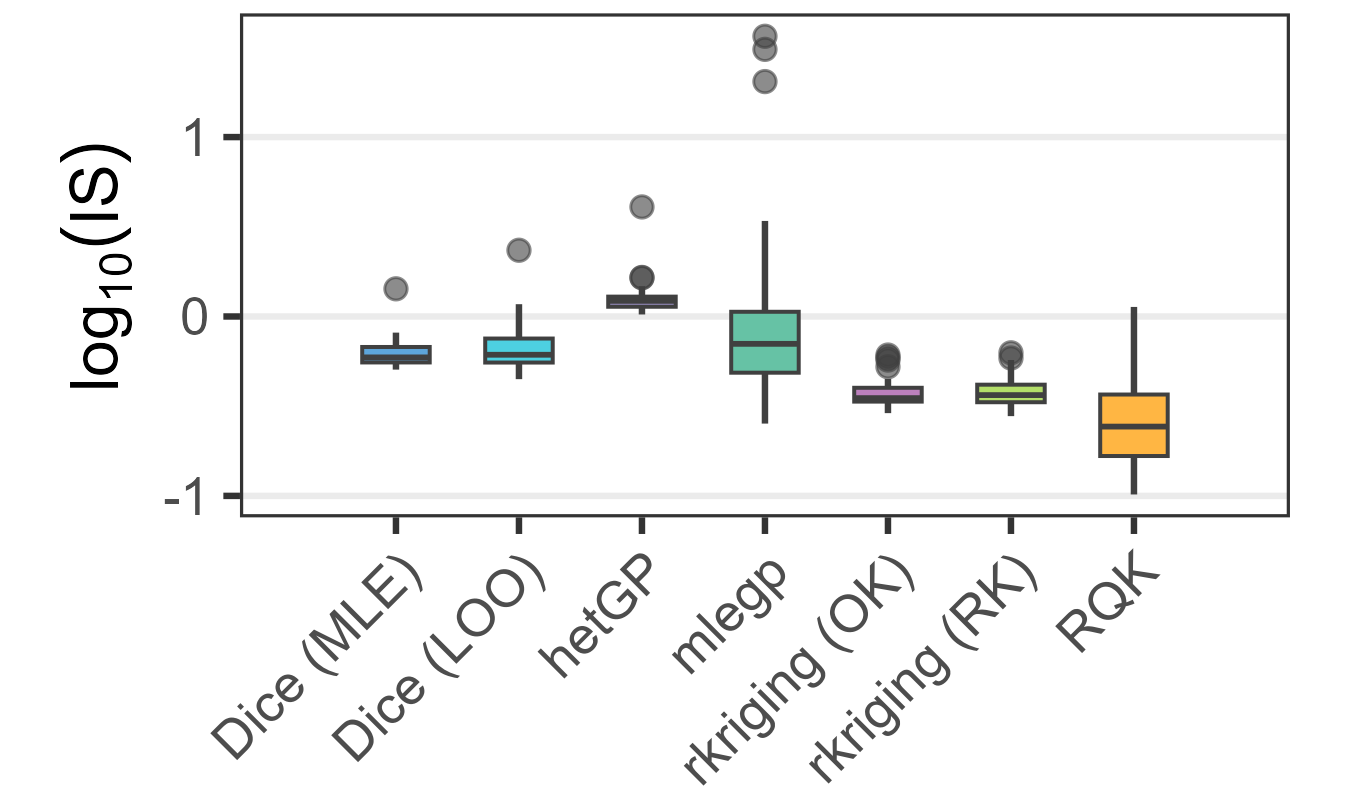}
\includegraphics[width=0.4\textwidth]{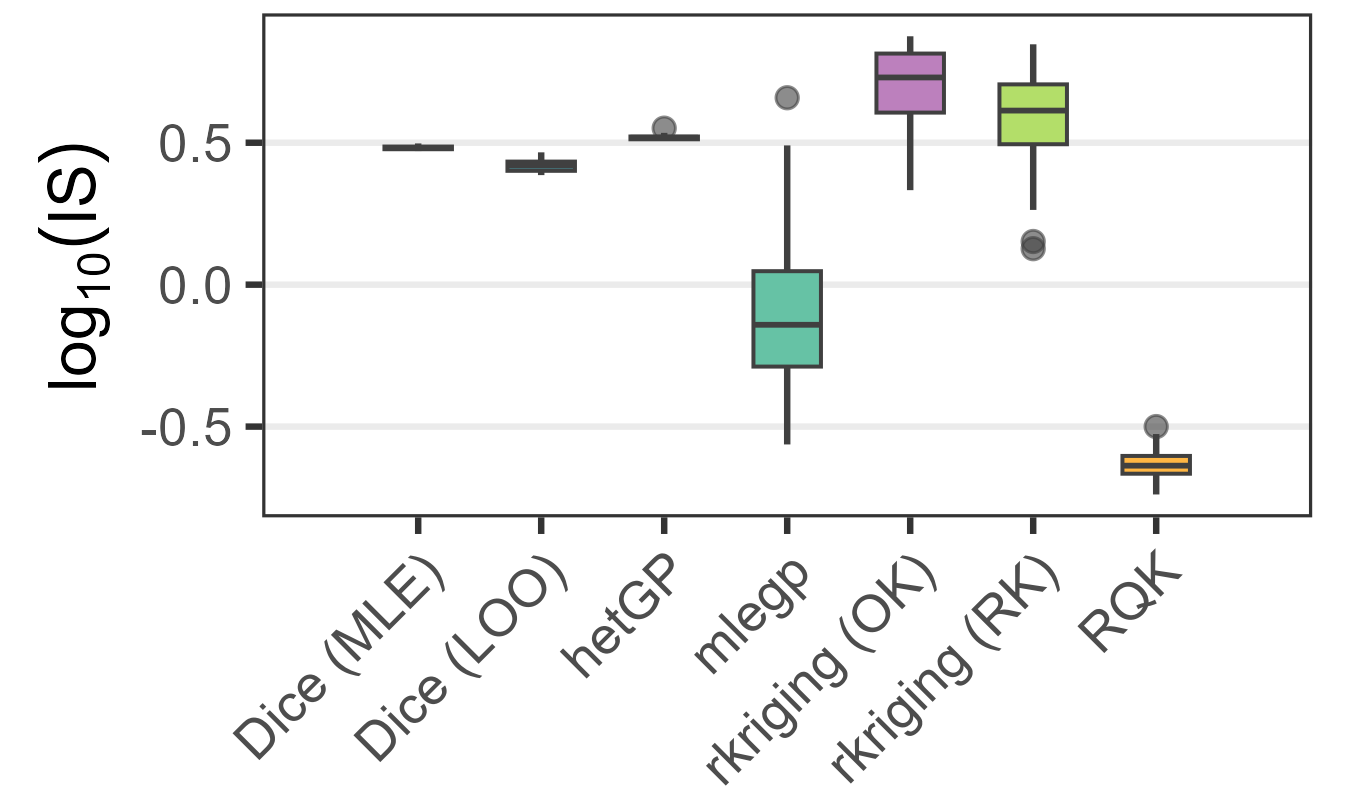}
\caption{Friedman function. Left: 5 dimensions. Right: 20 dimensions.}
\end{subfigure}%

\begin{subfigure}{\textwidth}
\centering
\includegraphics[width=0.4\textwidth]{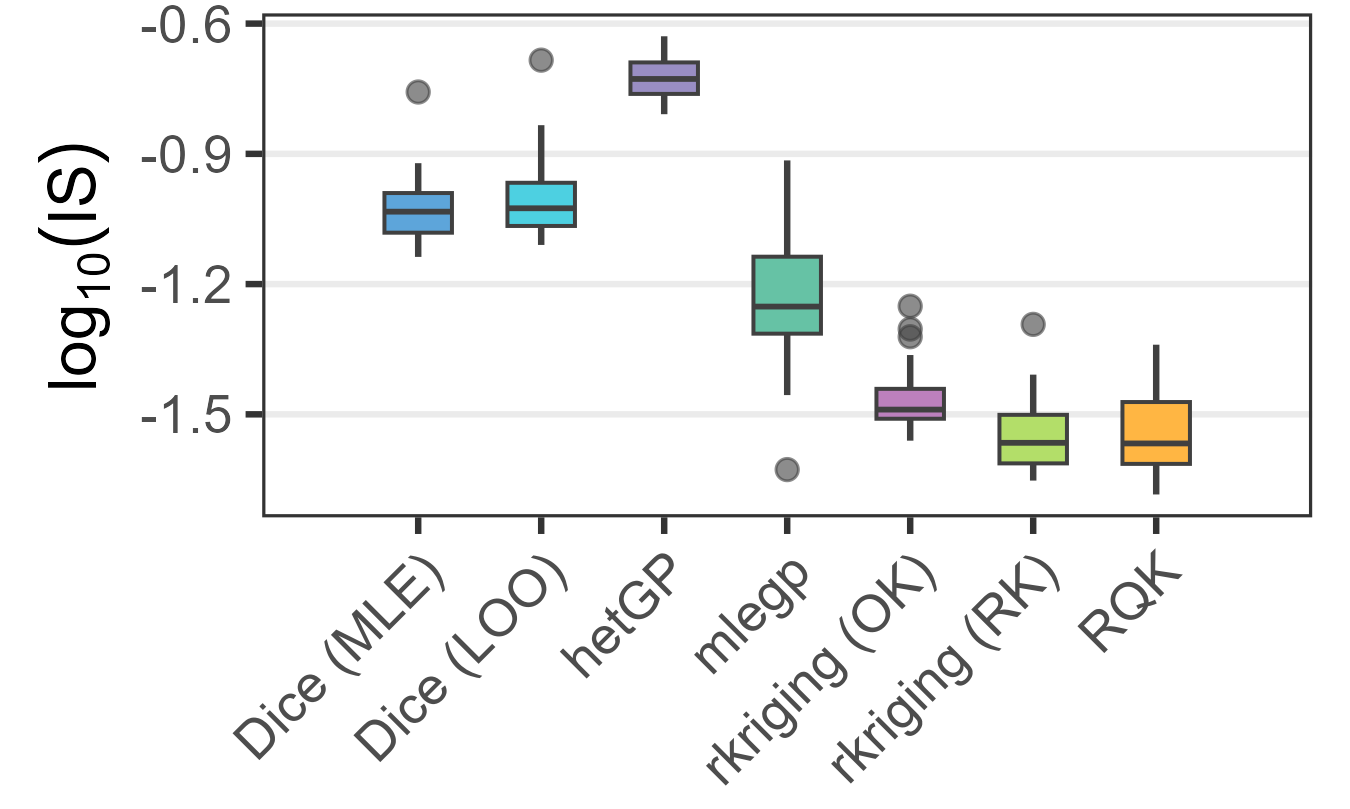}
\includegraphics[width=0.4\textwidth]{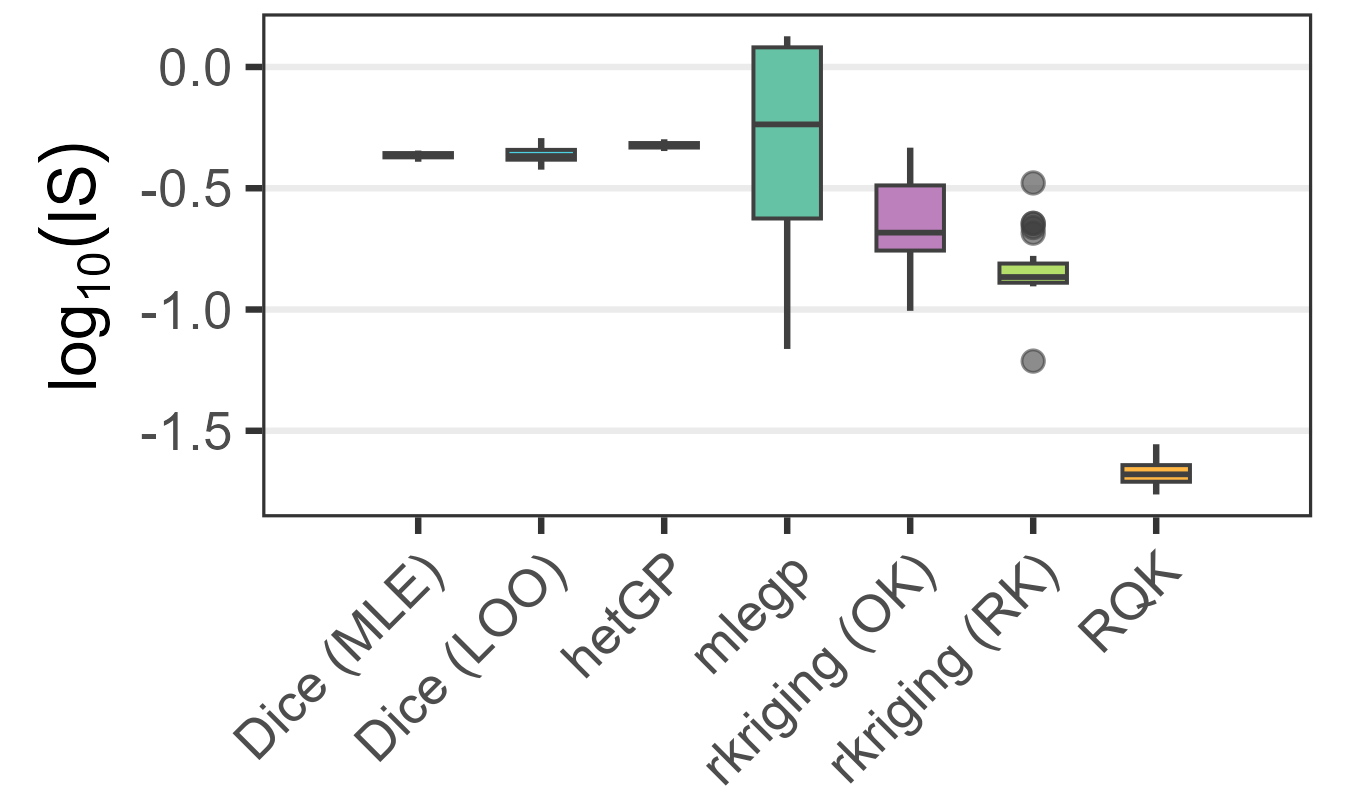}
\caption{OTL circuit function. Left: 6 dimensions. Right: 20 dimensions.}
\end{subfigure}%

\caption{Logarithm IS plots for deterministic functions with $n = 20d$. The original functions are on the left, and functions with augmented inert variable are on the right. Each experiment is repeated over 30 random seeds for data generation.}
\label{fig:noiseless_IS_20d}

\end{figure}

\begin{figure}[p]
\centering

\begin{subfigure}{\textwidth}
\centering
\includegraphics[width=0.4\textwidth]{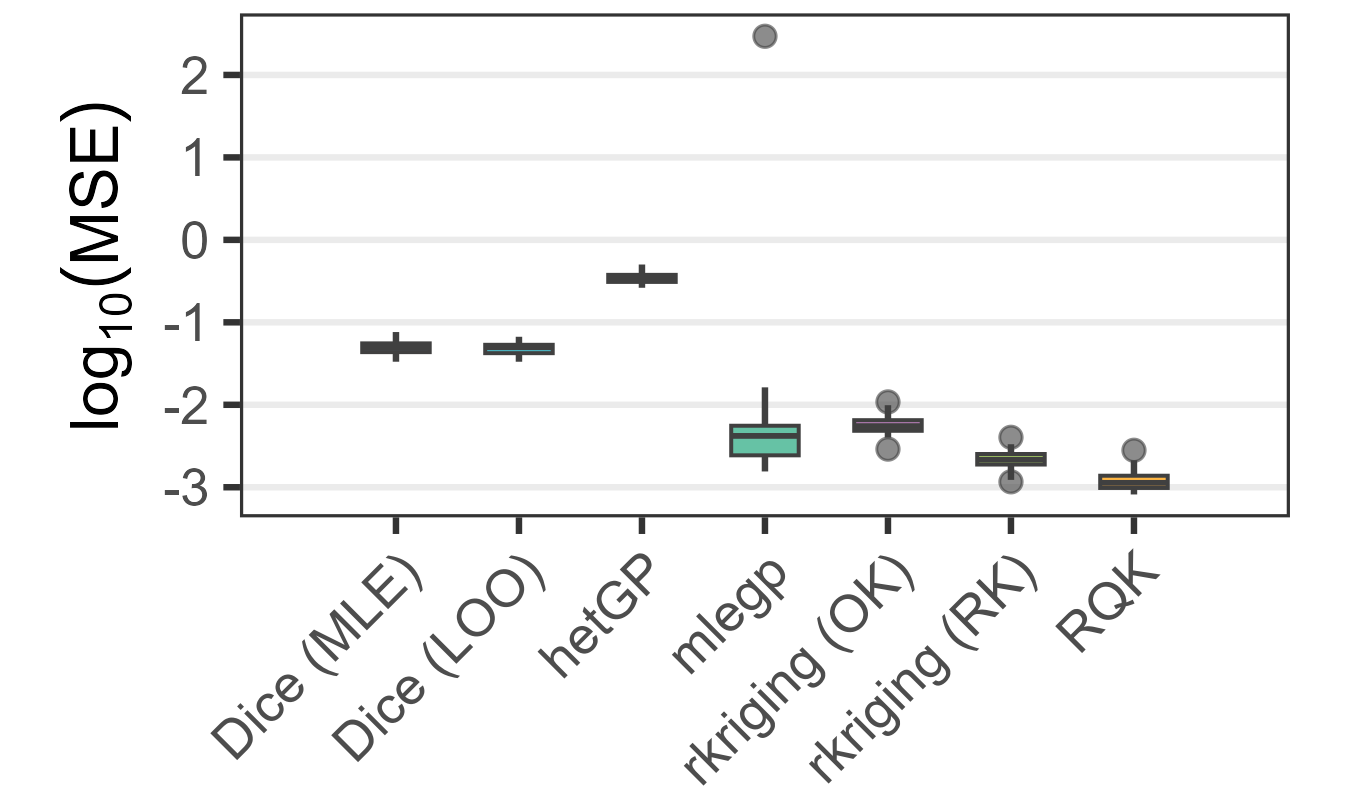}
\includegraphics[width=0.4\textwidth]{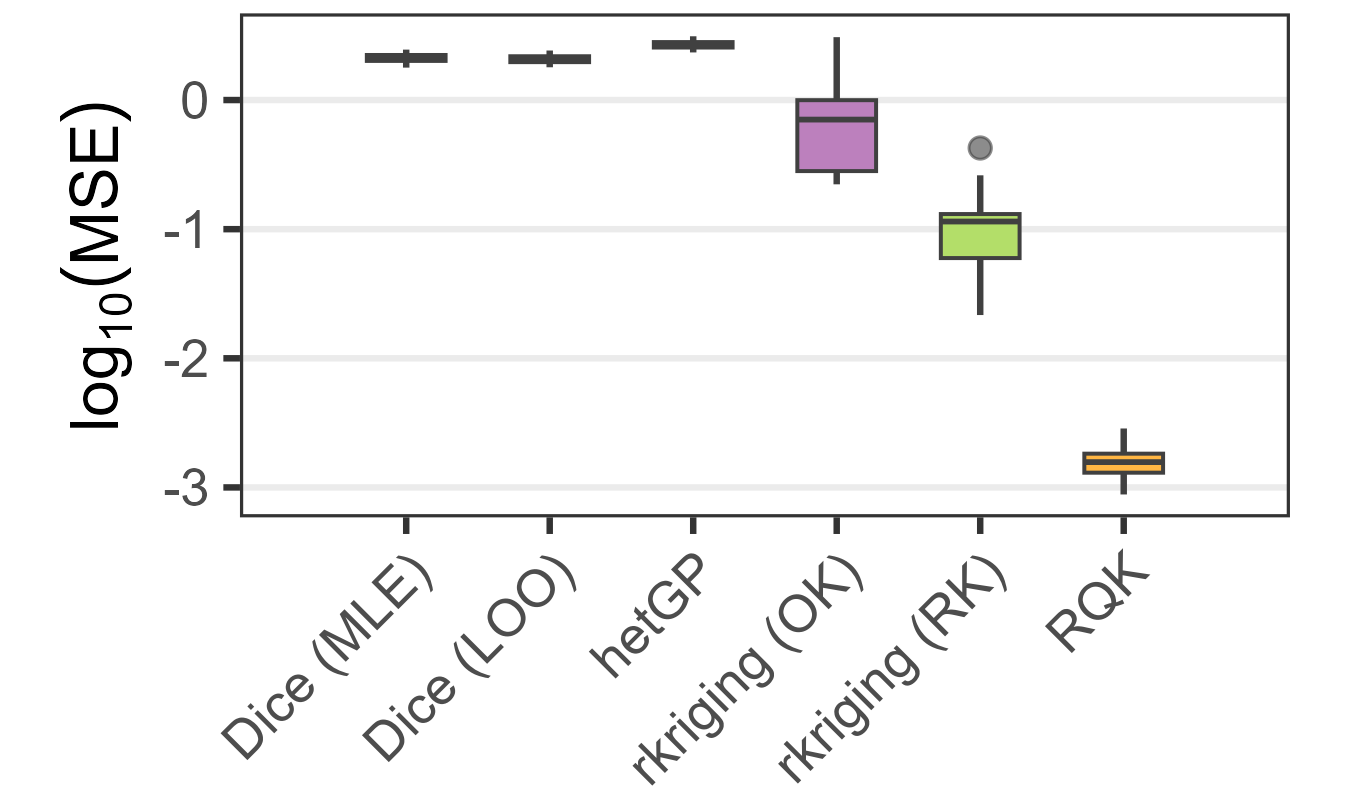}
\caption{Borehole function. Left: 8 dimensions. Right: 20 dimensions.}
\end{subfigure}%

\begin{subfigure}{\textwidth}
\centering
\includegraphics[width=0.4\textwidth]{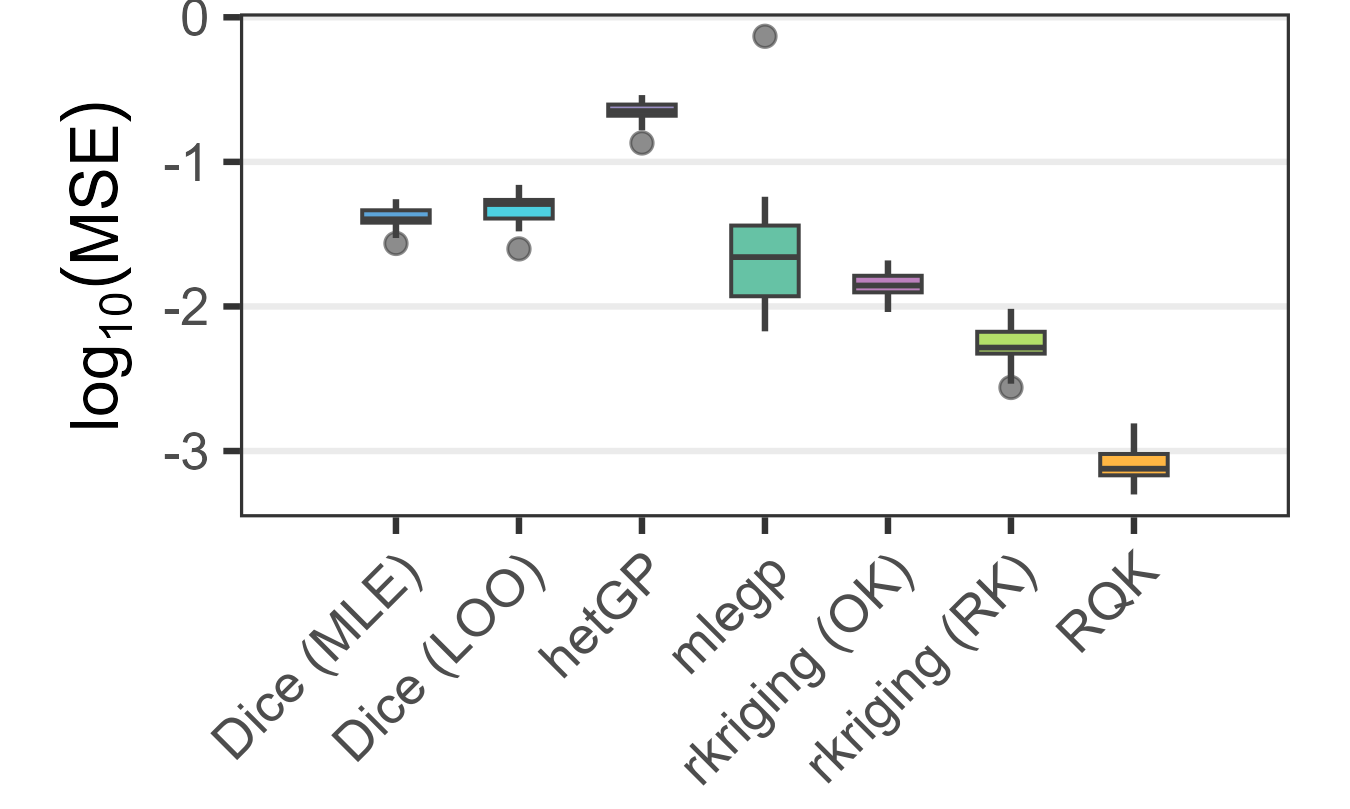}
\includegraphics[width=0.4\textwidth]{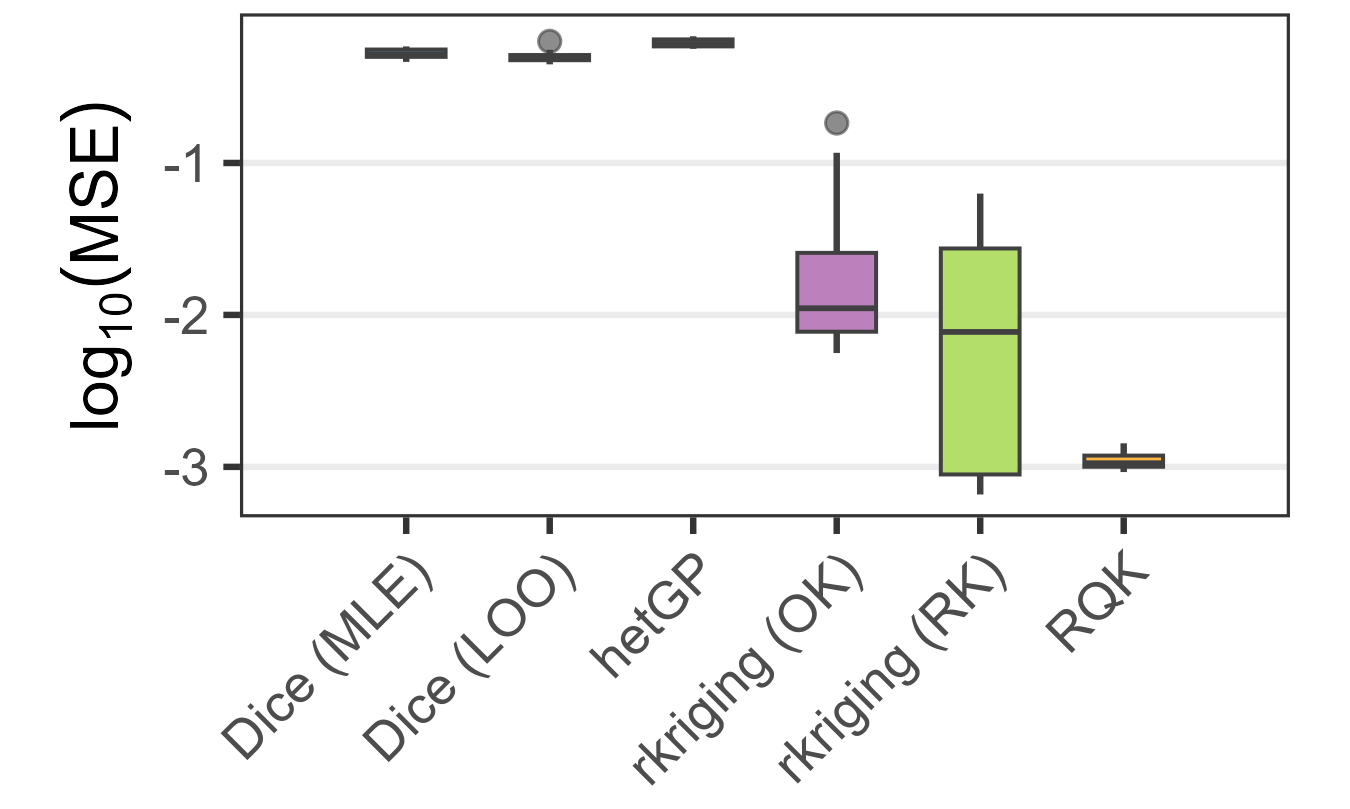}
\caption{Dette and Pepelyshev 8d function. Left: 8 dimensions. Right: 20 dimensions.}
\end{subfigure}%

\begin{subfigure}{\textwidth}
\centering
\includegraphics[width=0.4\textwidth]{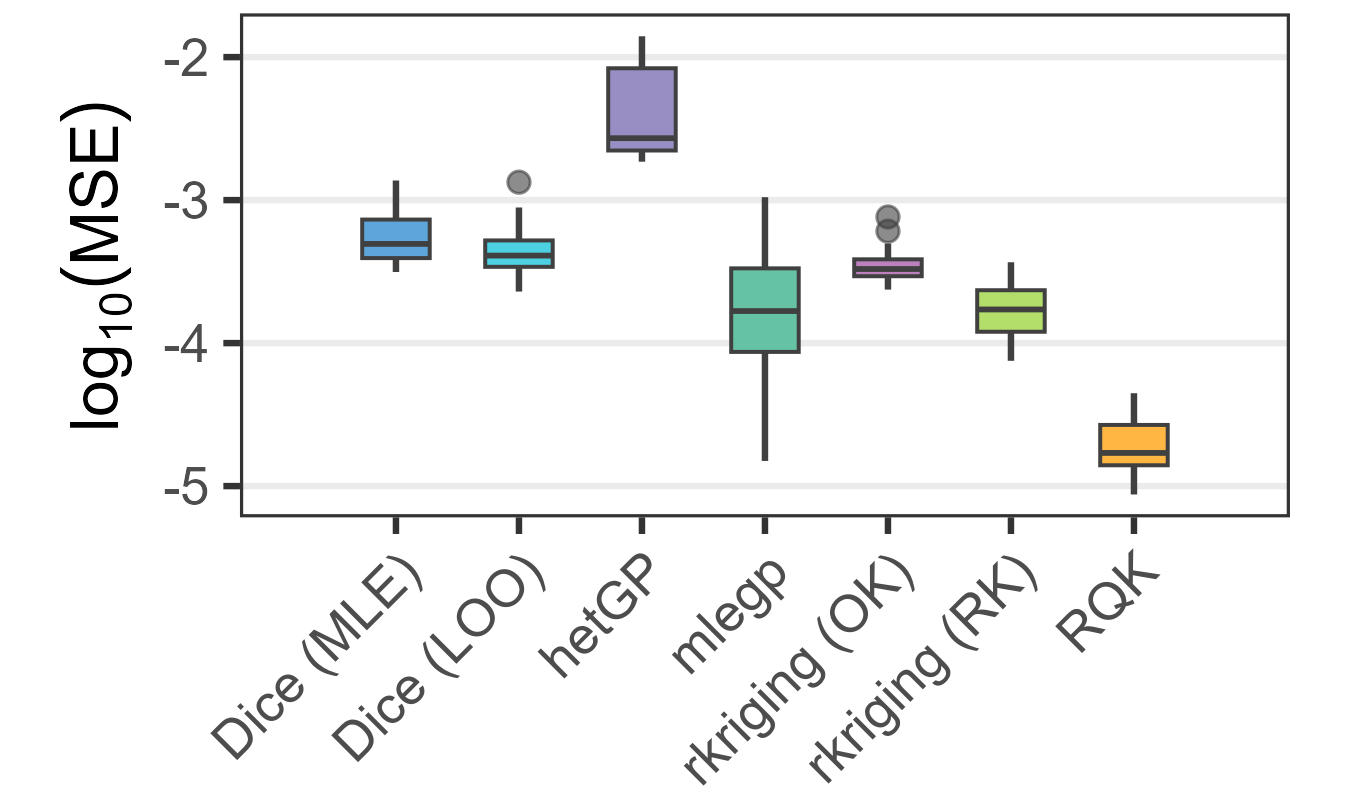}
\includegraphics[width=0.4\textwidth]{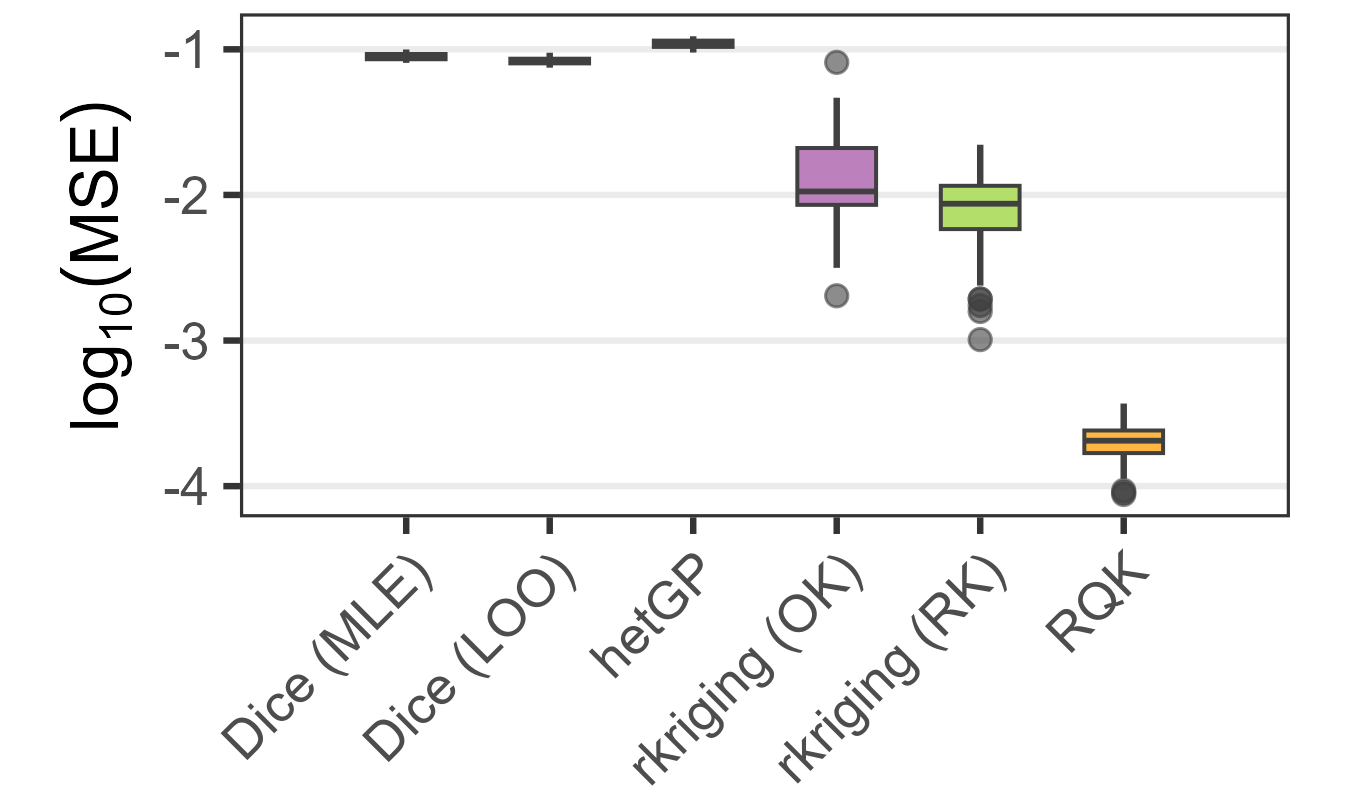}
\caption{Friedman function. Left: 5 dimensions. Right: 20 dimensions.}
\end{subfigure}%

\begin{subfigure}{\textwidth}
\centering
\includegraphics[width=0.4\textwidth]{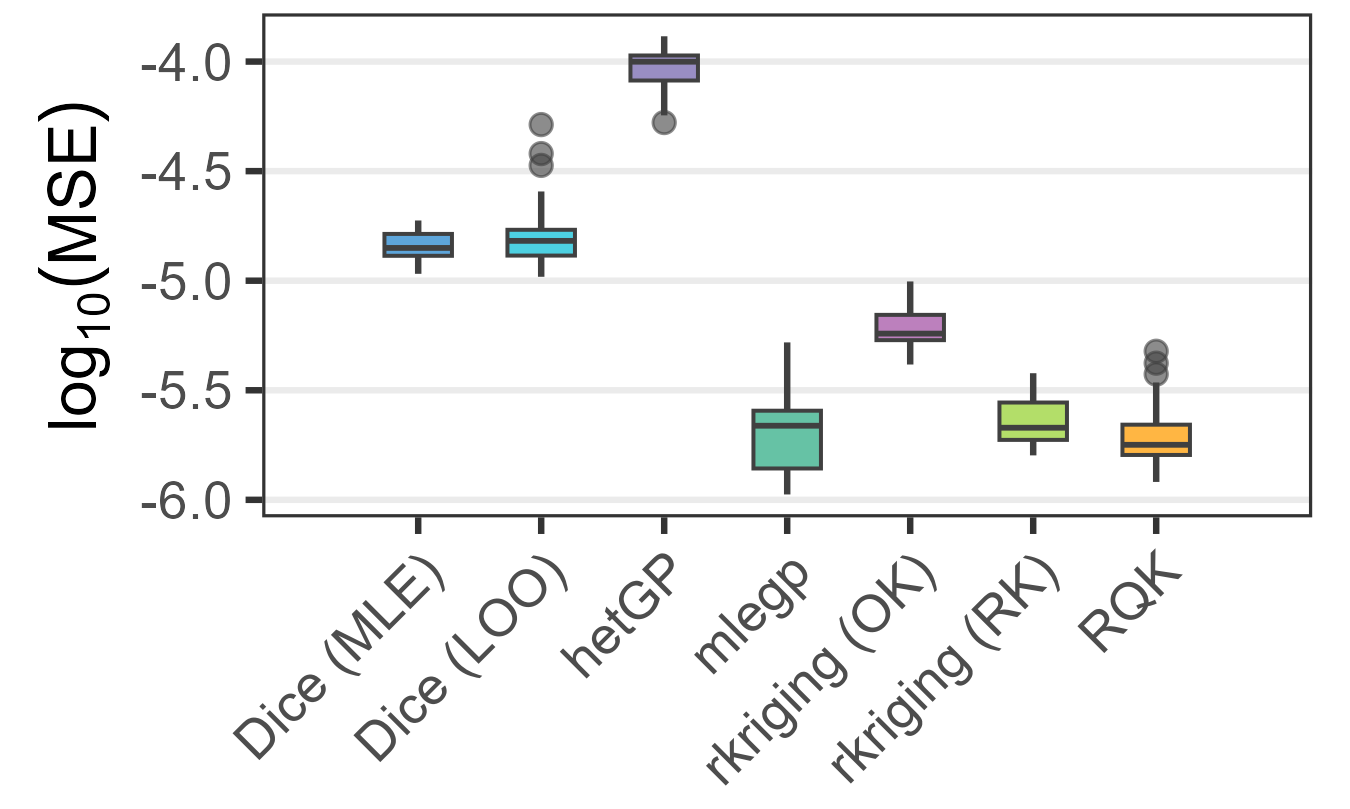}
\includegraphics[width=0.4\textwidth]{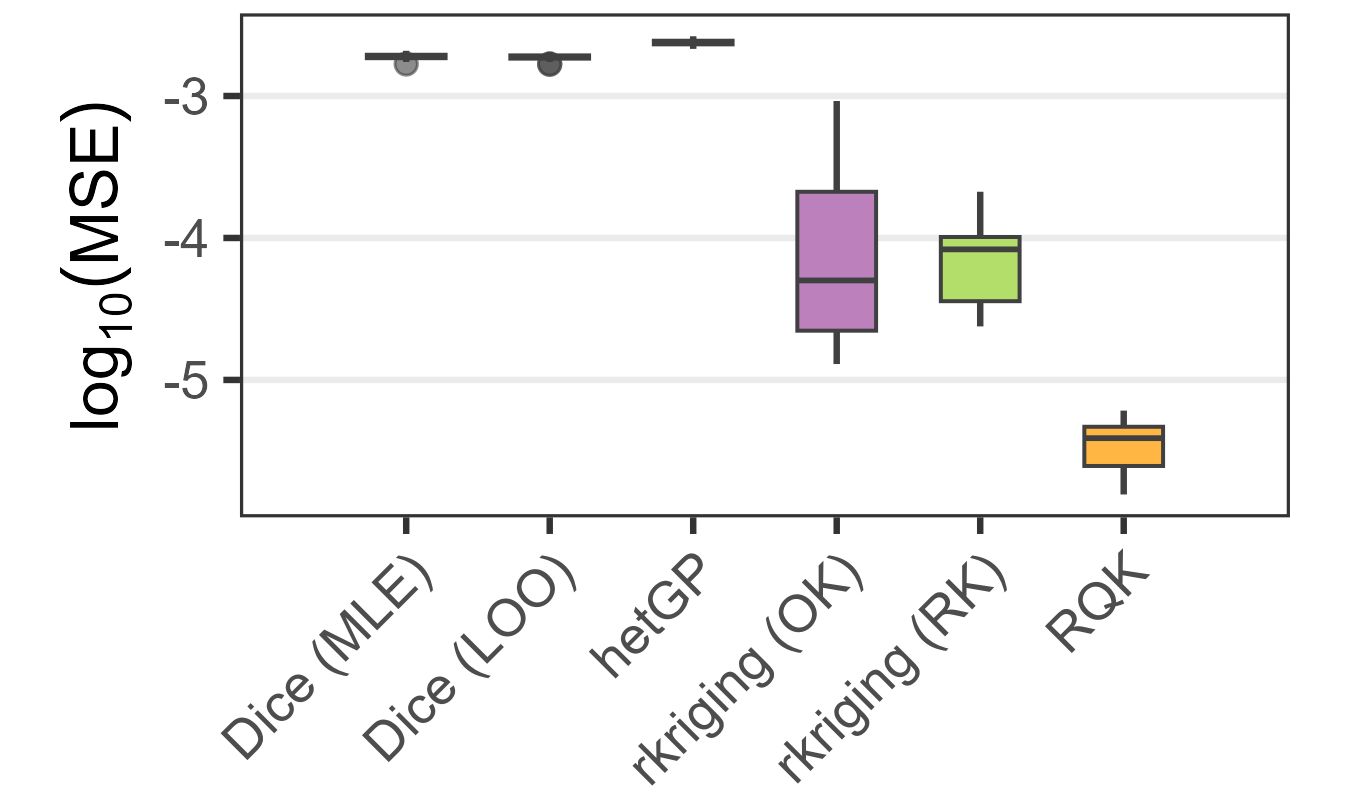}
\caption{OTL circuit function. Left: 6 dimensions. Right: 20 dimensions.}
\end{subfigure}%

\caption{Logarithm MSE plots for deterministic functions with $n = 60d$. The original functions are on the left, and functions with augmented inert variable are on the right. Each experiment is repeated over 30 random seeds for data generation.}
\label{fig:noiseless_60d}

\end{figure}

\begin{figure}[p]
\centering

\begin{subfigure}{\textwidth}
\centering
\includegraphics[width=0.4\textwidth]{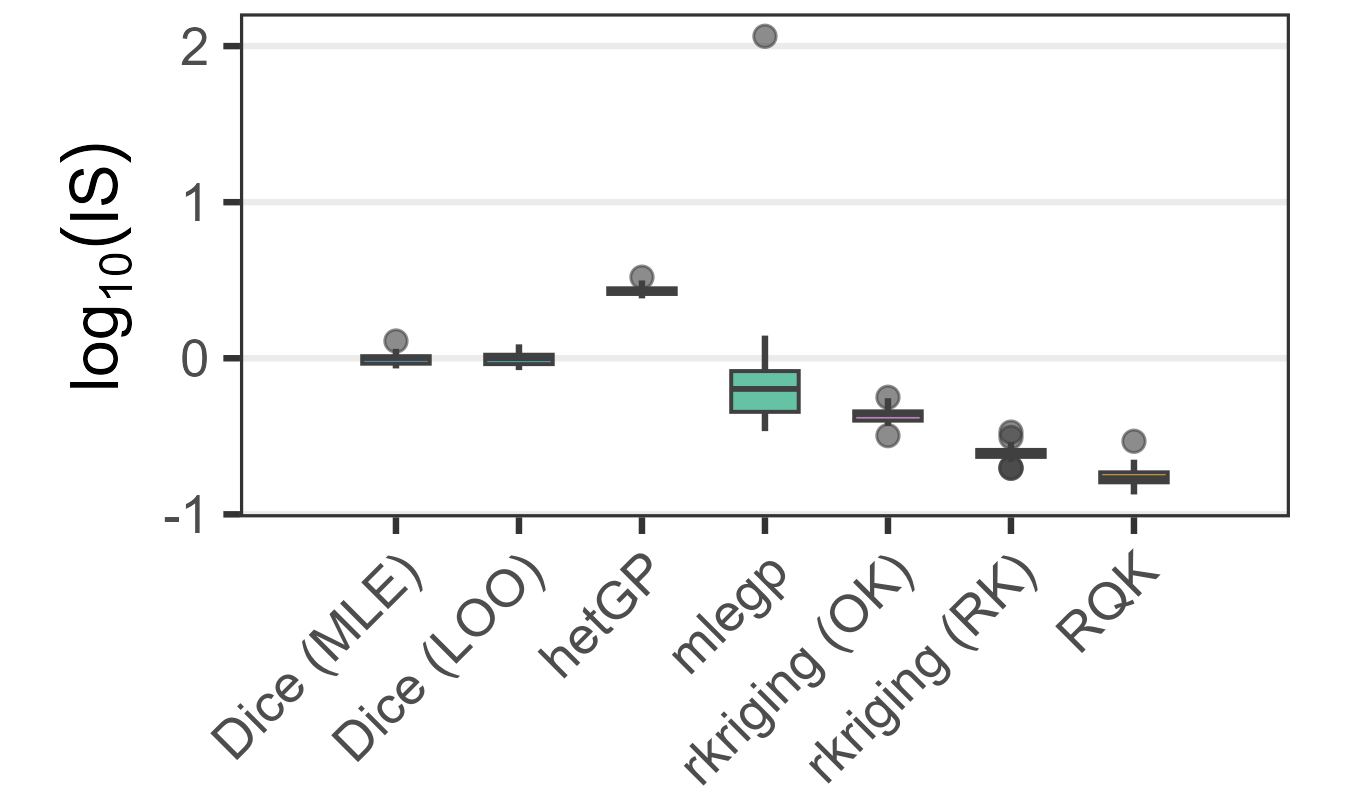}
\includegraphics[width=0.4\textwidth]{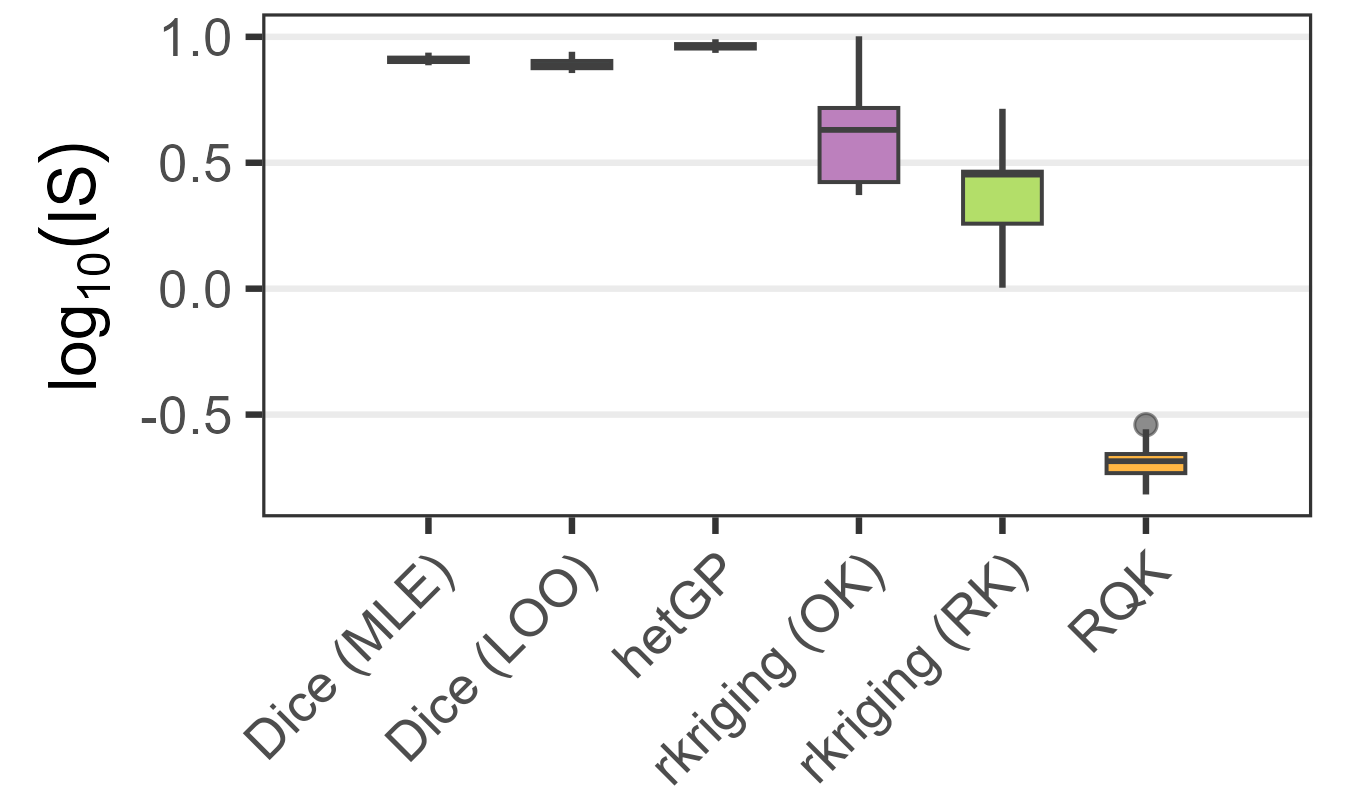}
\caption{Borehole function. Left: 8 dimensions. Right: 20 dimensions.}
\end{subfigure}%

\begin{subfigure}{\textwidth}
\centering
\includegraphics[width=0.4\textwidth]{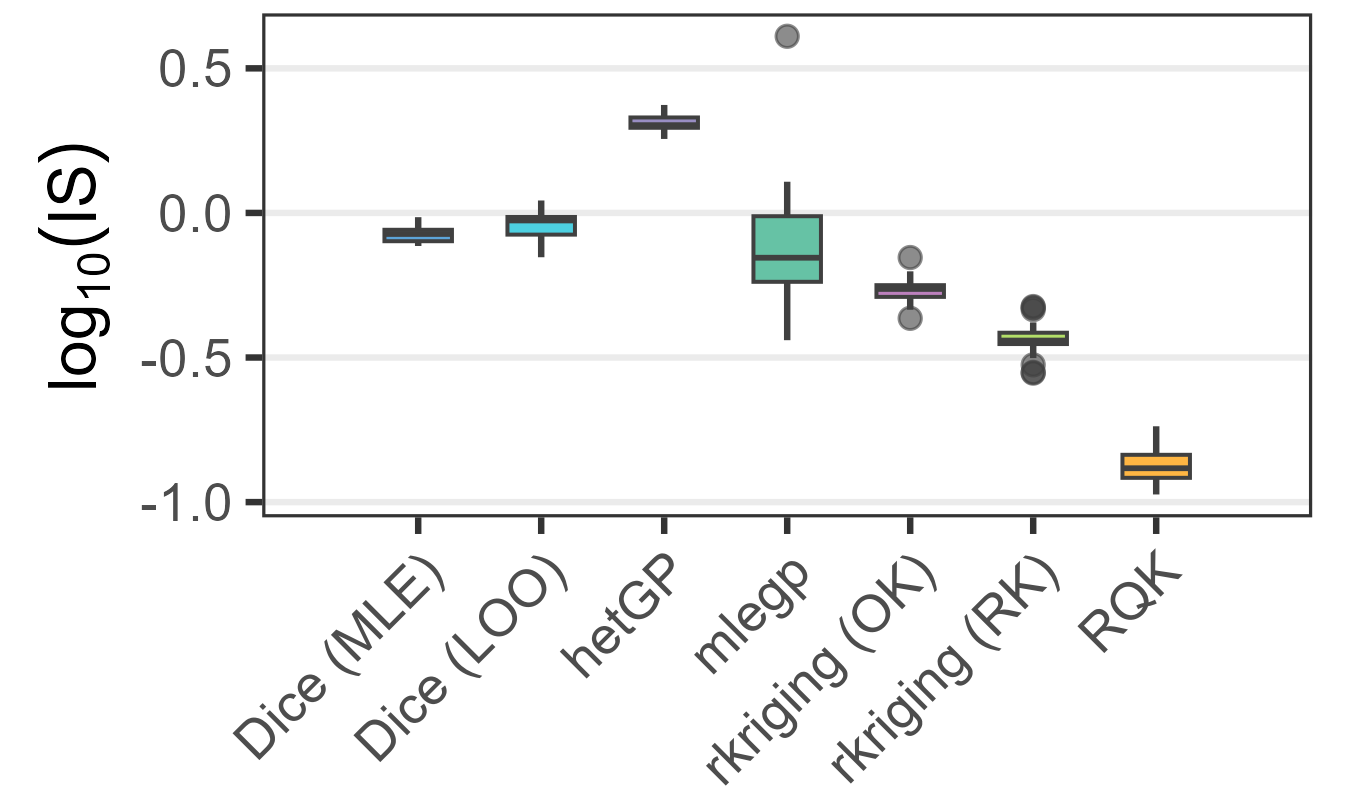}
\includegraphics[width=0.4\textwidth]{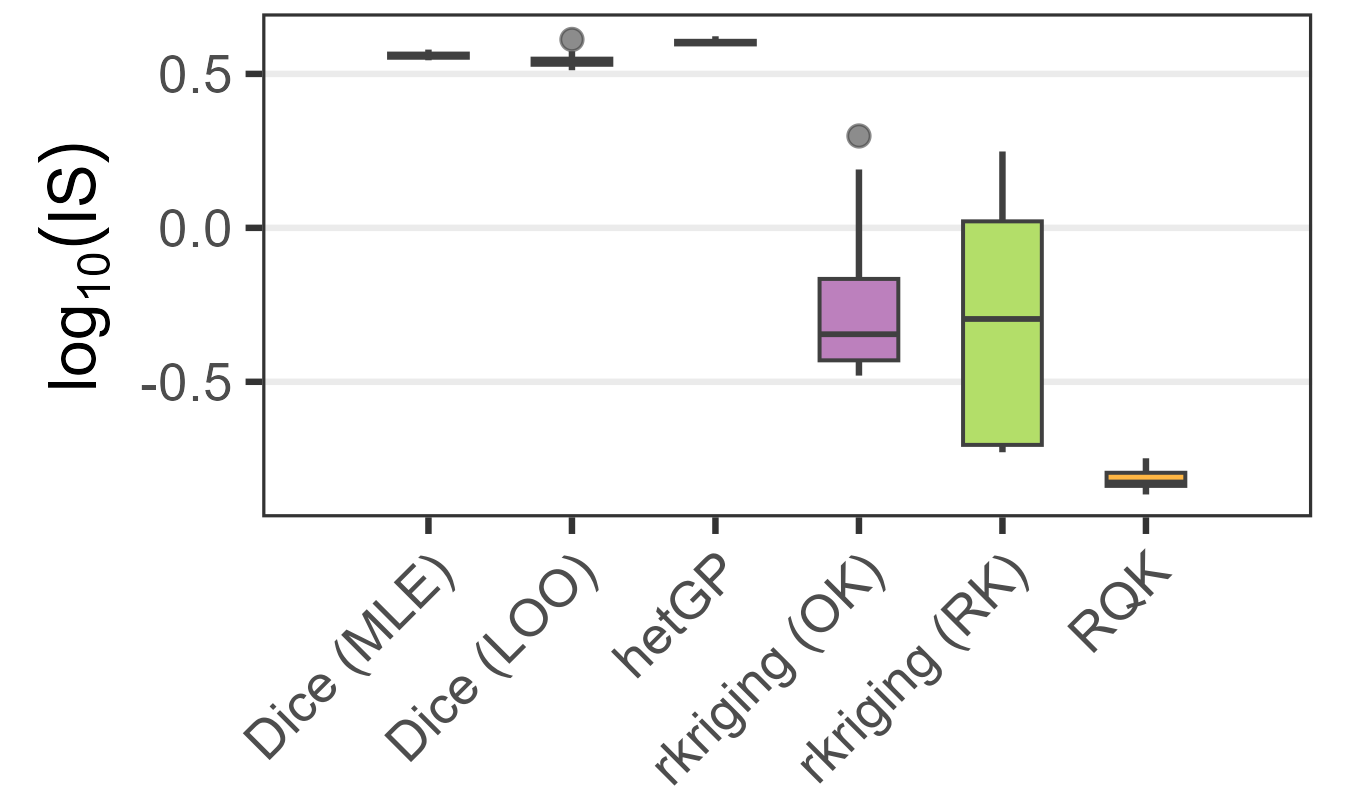}
\caption{Dette and Pepelyshev 8d function. Left: 8 dimensions. Right: 20 dimensions.}
\end{subfigure}%

\begin{subfigure}{\textwidth}
\centering
\includegraphics[width=0.4\textwidth]{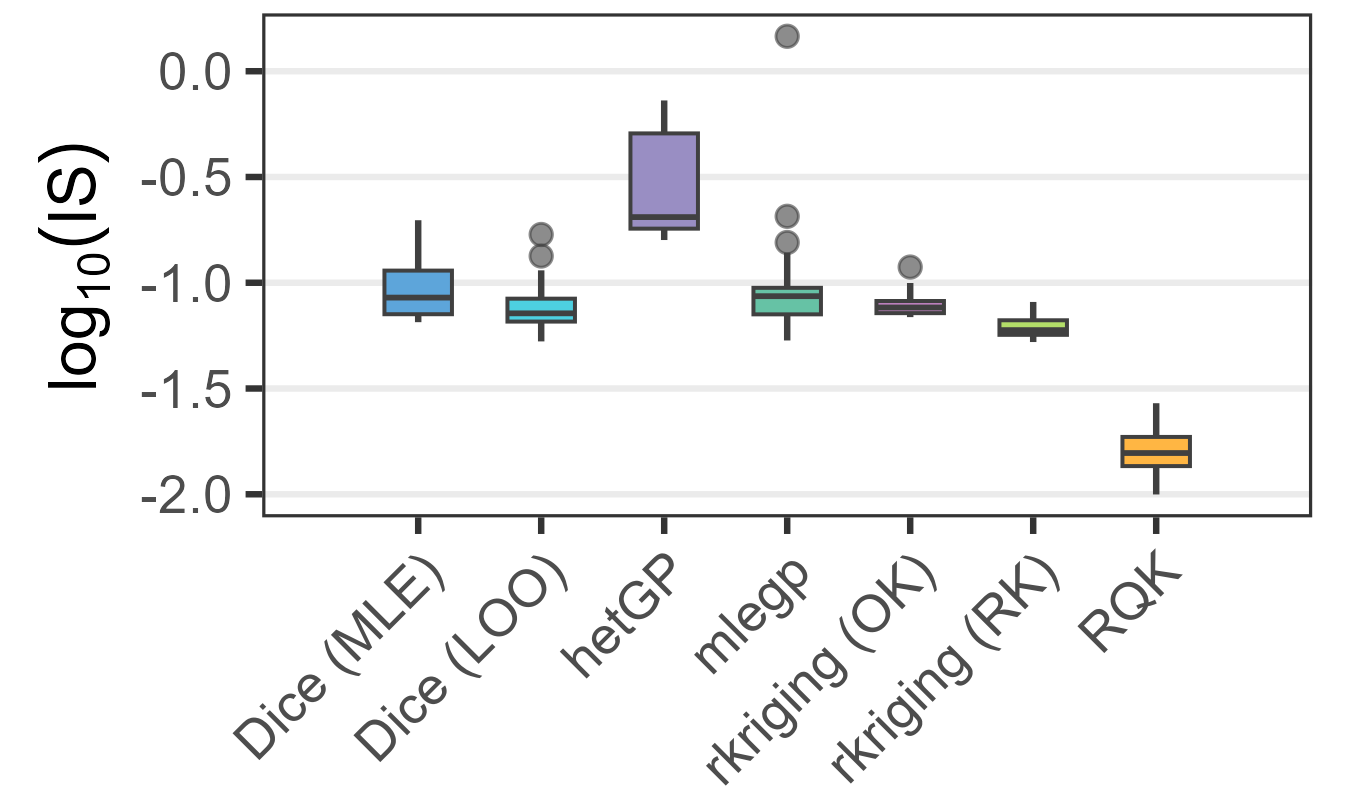}
\includegraphics[width=0.4\textwidth]{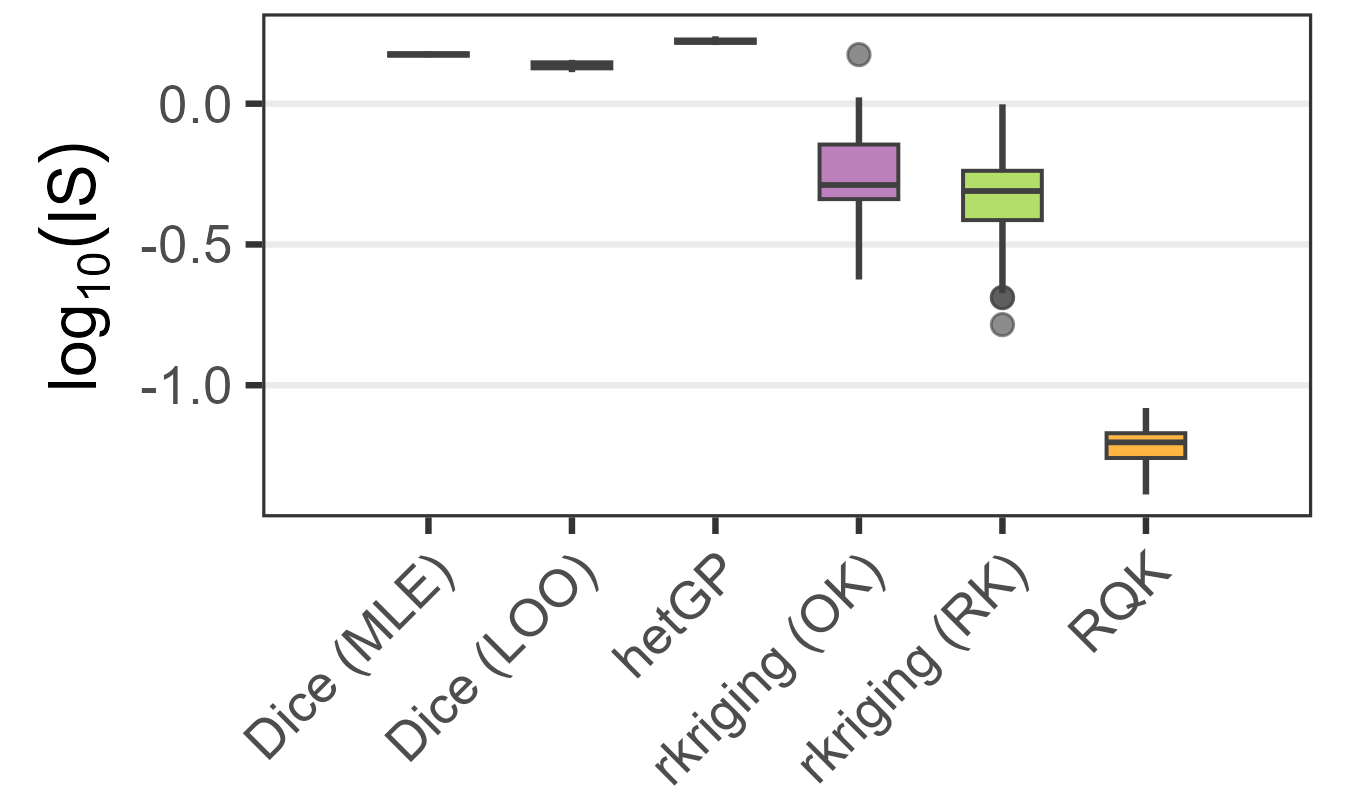}
\caption{Friedman function. Left: 5 dimensions. Right: 20 dimensions.}
\end{subfigure}%

\begin{subfigure}{\textwidth}
\centering
\includegraphics[width=0.4\textwidth]{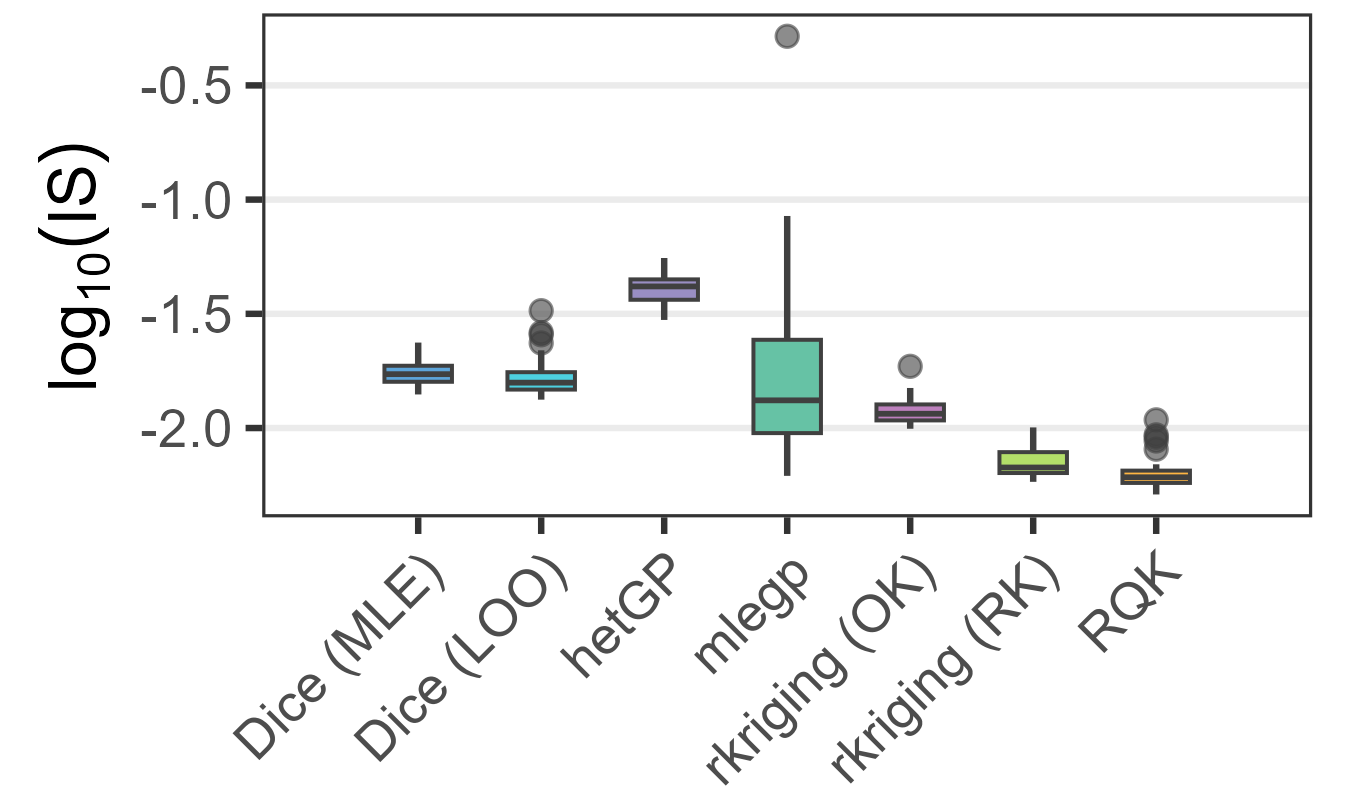}
\includegraphics[width=0.4\textwidth]{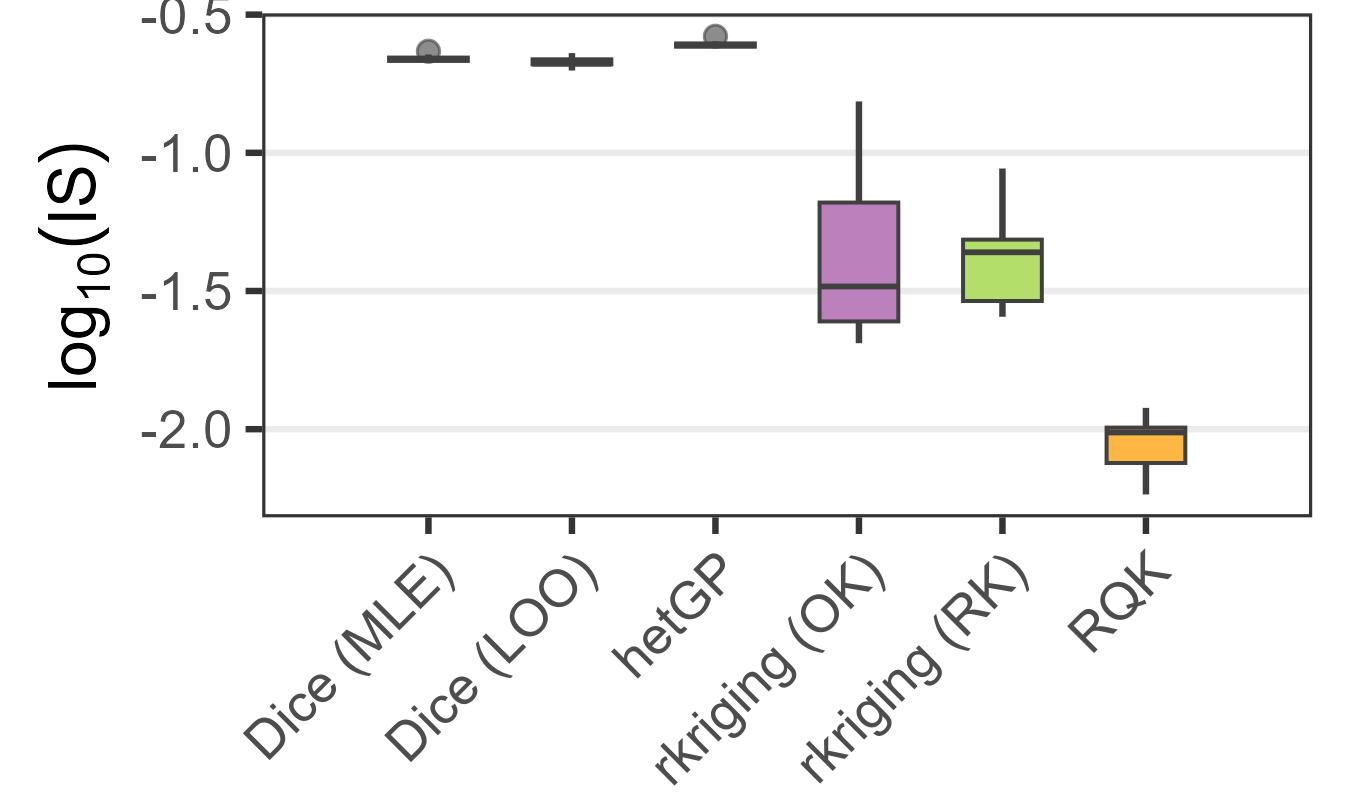}
\caption{OTL circuit function. Left: 6 dimensions. Right: 20 dimensions.}
\end{subfigure}%

\caption{Logarithm IS plots for deterministic functions with $n = 60d$. The original functions are on the left, and functions with augmented inert variable are on the right. Each experiment is repeated over 30 random seeds for data generation.}
\label{fig:noiseless_IS_60d}

\end{figure}

\vspace{.25in}


\FloatBarrier
\bibliography{bib}
\pagebreak

\end{document}